\documentclass[12pt]{article}

\usepackage{amsmath, amsfonts, amssymb, amsthm, fullpage, color, bbm, enumerate, titlesec, graphicx, epstopdf, multirow, array, mathtools}
\usepackage{subcaption}
\usepackage[T1]{fontenc}
\usepackage{lmodern}
\usepackage{pdflscape}
\usepackage{booktabs}
\usepackage{longtable}
\usepackage[title]{appendix}

\usepackage[onehalfspacing]{setspace}

\definecolor{darkred}{rgb}{0.5,0,0}

\titleformat*{\paragraph}{\sc}

\usepackage{hyperref}
\hypersetup{allbordercolors={1 1 1},allcolors=darkred,colorlinks=true}
\usepackage[capitalize,nameinlink,noabbrev]{cleveref}

\usepackage[font={small}]{caption}

\usepackage[round]{natbib}
\defcitealias{Chernozhukov2018}{Chernozhukov et al.\ (2018)}
\defcitealias{Chernozhukov2022}{Chernozhukov et al.\ (2022)}

\usepackage[runin]{abstract}
\renewcommand{\abstractname}{Abstract:}

\newcolumntype{C}[1]{>{\centering\arraybackslash}p{#1}}

\theoremstyle{plain}
\newtheorem{asn}{Assumption}
\crefname{asn}{Assumption}{Assumptions}

\crefname{lem}{Lemma}{Lemmas}

\newtheorem{prop}{Proposition}
\crefname{prop}{Proposition}{Propositions}
\newtheorem{cor}{Corollary}
\crefname{cor}{Corollary}{Corollaries}

\theoremstyle{definition}

\numberwithin{equation}{section}
\numberwithin{figure}{section}
\numberwithin{table}{section}
\numberwithin{lem}{section}
\numberwithin{prop}{section}
\numberwithin{cor}{section}
\numberwithin{asn}{section}

\newcommand\independent{\protect\mathpalette{\protect\independenT}{\perp}}
\def\independenT#1#2{\mathrel{\rlap{$#1#2$}\mkern2mu{#1#2}}}

\DeclareMathOperator*{\argmin}{argmin}
\DeclareMathOperator*{\tr}{trace}
\DeclareMathOperator*{\ve}{vec}

\DeclareMathOperator*{\diag}{diag}

\DeclareMathOperator*{\var}{Var}
\DeclareMathOperator*{\cov}{Cov}

\DeclareMathOperator*{\plim}{plim}
\DeclareMathOperator{\avar}{aVar}
\DeclareMathOperator{\acov}{aCov}
\DeclareMathOperator{\abias}{aBias}
\DeclareMathOperator{\aMSE}{aMSE}
\DeclareMathOperator{\sspan}{span}

\DeclareMathOperator{\ci}{CI}
\DeclareMathOperator{\cv}{cv}

\DeclareMathOperator{\ce}{CE}

\newcommand{\alert}[1]{#1}

\allowdisplaybreaks

\crefname{sappsec}{Supplemental Appendix}{Supplemental Appendices}
\crefname{sappsubsec}{Supplemental Appendix}{Supplemental Appendices}
\crefname{sappsubsubsec}{Supplemental Appendix}{Supplemental Appendices}
\crefname{appsec}{Appendix}{Appendices}

\begin{document}

\title{\texorpdfstring{\vspace{-1.5\baselineskip}}{} Double Robustness of Local Projections \texorpdfstring{\\}{}and Some Unpleasant VARithmetic\thanks{Email: {\tt montiel.olea@gmail.com}, {\tt mikkelpm@uchicago.edu}, {\tt ericqian@princeton.edu}, and {\tt ckwolf@mit.edu}. We received helpful comments from three anonymous referees, Isaiah Andrews, Tim Armstrong, Joel Flynn, Jim Hamilton, Lars Peter Hansen, \`{O}scar Jord\`{a}, Lutz Kilian, Michal Koles\'{a}r, Ulrich M\"{u}ller, Pablo Ottonello, Frank Schorfheide, Harald Uhlig, Mark Watson, Ke-Li Xu, and seminar participants at several venues. Nelson Matthew Tan provided excellent research assistance. Plagborg-M{\o}ller acknowledges that this material is based upon work supported by the NSF under Grant {\#}2238049 and by the Alfred P.\ Sloan Foundation, and Wolf does the same for NSF Grant {\#}2314736.}}
\author{\begin{tabular}{ccc}
Jos\'{e} Luis Montiel Olea && Mikkel Plagborg-M{\o}ller \\
{\small Cornell University} && {\small University of Chicago} \\[1ex]
Eric Qian && Christian K. Wolf \\
{\small Princeton University} && {\small MIT \& NBER}
\end{tabular}}
\date{\texorpdfstring{\bigskip}{ }January 12, 2026}
\maketitle

\vspace{-1.5em}
\begin{abstract}
We consider impulse response inference in a locally misspecified vector autoregression (VAR) model. The conventional local projection (LP) confidence interval has correct coverage even when the misspecification is so large that it can be detected with probability approaching 1. This result follows from a ``double robustness'' property analogous to that of popular partially linear regression estimators. By contrast, the conventional VAR confidence interval with short-to-moderate lag length can severely undercover for misspecification that is small, difficult to detect statistically, \alert{and cannot be ruled out based on economic theory}. The VAR confidence interval has robust coverage \alert{if, and only if,} the lag length is so large that the interval is as wide as the LP interval.
\end{abstract}

\clearpage

\section{Introduction}
\label{sec:intro}

In recent years, local projection (LP) estimators of impulse response functions have become a very popular alternative to structural vector autoregressions \citep[henceforth interchangeably referred to as VAR or SVAR,][]{Sims1980}. In addition to their simplicity, one potential explanation for the popularity of LPs is their perceived robustness to misspecification, as claimed by \citet{Jorda2005} in his seminal article that proposed the estimation method:

\begin{quote}

``\emph{[T]hese projections are local to each forecast horizon and therefore more robust [than VARs] to misspecification of the unknown DGP.}''

\end{quote}
While this sentiment has been echoed in influential reviews \citep[e.g.,][]{Ramey2016,Nakamura2018,Jorda2023}, there so far exist essentially no theoretical results on the relative robustness of LP and VAR inference procedures to misspecification. \citet{Plagborg2021} and \citet{Xu2023} show that the two estimators are in fact asymptotically equivalent---and thus equally robust to misspecification---in a general VAR($\infty$) model if the estimation lag length diverges to infinity with the sample size. However, this result does not directly speak to the empirically relevant case where researchers employ small-to-moderate lag lengths to preserve degrees of freedom. Applied researchers must therefore base their choice of inference procedure on empirically calibrated simulation studies \citep{Kilian2011,Li2024}.

In this paper we provide a formal proof of \citeauthor{Jorda2005}'s claim that conventional LP confidence intervals for impulse responses are surprisingly robust to misspecification. On the other hand, VAR confidence intervals are robust if, \emph{and only if}, they are as wide as LP intervals asymptotically, as is the case when they control for a large number of lags. If the confidence interval is shorter, then it is \emph{necessarily} unreliable.

We consider a large class of stationary data generating processes (DGPs) that are well approximated by a finite-order SVAR model, but subject to local misspecification in the form of an asymptotically vanishing moving average (MA) process, of potentially infinite order. This class is consistent with essentially all linearized structural macroeconomic models and covers many types of dynamic misspecification, such as under-specification of the lag length, failure to include relevant control variables, inappropriate aggregation, and measurement error. Intuitively, with this set-up we capture the idea that finite-order VAR models provide a good but imperfect approximation of reality.

In this setting, we prove that the conventional LP confidence interval has correct (pointwise) asymptotic coverage even for local misspecification that is of such a large magnitude that it can be detected with probability 1 in large samples. This robustness property requires that we control for those lags of the data that are strong predictors of the outcome or impulse variables, but---crucially for applied work---the omission of lags with small-to-moderate predictive power does not threaten coverage. We argue that our result can be interpreted as a consequence of the \emph{double robustness} of the LP estimator, which is analogous to the double robustness of modern partially linear regression estimators in the literature on debiased machine learning.\footnote{\alert{Important contributions include \citet{Robins1992}, \citet{Robins1995}, \citet{robins2000profile}, \citet{robins_rotnitzky}, \citet{Bang2005}, \citet{Ai2007}, \citet{Chernozhukov2018}, and \citet{Chernozhukov2022}.}}

In stark contrast to LP, even small amounts of misspecification can cause conventional VAR confidence intervals for impulse responses to suffer from severe undercoverage. We first derive analytically the worst-case bias and coverage of VARs over all possible misspecification processes, subject to a constraint on the overall magnitude of the misspecification. The worst-case bias and coverage distortion are small if, \emph{and only if}, the asymptotic variance is close to that of LP. In general, the only way to guarantee robustness of conventional VAR inference is thus to include so many lags that the VAR estimator is asymptotically equivalent with LP. If instead the VAR confidence interval is much shorter (as is typically the case in applied practice), then it will severely undercover even for a misspecification term that: (i) is small in magnitude; (ii) has dynamic properties that cannot be ruled out \emph{ex ante} based on economic theory; and (iii) is difficult to detect \emph{ex post} with model specification tests. Instead of increasing the lag length, coverage can also be restored by using a larger bias-aware critical value \citep{Armstrong2021}, but we show that the resulting confidence intervals are so wide that one may as well report the LP interval.

\alert{We demonstrate the practical relevance of our theoretical results through a comprehensive review of current practice in the applied VAR literature, together with a simulation study. In papers published in top economics journals, researchers tend to select small to moderate lag lengths, and often report impulse responses at horizons far exceeding the lag length. Our theory suggests this practice is likely to render inference vulnerable to misspecification. To substantiate this conclusion, we simulate data calibrated to the oil shock application in \citet{Kaenzig2021}. The DGP is taken to be a VAR estimated on the paper's actual data, but with 18 lags rather than 12. The VAR confidence interval materially undercovers---particularly at medium and long horizons---if the lag length is set to 12 or selected by AIC, in line with applied practice, while the LP interval attains close to nominal coverage. Increasing the estimation lag length beyond the conventional choices ameliorates the VAR undercoverage, at the cost of delivering confidence intervals as wide as those of LP.}

\paragraph{Literature.}
Relative to the previously cited simulation studies of LPs and VARs, we here derive \emph{analytical} results on the worst-case asymptotic properties of these two inference procedures that hold for a wide range of stationary, locally misspecified VAR models. The simulations in \citet{Li2024} suggest a stark bias-variance trade-off between LP (low bias, high variance) and moderate-lag VAR estimators (moderate bias, low variance). The reason behind the theoretical superiority of LP proved in this paper is that, if the objective is to construct confidence intervals with robust coverage for a wide range of DGPs, then even a moderate amount of VAR bias cannot be tolerated, as it causes the VAR confidence interval to be poorly centered. A concern for correct confidence interval coverage thus effectively induces a large weight on bias in the researcher's objective function, justifying the use of LP despite its higher variance.

The robustness of LPs to misspecification discussed here---with stationary data and at fixed horizons---is conceptually and theoretically distinct from the robustness of LPs to the persistence in the data and the length of the impulse response horizon shown by \citet{MontielOlea2021}.  Nevertheless, it turns out that controlling for lags (``lag augmentation'') is key to all the robustness properties established in \citet{MontielOlea2021} and in the present paper.

We also build upon previous research into misspecified VAR models, uncovering novel results about the robustness of LPs and the worst-case properties of VAR procedures. \citet{Braun1993} derive expressions for the probability limits of VAR estimators under global MA misspecification; however, since bias always dominates variance asymptotically in their framework, they do not characterize the properties of LP and VAR inference procedures, which is the focus of our paper. \citet{Schorfheide2005} characterizes the asymptotic mean squared errors of iterated and direct multi-step forecasts in a reduced-form VAR model with MA terms of order $T^{-1/2}$, and \citet{GonzalesCasasus2025} use this framework to select hyperparameters for VAR forecasts. \citet{Mueller2011} construct Bayesian forecast intervals in a locally misspecified univariate AR model. Relative to these papers, we here contribute by: (i) focusing on structural impulse responses rather than forecasting; (ii) allowing for more general rates of local misspecification, which is key to uncovering the double robustness of LP; and (iii) deriving simple analytical formulae for worst-case bias and coverage of VARs. As such, our results formalize concerns by applied practitioners about the lack of VAR robustness and sensitivity to lag length (\citealp{Chari2008}; \citealp{Nakamura2018}; \alert{see also \citealp{Inoue2002}, and \citealp[Chapters 2.6.5 and 6.2]{Kilian2017}}).

Whereas our paper deals with bias imparted by dynamic misspecification, the analysis does not capture other familiar sources of small-sample bias. In particular, our asymptotics abstract from the order-$T^{-1}$ biases that arise from (i) persistence in the data \citep{Pope1990,Kilian1998,Herbst2023} and (ii) the nonlinearity of the impulse response transformation of the VAR parameters (Jensen's inequality).

\paragraph{Outline.}
\cref{sec:framework} defines the local-to-SVAR model as well as the LP and VAR estimators. \cref{sec:asy} proves the robustness of LP and the fragility of VAR confidence intervals. \cref{sec:varithmetic} derives analytically the worst-case bias and coverage of VARs, and shows that bias-aware VAR confidence intervals tend to be wider than the LP interval. \alert{\cref{sec:sim} demonstrates the practical relevance of our theoretical results through a review of the applied literature and a simulation study.} \cref{sec:conc} concludes. Replication materials are available online.\footnote{\url{https://github.com/ckwolf92/lp_var_inference}}

\paragraph{Notation.}
All asymptotic limits are taken as the sample size $T\to\infty$ and are \emph{pointwise} in the sense of fixing the true model parameters and the impulse response horizon. A sum $\sum_{\ell=a}^b c_\ell$ is defined to equal 0 when $a>b$.

\section{Framework}
\label{sec:framework}

We start out by defining the model and estimators.

\subsection{Model and assumptions}
\label{sec:model}
Extending the forecasting model of \citet{Schorfheide2005}, we consider a multivariate, stationary structural VARMA($1,\infty$) model that is local to an SVAR(1) model:
\begin{equation} \label{eqn:model}
	y_t = A y_{t-1} + H[I + T^{-\zeta}\alpha(L)]\varepsilon_t,\quad \text{for all } t,
\end{equation}
where the data vector $y_t=(y_{1,t},\dots,y_{n,t})'$ is $n$-dimensional, the shock vector $\varepsilon_t=(\varepsilon_{1,t},\dots,\varepsilon_{m,t})'$ is $m$-dimensional, $A$ is an $n \times n$ matrix, $H$ is an $n \times m$ matrix, $\alpha(L)=\sum_{\ell=1}^\infty \alpha_\ell L^\ell$ is an $m \times m$ lag polynomial, and $T$ denotes the sample size. We allow the number of shocks $m$ to potentially exceed the number of variables $n$, and \emph{vice versa}. We show below that equation \eqref{eqn:model} encompasses local-to-SVAR models with $p>1$ lags by writing them in companion form.

The model \eqref{eqn:model} captures the idea that the time series dynamics of the data are well approximated by an autoregressive model driven by unobserved white noise shocks $\varepsilon_t$, but with a small amount of misspecification in the form of an MA process $T^{-\zeta}\alpha(L)\varepsilon_t$. The misspecification is asymptotically small in the sense that the MA coefficients converge to zero at the rate $T^{-\zeta}$, though the misspecification may still affect the properties of estimators, as shown by \citet{Schorfheide2005} and as demonstrated below. We argue below that MA misspecification of this form can capture many empirically relevant types of dynamic misspecification. We consider local rather than global misspecification in the spirit of local power analysis \citep[e.g.,][]{Rothenberg1984}, since this makes the bias-variance trade-off between the VAR and LP estimators matter even asymptotically as the sample size $T$ diverges, allowing us to make tractable analytical comparisons between these two procedures.

The parameter of interest is the response at horizon $h$ of the variable $y_{i^*,t}$ with respect to the shock $\varepsilon_{j^*,t}$ for some indices $i^* \in \lbrace 1,\dots,n\rbrace$ and $j^* \in \lbrace 1,\dots,m \rbrace$, to be defined below.

\begin{asn} \label{asn:model}
For each $T$, $\lbrace y_t \rbrace_{t \in \mathbb{Z}}$ is the stationary solution to equation \eqref{eqn:model}, given the following restrictions on parameters and shocks:
\begin{enumerate}[i)]
	\item \label{itm:iid} $\varepsilon_t \overset{i.i.d.}{\sim} (0_{m \times 1}, D)$, where $D \equiv \diag(\sigma_1^2,\dots,\sigma_m^2)$, and the elements of $\varepsilon_t$ are mutually independent. For all $j=1,\dots,m$, $\sigma_j^2 >0$ and $E(\varepsilon_{j,t}^4)<\infty$.
	\item \label{itm:statio} All eigenvalues of $A$ are strictly below 1 in absolute value.
	\item \label{itm:recursive} The first $j^*$ rows of $H$ are of the form $(\tilde{H},0_{j^* \times (m-j^*)})$, where $\tilde{H}$ is a $j^* \times j^*$ lower triangular matrix with 1's on the diagonal. In particular, we require $j^* \leq n$. \label{itm:asn_triang}
	\item \label{itm:S} $S \equiv \var(\tilde{y}_t)$ is non-singular, where $\tilde{y}_t \equiv (I-AL)^{-1}H\varepsilon_t$ is the stationary solution to \eqref{eqn:model} when $\alpha(L)=0$. Specifically, $\ve(S) = (I-A \otimes A)^{-1}\ve(\Sigma)$, where $\Sigma \equiv HDH'$.
	\item \label{itm:abssum} $\alpha(L)$ is absolutely summable.
	\item $\zeta>0$.
\end{enumerate}
\end{asn}
\alert{The assumptions that the shocks are mutually and serially independent are made to simplify the exposition; we prove formally in Supplemental Appendix C.1 that our results on the robustness of LP and on the asymptotic bias of VAR go through for a large class of conditionally heteroskedastic shock processes.} The assumptions on $H$ correspond to recursive (also known as Cholesky) identification of the shock of interest $\varepsilon_{j^*,t}$, with a unit effect normalization $H_{j^*,j^*}=1$. A special case is when the shock is directly observed, which corresponds to ordering it first (i.e., $j^*=1$). \alert{Supplemental Appendix C.2 shows that our results extend also to identification via external instruments or proxies \citep{Stock2018}.} Absolute summability of $\alpha(L)$ is a weak regularity condition ensuring the vector MA($\infty$) process $\alpha(L)\varepsilon_t$ is well-defined \citep[Proposition 3.1.1]{Brockwell1991}. \alert{Finally, the assumption that $\zeta>0$ restricts attention to local misspecification, as discussed earlier.}

The impulse response of interest is defined as
\[\theta_{h,T} \equiv e_{i^*,n}'\left(A^h H + T^{-\zeta}\sum_{\ell=1}^h A^{h-\ell} H\alpha_\ell\right)e_{j^*,m} = E[y_{i^*,t+h} \mid \varepsilon_{j^*,t}=1]-E[y_{i^*,t+h} \mid \varepsilon_{j^*,t}=0],\]
where $e_{i,n}$ denotes the $n$-dimensional unit vector with a 1 in position $i$. The first term in the parenthesis is the usual VAR impulse response formula, while the second term arises from the MA component. Importantly, and consistent with our focus on the consequences of dynamic misspecification, we do not treat the VAR misspecification as non-classical measurement error that should be ignored for structural analysis; instead, the true causal model has a VARMA form (with small but potentially non-zero MA terms), and we care about the full transmission mechanism of shocks in this model.

\paragraph{Additional lags.}
Our framework covers local-to-SVAR($p$) models of the form
\begin{equation} \label{eqn:model_p}
	\check{y}_t = \sum_{\ell=1}^p \check{A}_\ell \check{y}_{t-\ell} + \check{H}[I + T^{-\zeta}\alpha(L)]\varepsilon_t,
\end{equation}
where $\check{y}_t$ is $\check{n}$-dimensional, the $\check{A}_\ell$ matrices are $\check{n} \times \check{n}$, and $\check{H}$ is $\check{n} \times m$ and satisfies \cref{asn:model}(\ref{itm:recursive}). This fits into the original model \eqref{eqn:model} if we set $n=\check{n}p$ and define the companion form representation
\[y_t = \begin{pmatrix}
	\check{y}_t \\
	\check{y}_{t-1} \\
	\check{y}_{t-2} \\
	\vdots \\
	\check{y}_{t-p+1}
\end{pmatrix},\quad A = \begin{pmatrix}
	\check{A}_1 & \check{A}_2 & \dots & \check{A}_{p-1} & \check{A}_p \\
	I & 0 & \dots & 0 & 0 \\
	0 & I & \dots & 0 & 0 \\
	\vdots & & \ddots &  & \vdots \\
	0 & 0 & \dots & I & 0
\end{pmatrix},\quad H = \begin{pmatrix}
	\check{H} \\
	0 \\
	0 \\
	\vdots \\
	0
\end{pmatrix}.\]
In particular, we can allow the estimation lag length $p$ to exceed the true minimal lag length $p_0$ of the model by setting $\check{A}_\ell=0$ for $\ell>p_0$. This fact will prove useful when we consider what happens as the lag length of the estimated VAR is increased.

\paragraph{Types of misspecification.}
Our local-to-SVAR model \eqref{eqn:model} with MA misspecification covers several empirically relevant types of model misspecification. While essentially all modern discrete-time, linearized macro models have VARMA representations, they usually cannot be represented exactly as finite-order VAR models \citep[e.g.,][Chapter 6.2]{Kilian2017}. Even if the true DGP were a finite-order VAR, dynamic misspecification of the estimation model can give rise to MA terms, for example due to under-specification of the lag length or failing to control for some of the variables in the true system. Relatedly, MA terms may appear because of a failure of invertibility of the shocks \citep{Alessi2011}.  VARMA representations can also arise from temporal or cross-sectional aggregation of finite-order VAR models, including contamination by classical measurement error \citep{Granger1976,Lutkepohl1984}. In all of these cases, if the number of lags used for estimating the VAR is chosen to be sufficiently large, then the MA remainder will be small, consistent with the spirit of our locally misspecified model \eqref{eqn:model}.

In terms of structural shock identification, our framework accommodates both the case of a well-identified shock \alert{(or instrument/proxy, see Supplemental Appendix C.2)} but misspecification in other parts of the model, as well as misspecification in the structural shock identification itself. Key to this generality is that we allow the $m \times m$ MA polynomial $\alpha(L)$ to be arbitrary. To see this, consider the case $j^*=1$, so interest centers on the dynamic causal effects of the first shock $\varepsilon_{1,t}$. If the first row of $\alpha(L)$ is zero, then $\varepsilon_{1,t}$ is well-identified as the reduced-form residual in the first equation of the VAR. If the first row of $\alpha(L)$ is non-zero, then the reduced-form residual will be contaminated by lagged shocks, thus allowing for the possibility that shock identification itself is not entirely accurate.

\subsection{Estimators}
\label{sec:estimators}

We consider two estimators of the impulse response $\theta_{h,T}$ using the data $\lbrace y_t \rbrace_{t=1}^T$:
\begin{enumerate}[1.]
	\item The \emph{LP} estimator is the coefficient $\hat{\beta}_h$ in a regression of $y_{i^*,t+h}$ on $y_{j^*,t}$, controlling for $\underline{y}_{j^*,t} \equiv (y_{1,t},\dots,y_{j^*-1,t})'$ (i.e., the variables ordered before $y_{j^*,t}$, if any) and lagged data:
	\begin{equation} \label{eqn:lp_reg}
	y_{i^*,t+h} = \hat{\beta}_h y_{j^*,t} + \hat{\omega}_h'\underline{y}_{j^*,t} + \hat{\gamma}_h' y_{t-1} + \hat{\xi}_{i^*,h,t},
	\end{equation}
	where $\hat{\xi}_{i^*,h,t}$ is the least-squares residual. Recall from the previous subsection that if we are estimating an SVAR($p$) specification in the data $\check{y}_t$, then the vector $y_{t-1}$ actually contains $p$ lags $\check{y}_{t-1},\dots,\check{y}_{t-p}$.
	\item The \emph{VAR} estimator is defined as the response of $y_{i^*,t+h}$ with respect to the $j^*$-th recursively orthogonalized innovation, where the magnitude of the innovation is normalized such that $y_{j^*,t}$ increases by one unit on impact:
	\[\hat{\delta}_h \equiv e_{i^*,n}'\hat{A}^h \hat{\nu},\]
	where
	\[\hat{A} \equiv \left(\sum_{t=2}^T y_t y_{t-1}'\right)\left(\sum_{t=2}^T y_{t-1}y_{t-1}'\right)^{-1},\quad \hat{\nu} \equiv \hat{C}_{j^*,j^*}^{-1}\hat{C}_{\bullet, j^*},\]
	and $\hat{C}_{\bullet, j^*}$ is the $j^*$-th column of the lower triangular Cholesky factor $\hat{C}$ of the covariance matrix $\hat{\Sigma} \equiv \frac{1}{T}\sum_{t=1}^T \hat{u}_t \hat{u}_t'=\hat{C}\hat{C}'$ of the residuals $\hat{u}_t \equiv y_t - \hat{A}y_{t-1}$. Again, in the case of an SVAR($p$) specification, the above formulae operate on the companion form.
\end{enumerate}

Note that the two estimators coincide at the impact horizon: $\hat{\beta}_0=\hat{\delta}_0$ (see Lemma E.5 in
Supplemental Appendix E).

It is well known that conventional confidence intervals based on both these estimators would have correct asymptotic coverage in a well-specified VAR model. However, the presence of the additional MA term in the model \eqref{eqn:model} means that, in principle, both the LP and VAR estimators ought to control for infinitely many lags of the data, rather than just one. Nevertheless, as we will now establish, this dynamic misspecification has much more serious consequences for the VAR procedure than for LP.

\section{Robust local projections, fragile VARs}
\label{sec:asy}

This section shows that the conventional LP confidence interval is robust to large amounts of misspecification. In contrast, the conventional VAR confidence interval has fragile coverage, except when it is asymptotically as wide as the LP interval, as will be the case with sufficiently large lag length.

\subsection{Large-sample distributions and confidence interval coverage}

We begin by characterizing the large-sample distributions of the LP and VAR estimators.

\paragraph{The robustness of LPs.}
\label{sec:lp_robust}

Our first main result establishes that the large-sample distribution of the LP estimator is \alert{unaffected by large amounts of misspecification}.
\begin{prop} \label{thm:lp}
	Under \cref{asn:model},
	\[\hat{\beta}_h -\theta_{h,T} = \frac{1}{\sigma_{j^*}^2}\frac{1}{T}\sum_{t=1}^T \xi_{i^*,h,,t} \varepsilon_{j^*,t} + \alert{O_p(T^{-2\zeta})} + o_p(T^{-1/2}),\]
	where
	\[\xi_{h,t} = (\xi_{1,h,t},\dots,\xi_{n,h,t})' \equiv A^h \overline{H}_{j^*}\overline{\varepsilon}_{j^*,t} + \sum_{\ell=1}^h A^{h-\ell} H\varepsilon_{t+\ell},\]
	with $\overline{H}_{j^*} \equiv (H_{\bullet, j^*+1},\dots,H_{\bullet, m})$ and $\overline{\varepsilon}_{j^*,t} \equiv (\varepsilon_{j^*+1,t},\dots,\varepsilon_{m,t})'$.
\end{prop}

\begin{proof}
See \cref{subsection:Proof_thm_lp}.
\end{proof}

\alert{The result implies that the first-order asymptotic behavior of LP does not depend on the misspecification parameter $\alpha(L)$, provided $\zeta>1/4$ so that $O_p(T^{-2\zeta}) = o_p(T^{-1/2})$.} Though this robustness property of LP is with respect to local (i.e., asymptotically vanishing) misspecification, it is still quantitatively meaningful, given that MA terms of order $T^{-\zeta}$ with $\zeta \in (1/4,1/2)$ in the model \eqref{eqn:model} can be detected with probability 1 asymptotically by conventional VAR model specification tests, such as the Hausman test considered in \cref{subsec:hausman}.

Why is LP robust to misspecification of such large magnitude? We will offer two mathematically equivalent pieces of intuition, with our discussion throughout deliberately heuristic. The classic omitted variable bias (OVB) formula suggests that the bias of the LP impulse response estimator $\hat{\beta}_h$ in the regression \eqref{eqn:lp_reg} is proportional to the product of two factors: (i) the direct effect of omitted lags on $y_{i^*,t+h}$, and (ii) the covariance of the residualized regressor of interest $y_{j^*,t}-E[y_{j^*,t} \mid \underline{y}_{j^*,t},y_{t-1}]$ with the omitted lags. The factor (i) is of order $T^{-\zeta}$ in our local-to-SVAR model \eqref{eqn:model}. The factor (ii) is also of order $T^{-\zeta}$, since the residualized regressor equals $\varepsilon_{j^*,t} + O_p(T^{-\zeta})$ under \cref{asn:model}(\ref{itm:asn_triang}), and the shock $\varepsilon_{j^*,t}$ is uncorrelated with any lagged data. Hence, the OVB is of order $T^{-2\zeta}$, so when $\zeta>1/4$, the bias of the estimator is negligible relative to the standard deviation (which is of order $T^{-1/2}$, as in the correctly specified case). This argument relies on the LP regression controlling for the most important lags of the data (i.e., $y_{t-1}$); without lagged controls, one or both factors in the OVB formula may not be small \citep{GonzalesCasasus2025}.

The preceding intuition is a special case of the \emph{double robustness} property of partially linear regressions, see Example 1.1 in \citetalias{Chernozhukov2018} and Example 1 in \citetalias{Chernozhukov2022}. We will now argue that this property applies also to LP, again settling for a heuristic argument. For notational simplicity, set $j^*=1$ so $\underline{y}_{j^*,t}=0$. Consider any dynamic model (for example a VARMA($p,q$)) that implies the following LP representation:
\begin{equation*}
y_{i^*,t+h} = \theta_{0,h} y_{1,t} + \gamma_0(y^{t-1}) + \xi_{i^*,h,t},\quad \text{where}\quad \xi_{i^*,h,t} \independent y^t \equiv (y_t,y_{t-1},\dots).
\end{equation*}
Here $\theta_{0,h}$ is the true impulse response, $\gamma_0(\cdot)$ is a function of lagged data, and ``$\independent$'' signifies independence. Define $\nu_0(y^{t-1}) \equiv E[y_{1,t} \mid y^{t-1}]$. By applying the Frisch-Waugh lemma to the regression \eqref{eqn:lp_reg}, we see that the LP estimator $\hat{\beta}_h$ is the sample analogue of the solution $\theta_{0,h}$ to the moment condition
\begin{equation*}
E[\lbrace y_{i^*,t+h}-\theta_{0,h} y_{1,t} - \gamma_0(y^{t-1}) \rbrace \lbrace y_{1,t} - \nu_0(y^{t-1}) \rbrace] = 0.
\end{equation*}
If we evaluate the moment on the left-hand side at arbitrary functions $\gamma(\cdot)$ and $\nu(\cdot)$ rather than at the true ones $\gamma_0(\cdot)$ and $\nu_0(\cdot)$, a simple calculation shows that it equals $E[\lbrace \gamma_0(y^{t-1})-\gamma(y^{t-1}) \rbrace \lbrace \nu_0(y^{t-1})-\nu(y^{t-1}) \rbrace]$.\footnote{We can write the moment as $E[\lbrace y_{i^*,t+h}-\theta_{0,h} y_{1,t} - \gamma_0(y^{t-1}) + \gamma_0(y^{t-1}) -\gamma(y^{t-1}) \rbrace \lbrace y_{1,t} - \nu(y^{t-1}) \rbrace] = E[\lbrace \gamma_0(y^{t-1}) -\gamma(y^{t-1}) \rbrace \lbrace y_{1,t} - \nu_0(y^{t-1}) + \nu_0(y^{t-1}) - \nu(y^{t-1}) \rbrace]$, since $y_{i^*,t+h}-\theta_{0,h} y_{1,t} - \gamma_0(y^{t-1})=\xi_{i^*,h,t}$ is independent of $y^t$ (orthogonality would suffice if $\nu(\cdot)$ were linear). The claim now follows from $E[y_{1,t} - \nu_0(y^{t-1}) \mid y^{t-1}]=0$ by definition of $\nu_0(\cdot)$.} Hence, the moment condition is satisfied at the true impulse response parameter $\theta_{0,h}$ as long as \emph{either} $\gamma=\gamma_0$ or $\nu=\nu_0$, making the LP estimator \emph{doubly robust}: it is consistent if we correctly specify either the controls $\gamma(y^{t-1})$ in the outcome equation or the controls $\nu(y^{t-1})$ in the implicit first-stage regression that isolates the shock $\varepsilon_{j^*,t} = y_{j^*,t}-\nu(y^{t-1})$. Because of double robustness, and as argued more generally by \citetalias{Chernozhukov2018} (and confirmed by our proof), it turns out that estimation error in $\gamma_0$ and $\nu_0$ only affects the asymptotic distribution of $\hat{\beta}_h$ through the \emph{product} of the estimation errors $\|\hat{\gamma}-\gamma_0\| \times \|\hat{\nu}-\nu_0\|$. In our local-to-SVAR model \eqref{eqn:model}, both terms in this product are of order $T^{-\zeta}$ due to the omitted lags. The product is then of order $T^{-2\zeta}$ and thus asymptotically negligible \alert{when $\zeta>1/4$}, consistent with our earlier intuition.

\paragraph{The fragility of VARs.}
In contrast to LP, the VAR estimator is fragile.

\begin{prop} \label{thm:var}
	Under \cref{asn:model},
	\begin{align*}
	\hat{\delta}_h -\theta_{h,T} &= \tr\left\lbrace S^{-1} \Psi_h HT^{-1}\sum_{t=1}^T \varepsilon_t \tilde{y}_{t-1}'\right\rbrace + \frac{1}{\sigma_{j^*}^2}e_{i^*,n}'A^h T^{-1}\sum_{t=1}^T \xi_{0,t} \varepsilon_{j^*,t} \\
		&\quad + T^{-\zeta}\abias(\hat{\delta}_h) + o_p(T^{-1/2}+T^{-\zeta}),
	\end{align*}
	where
	\[\abias(\hat{\delta}_h) \equiv \tr\left\lbrace S^{-1} \Psi_h H\sum_{\ell=1}^{\infty} \alpha_\ell DH'(A')^{\ell-1}\right\rbrace - e_{i^*,n}'\sum_{\ell=1}^h A^{h-\ell} H\alpha_\ell e_{j^*,m},\]
	\[\Psi_h \equiv \sum_{\ell=1}^h A^{h-\ell}H_{\bullet, j^*} e_{i^*,n}'A^{\ell-1},  \]
and $\lbrace \tilde{y}_t \rbrace$ and $S$ are defined in \cref{asn:model}.
\end{prop}
\begin{proof}
See \cref{subsection:Proof_thm_var}.
\end{proof}

\alert{The convergence rate $T^{-\min\lbrace 1/2,\zeta\rbrace}$ of the VAR estimator is weakly slower than the $T^{-\min\lbrace 1/2,2\zeta\rbrace}$ rate achieved by LP.} This is because the VAR estimator suffers from bias of order $T^{-\zeta}$, while the stochastic terms of order $T^{-1/2}$ are the same as they would be in a correctly specified SVAR($p$) model.\footnote{The first stochastic term captures sampling uncertainty in the reduced-form impulse responses $\hat{A}^h$, while the second term captures uncertainty in the structural impact response vector $\hat{\nu}$.} The VAR bias is only asymptotically negligible if $\zeta>1/2$, a much smaller degree of robustness than shown above for LP. The case $\zeta=1/2$ is of particular interest, as then the bias and standard deviation are of the same asymptotic order \citep[see also][]{Schorfheide2005}. MA terms of order $T^{-1/2}$ can be detected with asymptotic probability strictly between 0 and 1 by specification tests, as will be shown in \cref{subsec:hausman}.

The asymptotic bias is due to two forces: first, the coefficient matrix $\hat{A}$ is biased due to the endogeneity caused by the MA terms, and second, the VAR estimator extrapolates the horizon-$h$ impulse response based on a parametric formula $\hat{A}^h$ that does not hold exactly in the true VARMA model \eqref{eqn:model}. This is more easily seen in the special case of a univariate model $y_t = \rho y_{t-1} + [1+T^{-\zeta}\alpha(L)]\varepsilon_t$ with $n=m=1$, in which case
\[\abias(\hat{\delta}_h) \equiv \underbrace{h\rho^{h-1}}_{\frac{\partial (\rho^h)}{\partial \rho}}\underbrace{(1-\rho^2)\sum_{\ell=1}^\infty \rho^{\ell-1}\alpha_{\ell}}_{\abias(\hat{\rho})=\frac{\cov(\alpha(L)\varepsilon_t,\tilde{y}_{t-1})}{\var(\tilde{y}_{t-1})}}-\underbrace{\sum_{\ell=1}^h \rho^{h-\ell}\alpha_\ell}_{\theta_{h,T}-\rho^h},\]
where $\hat{\rho}=\hat{A}$ is the AR(1) coefficient from an OLS regression of $y_t$ on $y_{t-1}$.\footnote{Lag augmentation of the VAR impulse response estimator as in \citet{Inoue2020} may reduce the first term in the bias formula, but it does not affect the second term.} \alert{While the dynamic responses estimated by the VAR are prone to bias, the shock identification \emph{per se} is not, in the sense that the VAR's estimated impact (horizon-0) response is identical to the doubly robust LP estimate (see \cref{sec:estimators}).} 

\paragraph{Confidence intervals.}
The preceding results imply that the conventional LP confidence interval is robust to misspecification while the conventional VAR interval is not. We define the level-($1-a$) LP and VAR confidence intervals using the standard formulae:
\begin{equation} \label{eqn:ci}
\ci(\hat{\beta}_h) \equiv \left[\hat{\beta}_h \pm z_{1-a/2}\sqrt{\avar(\hat{\beta}_h)/T}\right],\quad \ci(\hat{\delta}_h) \equiv \left[\hat{\delta}_h \pm z_{1-a/2}\sqrt{\avar(\hat{\delta}_h)/T}\right].
\end{equation}
Here $z_{1-a/2}$ is the $1-a/2$ quantile of the standard normal distribution, and $\avar(\hat{\beta}_h)$ and $\avar(\hat{\delta}_h)$ are the asymptotic variances of the leading (order-$T^{-1/2}$) stochastic terms in the representations of the LP and VAR estimators in \cref{thm:lp,thm:var}; explicit formulae for the asymptotic variances are given in \cref{thm:asy_var} in \cref{app:cov}, which also implies that $\avar(\hat{\beta}_h) \geq \avar(\hat{\delta}_h)$. None of the results below would change if we replaced the asymptotic variances with the conventional consistent estimates of these (that assume correct specification, as implemented in standard econometric software packages).\footnote{Under \cref{asn:model}, homoskedastic standard errors suffice. For LP, \alert{Heteroskedasticity and Autocorrelation Robust inference would generally be required under the weaker assumptions in Supplemental Appendix C.1, though simple heteroskedasticity-robust standard errors suffice under the assumptions discussed by \citet{MontielOlea2021} and \citet{Xu2023}.}}

\begin{cor} \label{thm:coverage}
Under \cref{asn:model} \alert{and $\zeta>1/4$}, $\lim_{T\to\infty} P(\theta_{h,T} \in \ci(\hat{\beta}_h)) = 1-a$. If moreover $\avar(\hat{\delta}_h)>0$ and $\abias(\hat{\delta}_h) \neq 0$, then $\lim_{T\to\infty} P(\theta_{h,T} \in \ci(\hat{\delta}_h)) = \lim_{T\to\infty} \lbrace 1-r\left(T^{1/2-\zeta}b_h;z_{1-a/2}\right) \rbrace$, where $b_h \equiv \abias(\hat{\delta}_h)/\sqrt{\avar(\hat{\delta}_h)}$, $r(b;c) \equiv P_{Z \sim N(0,1)}(|Z+b| > c)=\Phi(-c-b)+\Phi(-c+b)$, and $\Phi(\cdot)$ is the standard normal distribution function.
\end{cor}

\begin{proof}
Considering separately the three cases $\zeta \in (1/4,1/2)$, $\zeta=1/2$, and $\zeta>1/2$, the result is an immediate consequence of \cref{thm:lp,thm:var}.
\end{proof}

LP robustly controls coverage \alert{when $\zeta>1/4$}, while the VAR confidence interval generically has coverage converging to zero for $\zeta \in (1/4,1/2)$, and strictly below the nominal level $1-a$ for $\zeta = 1/2$. Intuitively, the VAR confidence interval has the right width (the same as in the correctly specified case) but the wrong location due to the bias.

\subsection{Hausman misspecification test}
\label{subsec:hausman}

To aid in interpreting the magnitude of the local misspecification in our set-up, we consider a \citet{Hausman1978} test of correct specification of the VAR model that compares the VAR and LP impulse response estimates. This test rejects for large values of $\sqrt{T}|\hat{\beta}_h-\hat{\delta}_h|/\sqrt{\avar(\hat{\beta}_h)-\avar(\hat{\delta}_h)}$. A test of this kind was proposed by \citet{Stock2018} in the context of testing for invertibility.

\begin{prop} \label{thm:hausman}
	Impose \cref{asn:model}, \alert{$\zeta>1/4$}, and $\avar(\hat{\beta}_h) > \avar(\hat{\delta}_h) > 0$. Then the asymptotic rejection probability of the Hausman test equals
	\[\lim_{T \to \infty} P\left(\frac{\sqrt{T}|\hat{\beta}_h-\hat{\delta}_h|}{\sqrt{\avar(\hat{\beta}_h)-\avar(\hat{\delta}_h)}}>z_{1-a/2}\right) = \lim_{T\to\infty} r\left(\frac{T^{1/2-\zeta}b_h}{\sqrt{\avar(\hat{\beta}_h)/\avar(\hat{\delta}_h)-1}} ;z_{1-a/2}\right),\]
	where $b_h$ and $r(\cdot,\cdot)$ were defined in \cref{thm:coverage}.
\end{prop}
\begin{proof}
Considering separately the three cases $\zeta \in (1/4,1/2)$, $\zeta=1/2$, and $\zeta>1/2$, the result follows from \cref{thm:lp,thm:var} as well as \cref{thm:asy_var} in \cref{app:cov}.
\end{proof}
As claimed previously, the Hausman test is consistent against MA misspecification of order $T^{-\zeta}$ with $\zeta \in (1/4,1/2)$, except in the knife-edge case where $\abias(\hat{\delta}_h)=0$. When $\zeta=1/2$ and $\abias(\hat{\delta}_h) \neq 0$, the asymptotic rejection probability is strictly between the significance level $a$ and 1. In \cref{sec:varithmetic} we will use the Hausman test to quantify the difficulty of detecting especially pernicious types of model misspecification.

\subsection{Lag length selection}
\label{subsec:lags}

\alert{We now elaborate further on the role of the estimation lag length. We first discuss the properties of LP when the lag length is selected using standard information criteria. We then show that increasing the lag length robustifies VAR inference by rendering it equivalent with LP inference.}

\paragraph{Lag length selection for LP.}
\alert{Supplemental Appendix C.3 shows that the LP confidence interval maintains correct asymptotic coverage when the lag length $p$ is selected via the Bayesian Information Criterion (BIC), provided that $\zeta \geq 1/2$. Specifically, the BIC should be applied to an auxiliary VAR in the observed data series $\check{y}_t$; the selected lag length then determines the number $p$ of lags to control for in subsequent LP inference (which otherwise discards the auxiliary VAR). The resulting LP inference is robust to model selection errors that are known to cause difficulties for VAR inference (\citealp{Leeb2005}; \citealp[chapter 2.6.5]{Kilian2017}). At a high level, the greater reliability of LP inference with data-dependent lag length is a consequence of the double robustness discussed earlier, see \citet{Belloni2013} and \citetalias{Chernozhukov2018}.}

\alert{While the BIC suffices in theory for valid local projection inference, we follow \citet{Kilian2017} and recommend that researchers employ the more conservative Akaike Information Criterion (AIC) in practice. The reason is that, in finite samples, the downsides of under-specifying the lag length outweigh the slight inefficiency associated with over-selecting the lag length. \cref{sec:sim} demonstrates that this lag length selection procedure delivers LP confidence intervals with accurate coverage in realistic DGPs.}

\paragraph{Long-lag VARs.}
One simple way to remove the asymptotic bias of the VAR estimator is to control for sufficiently many lags\alert{---typically many more lags than indicated by conventional information criteria}. This is because in this case the estimator is asymptotically equivalent with the LP estimator. See \citet{Plagborg2021} and \citet{Xu2023} for related results in models without explicit MA misspecification.

\begin{cor} \label{thm:lag}
Suppose the model \eqref{eqn:model_p} written in companion form \eqref{eqn:model} satisfies \cref{asn:model} \alert{and $\zeta>1/4$}. Let $\tilde{\check{y}}_t$ denote the stationary solution to equation \eqref{eqn:model_p} when $\alpha(L)=0$. If $\varepsilon_{j^*,t-\ell} \in \sspan(\tilde{\check{y}}_{t-1},\dots,\tilde{\check{y}}_{t-p})$ for all $\ell=1,\dots,h$, then $\abias(\hat{\delta}_h)=0$ and $\avar(\hat{\delta}_h)=\avar(\hat{\beta}_h)$. In particular, these results obtain if either of the following two sufficient conditions hold:
\begin{enumerate}[i)]
	\item The model is a local-to-SVAR($p_0$) model (i.e., $\check{A}_\ell=0$ for $p_0<\ell\leq p$) and $h \leq p-p_0$, where $p$ is the estimation lag length.
	\item The shock of interest is directly observed and ordered first (i.e., $j^*=1$ and $\check{A}_{1,j,\ell}=0$ for all $j,\ell$), and $h \leq p$.
\end{enumerate}
\end{cor}
\begin{proof}
See \cref{subsection:proof_thm_lag}.
\end{proof}

\alert{We see that, the larger the impulse horizon $h$ of interest, the larger is the estimation lag length $p$ required for bias reduction.} In fact, \cref{sec:varithmetic} shows that the \emph{only} way to guarantee that the asymptotic bias of the VAR estimator is zero is to control for so many lags that LP and VAR are asymptotically equivalent.

\section{VAR inference under bounded misspecification}
\label{sec:varithmetic}

To show that the fragility of VARs is likely to matter in practice, we now investigate the worst-case properties of VAR procedures under a tight constraint on the amount of misspecification. We prove that the conventional VAR confidence interval is robust if, \alert{\emph{and only if},} LP and VAR intervals coincide asymptotically. VARs with short-to-moderate lag lengths suffer from severe coverage distortions even for small amounts of misspecification that are hard to rule out either economically or statistically. Beyond increasing the lag length, an alternative strategy to fix VAR undercoverage is to use a larger bias-aware critical value; however, we show that the resulting confidence interval is usually wider than the LP interval. Finally, we show that all conclusions extend to the case of joint inference on multiple impulse responses.

Throughout this section we set $\zeta=1/2$ so that the asymptotic bias-variance trade-off between LP and VAR is non-trivial.

\subsection{Worst-case bias and mean-squared error}
\label{subsec:wc_bias_mse}

Building towards our main results on VAR coverage distortions, we begin by deriving the worst-case bias and mean-squared error of the VAR estimator.

\paragraph{Misspecification bound.}
To quantify the amount of misspecification in the local-to-SVAR model \eqref{eqn:model} with $\zeta=1/2$, we define the \emph{noise-to-signal ratio} 
\[\tr\left\lbrace \var(T^{-1/2}\alpha(L)\varepsilon_t)\var(\varepsilon_t)^{-1} \right\rbrace = \tr\left\lbrace \left(T^{-1}\sum_{\ell=1}^\infty \alpha_\ell D \alpha_\ell'\right)D^{-1}\right\rbrace = T^{-1}\|\alpha(L)\|^2,\]
where we define the norm $\|\alpha(L)\| \equiv \sqrt{\sum_{\ell=1}^\infty \tr\lbrace D \alpha_\ell'D^{-1}\alpha_\ell\rbrace}$. Suppose we are willing to impose \emph{a priori} that the noise-to-signal ratio is at most $M^2/T$ for some constant $M \in (0,\infty)$. For small $M^2/T$, this roughly means that a fraction $M^2/T$ of the variance of the model's error term is due to the misspecification. This corresponds to restricting the parameter space for $\alpha(L)$ to all absolutely summable lag polynomials that satisfy $\|\alpha(L)\|\leq M$. In the following we will consider the worst-case properties of the VAR estimator over this parameter space, treating the other (consistently estimable) parameters $(A,H,D)$ as fixed.

\paragraph{Worst-case bias.}

\begin{prop} \label{thm:bias_bound}
Impose \cref{asn:model}, $\zeta=1/2$, and $\avar(\hat{\delta}_h) > 0$. Then
\[\max_{\alpha(L) \colon \|\alpha(L)\| \leq M} |b_h| = M\sqrt{\frac{\avar(\hat{\beta}_h)}{\avar(\hat{\delta}_h)}-1},\]
where we recall the definition $b_h=\abias(\hat{\delta}_h)/\sqrt{\avar(\hat{\delta}_h)}$. Recall also that $\avar(\hat{\beta}_h)$ and $\avar(\hat{\delta}_h)$ do not depend on $\alpha(L)$.
\end{prop}

\begin{proof}
The claim is a special case of Proposition C.2 in Supplemental Appendix C.4.
\end{proof}

Under our bound $M^2/T$ on the noise-to-signal ratio, the worst-case (scaled) VAR bias is a simple function of $M$ and of the relative asymptotic precision $\avar(\hat{\beta}_h)/\avar(\hat{\delta}_h)$ of the VAR estimator vs.\ LP. These two quantities are ``sufficient statistics'' for the worst-case bias regardless of the number $n$ of variables in the VAR, the lag length $p$, the specific VAR parameters $(A,H,D)$, and the horizon $h$. Hence, our subsequent analysis of the worst-case properties of VAR procedures depends only on $M$ and on the relative precision, allowing us to concisely present analytical results that cover a wide range of local-to-SVAR models without having to resort to simulations that inevitably only cover a finite number of DGPs. 

\cref{thm:bias_bound} shows that \alert{VAR estimators must trade off efficiency and robustness}: the worst-case VAR bias is small precisely when the VAR estimator has nearly the same variance as LP. While the worst-case bias can be reduced by increasing the VAR estimation lag length $p$, the proposition shows that this can \emph{only} happen at the expense of increasing the variance. If we include so many lags that the worst-case bias is zero (cf.\ \cref{thm:lag}), then the VAR estimator must \emph{necessarily} be asymptotically equivalent with LP.

\paragraph{Worst-case mean squared error.}
For future reference we briefly discuss how the worst-case mean squared error (MSE) of the VAR estimator depends on the imposed bound on misspecification. Based on \cref{thm:lp,thm:var} as well as \cref{thm:asy_var}, we define the asymptotic MSE of the VAR and LP estimators as follows:
\[\aMSE(\hat{\beta}_h) \equiv \avar(\hat{\beta}_h),\quad \aMSE(\hat{\delta}_h) \equiv \abias(\hat{\delta}_h)^2 + \avar(\hat{\delta}_h).\]
\begin{cor} \label{thm:mse}
	Impose \cref{asn:model} and $\zeta=1/2$. Then
	\[\sup_{\alpha(L) \colon \|\alpha(L)\|\leq M} \lbrace \aMSE(\hat{\delta}_h)-\aMSE(\hat{\beta}_h)\rbrace = (M^2-1)\lbrace \avar(\hat{\beta}_h)-\avar(\hat{\delta}_h)\rbrace.\]
\end{cor}
\begin{proof}
See \cref{subsection:proof_thm_mse}.
\end{proof}
In words, the worst-case MSE regret of VAR relative to LP is proportional to the variance reduction of VAR relative to LP, with a proportionality constant of $M^2-1$. If $M > 1$ (corresponding to a noise-to-signal ratio greater than $1/T$), the worst-case MSE of VAR thus strictly exceeds the MSE of LP. From here it is also straightforward to recover the minimax optimal way to average LP and VAR estimates.

\begin{cor} \label{thm:model_avg}
	Impose \cref{asn:model}, $\zeta=1/2$, and $\avar(\hat{\beta}_h)>\avar(\hat{\delta}_h)$. Consider the model-averaging estimator $\hat{\theta}_h(\omega) \equiv \omega \hat{\beta}_h + (1-\omega)\hat{\delta}_h$, and denote its asymptotic MSE by $\aMSE(\hat{\theta}_h(\omega))$. Then
	\[\argmin_{\omega \in \mathbb{R}} \sup_{\alpha(L) \colon \|\alpha(L)\|\leq M} \aMSE(\hat{\theta}_h(\omega)) = \frac{M^2}{1+M^2}.\]
\end{cor}
\begin{proof}
See \cref{subsection:proof_thm_model_avg}.
\end{proof}
If $M=1$, it is minimax optimal to weight the LP and VAR estimates equally. If $M=2$ (corresponding to a noise-to-signal ratio of $4/T$), the LP estimator receives 80\% weight.

\subsection{Worst-case coverage}
\label{sec:var_coverage}

We now turn to our main area of interest: the worst-case asymptotic coverage of the conventional VAR confidence interval under our bound on the amount of misspecification. This turns out to take a very simple form.

\begin{cor} \label{thm:coverage_worstcase}
	Impose \cref{asn:model}, $\zeta=1/2$, and $\avar(\hat{\delta}_h)>0$. Then
	\[\inf_{\alpha(L) \colon \|\alpha(L)\|\leq M} \lim_{T\to\infty} P(\theta_{h,T} \in \ci(\hat{\delta}_h)) = 1-r\left(M\sqrt{\avar(\hat{\beta}_h)/\avar(\hat{\delta}_h)-1};z_{1-a/2}\right).\]
\end{cor}
\begin{proof}
This is an immediate consequence of \cref{thm:coverage,thm:bias_bound}.
\end{proof}

Based on this corollary, \cref{fig:coverage_worstcase} provides a complete characterization of the robustness-efficiency trade-off for VAR confidence intervals. It plots the worst-case coverage probability as a function of the ratio of standard errors for VAR and LP, given significance level $a=10\%$ and different values of $M$. The shaded area depicts an empirically relevant range of standard error ratios obtained in four empirical applications from \citet{Ramey2016}.\footnote{We replicate \citeauthor{Ramey2016}'s identification schemes for monetary, tax, government spending, and technology shocks. The shaded area shows the 10th to 90th percentiles of standard error ratios at horizons exceeding 1 year. See the online replication materials for details.} We see that, even for $M=1$ (corresponding to a noise-to-signal ratio of $1/T$), the worst-case coverage probability is below 48\% whenever the asymptotic standard deviation of the VAR estimator is less than half that of LP---a value that is typical in applied work. Further, at the bottom end of the empirically relevant range, the worst-case coverage probability is essentially zero as soon as $M \geq 1$. It is only at the very right side of the figure---when the VAR includes enough lags to remove nearly all bias, thus increasing the standard error almost to that of LP---that the VAR confidence interval has coverage close to the nominal level.

\begin{figure}[tp]
	\centering
	\includegraphics[width=0.75\linewidth]{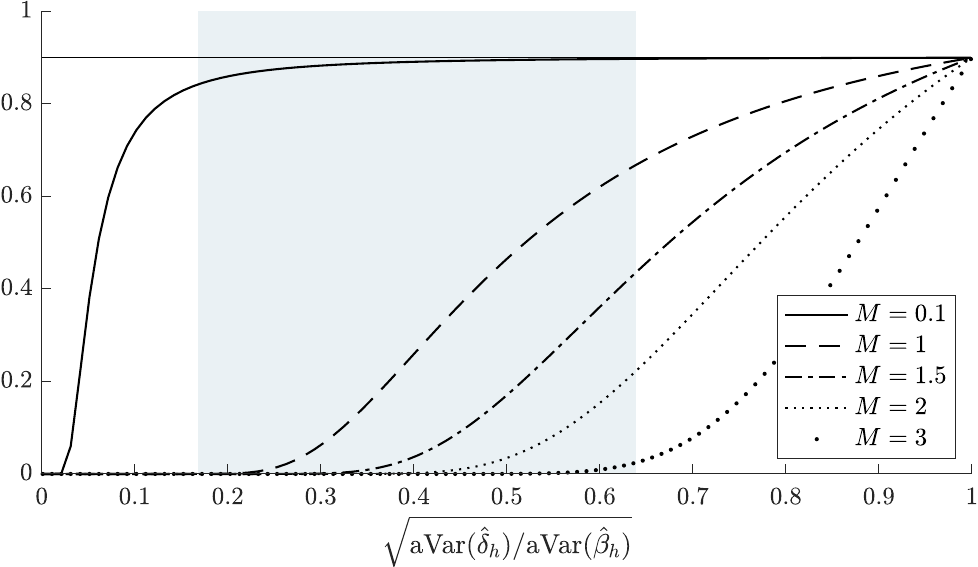}
	\caption{Worst-case asymptotic coverage probability of the conventional 90\% VAR confidence interval. Horizontal axis: relative asymptotic standard deviation of VAR vs.\ LP. Different lines: different bounds $M$ on $\|\alpha(L)\|$. Shaded area: empirical 10th--90th percentile range of relative standard errors based on \citet{Ramey2016}, see the online replication materials for details. The solid horizontal line marks the nominal coverage probability $1-a=90\%$.} \label{fig:coverage_worstcase}
\end{figure}

The potential for VAR undercoverage documented here may not be so concerning if the worst-case misspecification can be ruled out on economic theory grounds, or if it is easily detectable statistically. We now argue that neither appears to be the case.

\paragraph{Economic theory.}
The shape and magnitude of the least favorable misspecification is difficult to rule out generally based on economic theory. The least favorable MA polynomial $\alpha^\dagger(L;h,M)=\sum_{\ell=1}^\infty \alpha_{\ell,h,M}^\dagger L^\ell$ for VAR coverage is the same as the least favorable one for bias (i.e., the $\alpha(L)$ that achieves the maximum in \cref{thm:bias_bound}). Since $\abias(\hat{\delta}_h)$ is linear in $\alpha(L)$, the least favorable choice given the constraint $\|\alpha(L)\|\leq M$ follows from the Cauchy-Schwarz inequality (see the proof of Proposition C.2 in Supplemental Appendix C.4):
\begin{equation} \label{eqn:lf}
\alpha_{\ell,h,M}^\dagger \propto D^{1/2}H'\Psi_h'S^{-1}A^{\ell-1}HD^{1/2} - \mathbbm{1}(\ell \leq h)\sigma_{j^*}^{-1}D^{1/2}H'(A')^{h-\ell}e_{i^*,n}e_{j^*,m}',\quad \ell \geq 1,
\end{equation}
where the constant of proportionality (which does not depend on the lag $\ell$) is chosen so that $\|\alpha^\dagger(L;h,M)\|=M$. Note that the shape of the least favorable MA polynomial depends on the particular horizon $h$ of interest but not on $M$; i.e., the bound $M^2/T$ on the noise-to-signal ratio only scales the polynomial up or down.

We note two main properties of the least favorable misspecification. First, the magnitude of the MA coefficients $\alpha_{\ell,h,M}^\dagger$ decays exponentially as $\ell\to\infty$. In other words, not only is the overall magnitude of the least favorable model misspecification small (as imposed in the noise-to-signal bound), the MA coefficients at long lags are in fact particularly small. Second, numerical examples shown in \cref{app:least_fav} suggest that the MA coefficients tend to be largest in magnitude at horizon $h$, displaying either a hump-shaped pattern as a function of $\ell$---consistent with economic theories of adjustment costs or learning---or a single zig-zag pattern---consistent with theories of overshooting or lumpy adjustment. We thus view MA dynamics of the worst-case form as empirically and theoretically relevant.\footnote{However, the least favorable MA polynomial derived above need not be of interest to researchers who trust that some equations in their SVAR specification are exactly correctly specified, as this imposes the additional restrictions that some linear combinations of the rows of the MA polynomial $\alpha(L)$ equal zero.}

\paragraph{Statistical tests.}
The least favorable misspecification is also difficult to detect statistically. \cref{thm:hausman,thm:bias_bound} imply that, for $\alpha(L)=\alpha^\dagger(L;h,M)$, the asymptotic rejection probability of the Hausman test of correct VAR specification equals $r(M;z_{1-a/2})$. When $M=1$ (corresponding to a noise-to-signal ratio of $1/T$), the odds of the Hausman test \emph{failing} to reject the misspecification are nearly 3-to-1 at significance level $a=10\%$, since $r(1;z_{0.95}) = 26\%$. At significance level $a=5\%$, the odds are nearly 5-to-1, since $r(1;z_{0.975})=17\%$. Standard \emph{ex post} model misspecification tests are thus unlikely to indicate a problem even if the potential for undercoverage is severe.

Rather than committing \emph{a priori} to a parameter space for $\alpha(L)$ through choice of $M$, we can also ask a different question: across all possible types and magnitudes of misspecification, what is the worst-case probability that the conventional VAR confidence interval fails to cover the true impulse response, yet we fail to reject correct specification of the VAR model?
\begin{cor} \label{thm:jointprob_worstcase}
Impose \cref{asn:model}, $\zeta=1/2$, and $\avar(\hat{\beta}_h)>\avar(\hat{\delta}_h)>0$. Consider the joint event $\mathcal{A}_T$ that $\theta_{h,T} \notin \ci(\hat{\delta}_h)$ and the Hausman test in \cref{thm:hausman} fails to reject misspecification. Then
\[\sup_{\alpha(L)}\; \lim_{T\to\infty} P(\mathcal{A}_T) = \sup_{b \geq 0}\; r(b;z_{1-a/2})\left\lbrace 1 - r\left(\frac{b}{\sqrt{\avar(\hat{\beta}_h)/\avar(\hat{\delta}_h)-1}} ;z_{1-a/2}\right)\right\rbrace,\]
where the supremum on the left-hand side is taken over all absolutely summable lag polynomials $\alpha(L)$.
\end{cor}
\begin{proof}
See \cref{subsection:proof_thm_jointprob_worstcase}.
\end{proof}

\begin{figure}[tp]
	\centering
	\includegraphics[width=0.75\linewidth]{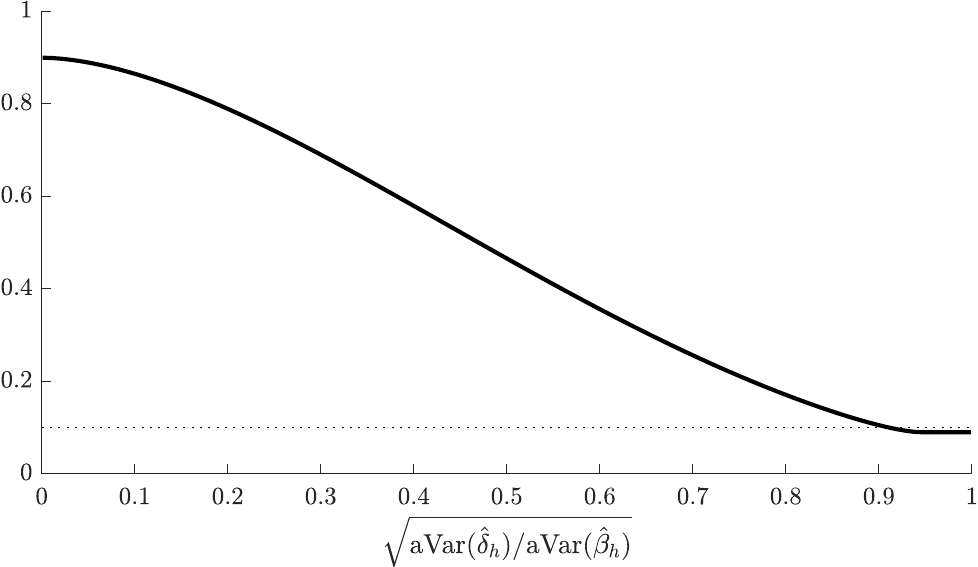}  
	\caption{Worst-case asymptotic probability of the joint event that the conventional VAR confidence interval fails to cover the true impulse response and yet the Hausman test fails to reject misspecification. Horizontal axis: relative asymptotic standard deviation of VAR vs.\ LP. The dotted horizontal line marks the nominal significance level $a=10\%$.} \label{fig:jointprob_worstcase}
\end{figure}

\cref{fig:jointprob_worstcase} plots this worst-case probability for a significance level of $a=10\%$, which by \cref{thm:jointprob_worstcase} depends only on the ratio $\avar(\hat{\delta}_h)/\avar(\hat{\beta}_h)$. Under correct specification, the probability of the joint event is equal to $a(1-a)$ ($=9\%$ when $a=10\%$). With misspecification, the joint probability instead exceeds 46\% when the asymptotic standard deviation of the VAR estimator is less than half that of the LP estimator. As $\avar(\hat{\delta}_h)/\avar(\hat{\beta}_h) \to 0$, the worst-case joint probability approaches $1-a$. We thus again see that statistical tests may fail to warn against the potential for severe VAR coverage distortions.

\subsection{Bias-aware inference}
\label{sec:var_biasaware}

Rather than removing bias by increasing the lag length (thus ensuring equivalence with LP), an alternative way to fix the undercoverage of the conventional VAR confidence interval is to adjust the critical value upward to compensate for the bias, as suggested in a general setting by \citet{Armstrong2021}. Suppose again that we restrict the misspecification $\alpha(L)$ to satisfy $\|\alpha(L)\|\leq M$. Then we define the \emph{bias-aware} VAR confidence interval
\begin{equation*}
\ci_B(\hat{\delta}_h;M) \equiv \left[\hat{\delta}_h \pm \cv_{1-a}\left(M\sqrt{\frac{\avar(\hat{\beta}_h)}{\avar(\hat{\delta}_h)}-1}\right)\sqrt{\avar(\hat{\delta}_h)/T}\right],
\end{equation*}
where the bias-aware critical value $\cv_{1-a}(b)$ is given by the number such that $r(b;\cv_{1-a}(b))=a$, and $r(\cdot,\cdot)$ is defined in \cref{thm:coverage}. By construction, this bias-aware confidence interval has correct (but potentially conservative) asymptotic coverage.
\begin{cor} \label{thm:coverage_biasaware}
	Impose \cref{asn:model}, $\zeta=1/2$, and $\avar(\hat{\delta}_h)>0$. Then
	\[\inf_{\alpha(L) \colon \|\alpha(L)\|\leq M}\lim_{T\to\infty} P(\theta_{h,T} \in \ci_B(\hat{\delta}_h;M)) = 1-a.\]
\end{cor}
\begin{proof}
The result follows immediately from \cref{thm:var,thm:bias_bound}.
\end{proof}
It turns out, however, that a very tight bound $M$ on the signal-to-noise ratio is required for the bias-aware VAR interval to be shorter than the LP interval. \cref{fig:rel_length} plots the relative interval length as a function of the relative asymptotic standard deviation of VAR and LP, for a significance level of $a = 10\%$ and for different misspecification bounds $M$. The figure shows that $M$ has to be quite small---apparently below 1---for the bias-aware VAR length to dominate the LP length regardless of the DGP and horizon. Even for $M=1.5$, bias-aware VAR is at best only moderately shorter than LP. Finally, for values of $M$ above 2 (corresponding to a noise-to-signal ratio above $4/T$), bias-aware VAR is dominated by LP.

\begin{figure}[tp]
	\centering
	\includegraphics[width=0.75\linewidth]{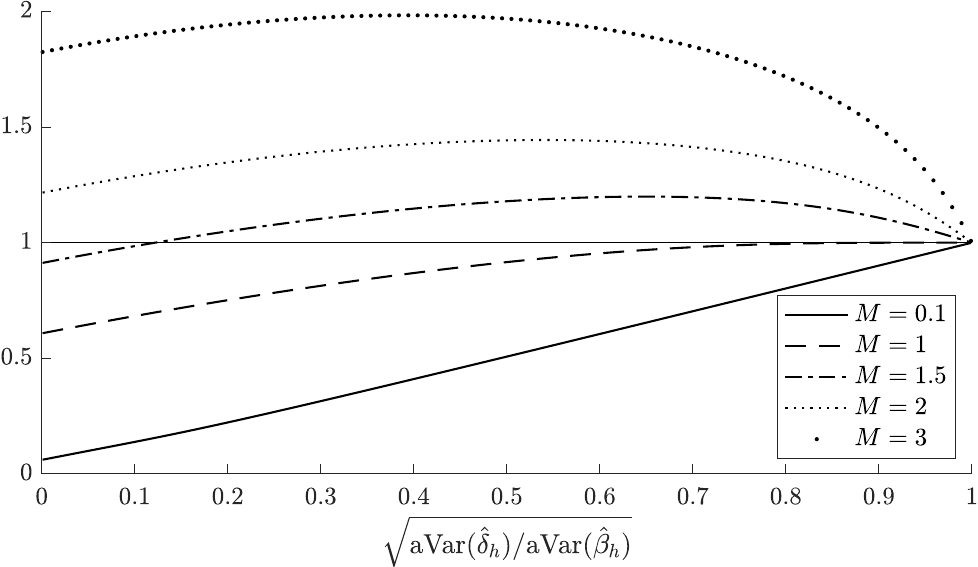}  
	\caption{Relative length of bias-aware VAR confidence interval vs.\ conventional LP interval. Significance level $a=10\%$. Horizontal axis: relative asymptotic standard deviation of VAR vs.\ LP. Different lines: different bounds $M$ on $\|\alpha(L)\|$. The solid horizontal line marks the value 1.} \label{fig:rel_length}
\end{figure}

In \cref{app:biasaware_opt} we furthermore show that the conventional LP confidence interval is at worst slightly wider than a more efficient bias-aware confidence interval centered at the model averaging estimator $\hat{\theta}_h(\omega) = \omega \hat{\beta}_h + (1-\omega)\hat{\delta}_h$, introduced in \cref{thm:model_avg} above. Even if the weight $\omega$ is chosen to optimize confidence interval length, the gains relative to the LP interval are very small when $M \geq 2$ (corresponding to a noise-to-signal ratio above $4/T$).

We thus conclude that, while bias-aware VAR inference is possible in theory, in practice the gains relative to the simpler LP interval are small at best, unless we put an extremely tight bound on the noise-to-signal ratio.

\subsection{Inference on multiple impulse responses}
\label{sec:joint}

Since the least favorable MA polynomial derived in \cref{sec:var_coverage} depends on the horizon $h$ of interest, one might hope that VARs would not be as prone to bias and thereby undercoverage if interest centers on \emph{multiple} impulse responses. \alert{Unfortunately, Supplemental Appendix C.4 shows that this is not the case. There we consider inference on a vector of impulse responses for any combination of response variables $i$, shocks $j$, and horizons $h$. Generalizing \cref{thm:bias_bound}, we show that the worst-case norm of the bias is non-negligible if the VAR offers efficiency gains for \emph{any} linear combination of the parameters of interest. This implies that the conventional VAR confidence interval has fragile coverage even if the target parameter is a linear combination of impulse responses (such as the integral or sum across multiple horizons, as in the fiscal multiplier applications reviewed in \citealp{Ramey2016}). The conventional Wald confidence ellipsoid centered at the VAR estimator is similarly fragile.}

\section{Practical relevance}
\label{sec:sim}

\alert{This section establishes the practical relevance of our theoretical conclusions by comprehensively reviewing current practice for lag length selection in the applied VAR literature, coupled with a simulation study calibrated to the application in \citet{Kaenzig2021}.}

\subsection{Review of current practice for VAR lag length selection}
\label{subsec:practice}

\alert{To evaluate lag length selection practice in the applied VAR literature, we created a comprehensive list of articles published between January 2015 and June 2025 in six top economics journals: \emph{American Economic Review}, \emph{American Economic Journal: Macroeconomics}, \emph{Econometrica}, \emph{Journal of Political Economy}, \emph{Quarterly Journal of Economics}, and \emph{Review of Economic Studies}. We restrict attention to papers that use VARs on time series data (not panel data) to estimate \emph{structural} impulse response functions. This yields 81 papers total. For each paper, we picked a single specification as the ``main'' one.\footnote{In a few cases, the main use of VARs in a paper was only reported in the appendix.} To be conservative, if multiple specifications received equal attention in the main analysis, we picked the one with the largest lag length. The full list of papers, together with the recorded information about the VAR specifications, is provided in the online replication materials.}

\alert{Our findings suggest that the theoretical results of the preceding sections are likely to have bite in practice, as typical VAR estimation lag lengths in applied papers are short or moderate. The modal lag length is 4 in quarterly data and 12 in monthly data, the mean is slightly below the mode, and less than 10\% of papers employ lag lengths greater than or equal to twice the modal values.\footnote{\alert{Additionally, all the various empirical VAR specifications reported in the handbook chapters of \citet{Ramey2016} and \citet{Stock2016} are consistent with these patterns.}} 20\% of papers select the lag length in a data-dependent way, often using information criteria. On average across papers, the estimation lag length is only 28\% as large as the the longest reported impulse horizon. We conclude that few applied papers follow the recommendation of \citet[][pp.\ 58--66]{Kilian2017} to use long lag lengths and avoid information criteria. This suggests that VAR inference results reported in much of the applied literature could be subject to the fragility we highlighted in the preceding sections. In fact, since around 40\% of papers employ Bayesian shrinkage, the estimation bias could be even larger than indicated by the lag length alone.}

\subsection{Empirically calibrated simulation study}
\label{subsec:simul_oil}

\alert{We now show through simulations that our asymptotic results are informative about the performance of LPs and VARs in an empirically relevant finite-sample setting when the lag length is selected as in current applied practice. Our DGP is calibrated to the oil news shock application in \citet{Kaenzig2021}.}

\paragraph{Set-up.}
\alert{The DGP is a VAR estimated on the dataset of \citet{Kaenzig2021}, using a somewhat longer lag length than that used in the paper. The data series are that paper's oil shock proxy, the real price of oil, world oil production, world oil inventories, world industrial production, U.S. industrial production, and the U.S. consumer price index (CPI). Whereas \citeauthor{Kaenzig2021} employs 12 lags, we estimate a recursively identified VAR($18$) by OLS on his data and use this as the simulation DGP, with i.i.d.\ Gaussian shocks. We do not claim that the VAR(18) DGP is more ``realistic'' than a VAR(12) estimated on the same data, but we contend that it is desirable that confidence intervals should have reliable coverage in both these DGPs.}

\alert{The parameters of interest are the impulse responses of CPI to the observed oil shock. To be consistent with our theory, we use an ``internal instruments'' specification that orders the proxy first in the VAR; this differs from the ``external instruments'' specification used by \citeauthor{Kaenzig2021}. The sample size is $T = 720$ months. We will entertain different choices of the estimation lag length $p$ for the LP and VAR estimators. We report results for both delta method and bootstrap confidence intervals. Results are based on 10,000 Monte Carlo simulations. See Supplemental Appendix D for implementation details.}

\paragraph{Results.}
\alert{\cref{fig:oil_simulations} shows that our theoretical results on LP robustness and VAR fragility are practically relevant. The figure depicts the coverage probabilities (left panel) and median confidence interval length (right panel) for VARs (in red, solid and dashed) and LPs (in blue, solid and dashed). The top panel fixes the estimation lag length at $p = 12$, while the bottom panel selects the lag length by AIC. Both these choices are frequently encountered in the applied literature, as documented earlier. Given these conventional lag lengths, VAR confidence intervals tend to be shorter than LP intervals, but quite materially undercover, with coverage falling below 60\% at medium and long horizons. LP instead attains close to the nominal coverage level of 90\% throughout, as expected. The mean lag length selected by AIC is $9.7$, evidently insufficient to guard against dynamic misspecification.}

\begin{figure}[t!]
\centering
{\textsc{\small Lag length $p=12$}} \\
\vspace{0.6\baselineskip}
\begin{subfigure}{\textwidth}
\centering
\includegraphics[width=0.825\textwidth]{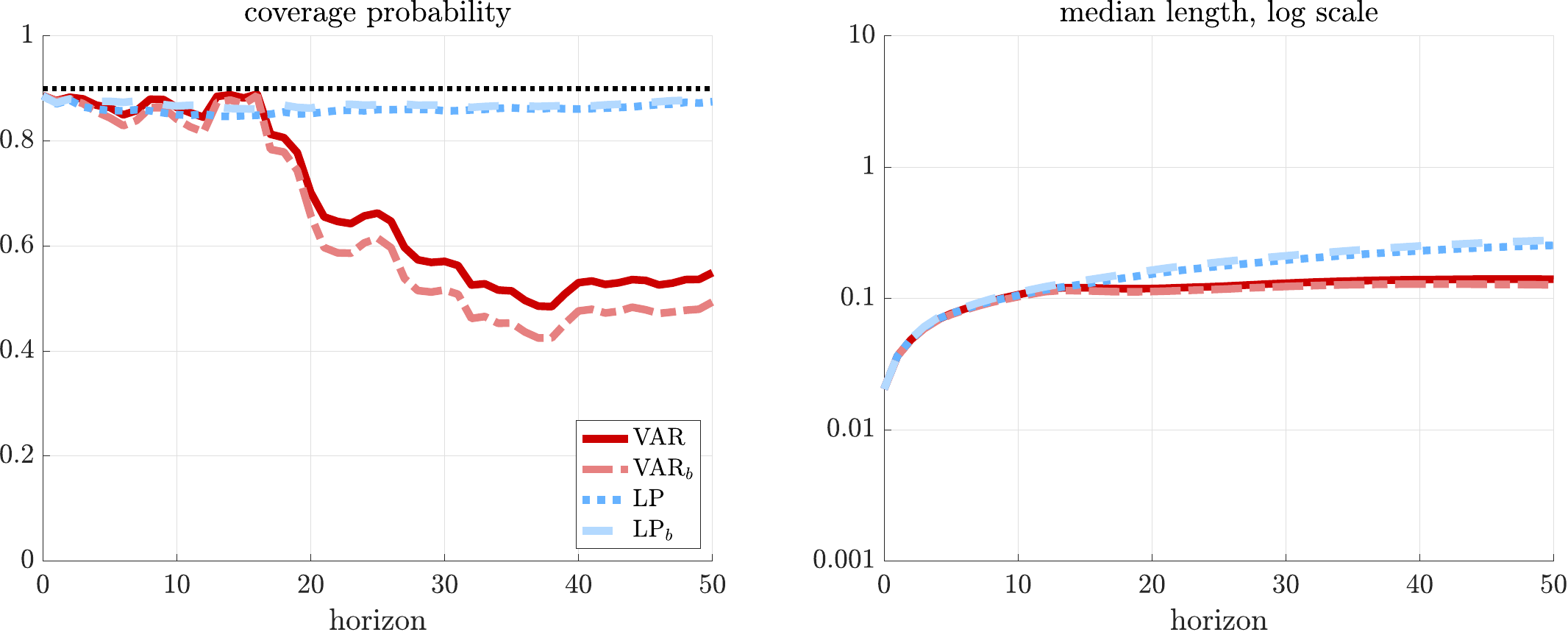}\vspace{0.2cm}
\end{subfigure}
\centering
{\textsc{\small Lag length via AIC}} \\
\vspace{0.6\baselineskip}
\begin{subfigure}{\textwidth}
\centering
\includegraphics[width=0.825\textwidth]{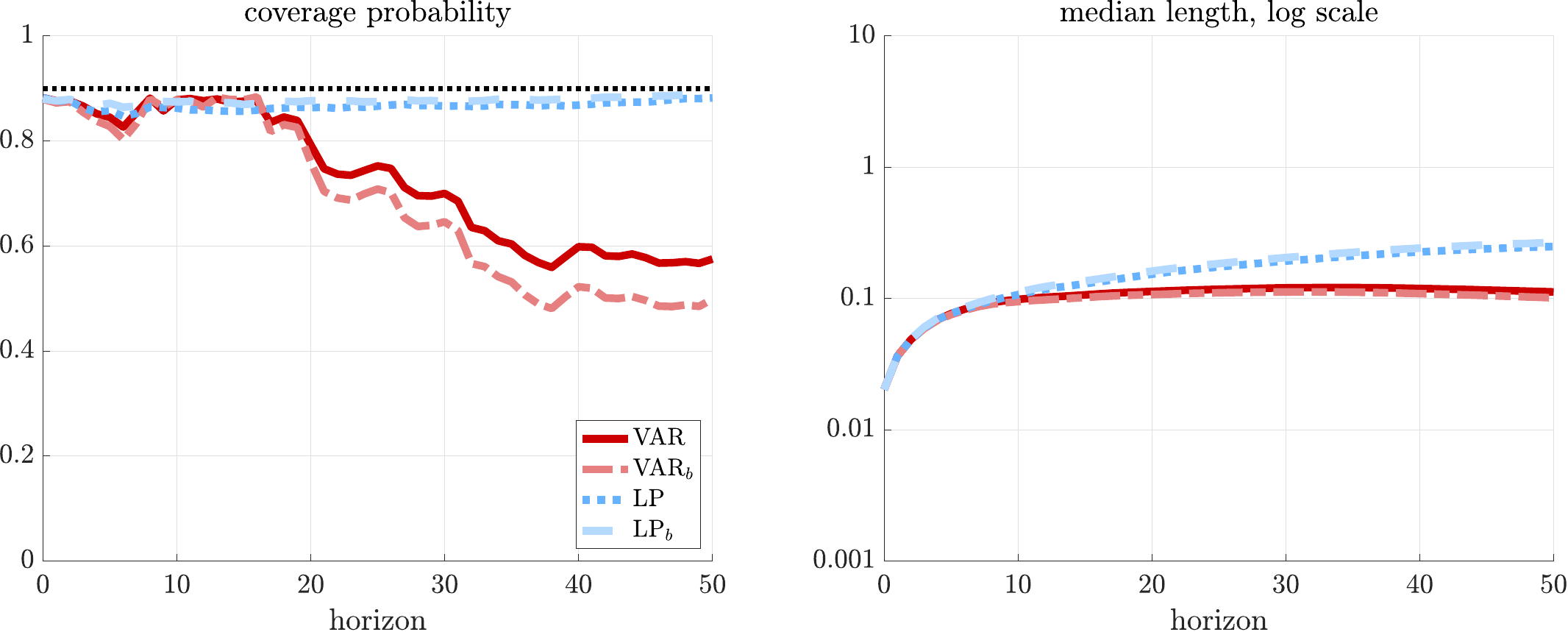}\vspace{0.2cm}
\end{subfigure}
\caption{Coverage probability (left) and median length (right) for VAR (red) and LP (blue) nominal 90\% confidence intervals computed via the delta method or bootstrap (the latter are indicated with subscript ``b'' in the legends). Lag length: fixed at $p = 12$ in the top panel, and selected using AIC in the bottom panel.}
\label{fig:oil_simulations}
\end{figure}

\alert{Supplemental Appendix D illustrates that the VAR under-coverage can be ameliorated by increasing the lag length beyond what is typically used in current applied practice, at the expense of higher variance, consistent with \cref{subsec:lags}. The supplement also reports that the VAR estimator achieves lower MSE than LP, suggesting that---despite the poor performance of the VAR confidence interval---the larger bias of the VAR estimator may not compromise its usefulness as a \emph{point estimator} \citep[see also][]{Li2024}.}

\alert{\citet{MPQW2025} find that the qualitative conclusions above extend to a wide range of empirically calibrated simulation DGPs based on richly specified dynamic factor models: VAR confidence intervals fail to adequately control coverage in a sizable fraction of DGPs, while LP confidence intervals robustly maintain accurate coverage. They document that LP is not only more robust to lag length selection, but also to the choice of control variables, consistent with the theory in this paper.}

\section{Conclusion}
\label{sec:conc}

Our theoretical results suggest the following practical take-aways:

\begin{enumerate}[1.]

\item When the goal is to construct confidence intervals for impulse responses that have accurate coverage in a wide range of empirically relevant DGPs---as opposed to minimizing MSE---then the smaller bias of LPs documented in simulations by \citet{Li2024} is more valuable than the smaller variance enjoyed by VAR estimators.

\item Researchers who use LP should control for those lags of the data that are strong predictors of the outcome or impulse variables. This is important \alert{not only when the shock is recursively identified, but} even if the researcher directly observes a near-perfect proxy for the shock of interest. However, \alert{unlike for VAR inference}, it is \emph{not} necessary to get the lag length or \alert{set of control variables} exactly right to achieve correct coverage. To select the number of lags to control for in the LP, \alert{we recommend running an auxiliary VAR in all variables used in the analysis and selecting the lag length to minimize the AIC; the auxiliary VAR is only used as a device to select the lag length and is otherwise discarded.}

\item \alert{With the moderate lag lengths typical in current applied practice, VAR confidence intervals will only have accurate coverage at short horizons, and only because they are approximately equivalent with LP intervals at these horizons. If, at some horizons of interest,} an estimated VAR yields confidence intervals that are substantially narrower than the corresponding LP intervals, we recommend increasing the VAR lag length until that is no longer the case, to guarantee robust confidence interval coverage. Conventional tests of correct VAR specification do not suffice to guard against coverage distortions.

\end{enumerate}

Is there a way forward for VAR inference, beyond just including a large number of lags? We showed how to construct a VAR confidence interval with a bias-aware critical value that robustly controls coverage, but found that it will typically lead to wider confidence intervals than LP. Another option would be to estimate VARMA models rather than pure VARs, though this would be computationally expensive, and the bias-variance trade-off relative to LPs is unclear. In principle, VAR procedures may work better under additional restrictions on the misspecification, such as shape restrictions on the impulse response functions.\footnote{Given any convex parameter space for the misspecification MA polynomial $\alpha(L)$, the worst-case bias of the VAR estimator (see \cref{thm:var}) can be computed using convex programming.} However, it appears that detailed application-specific restrictions would be required to generate a negligible worst-case bias, since we have shown that the least favorable misspecification in our baseline analysis \alert{cannot generally be ruled out based on economic theory}. Rather than restricting the parameter space, future research could instead investigate weakening the coverage requirement, e.g., only requiring a certain coverage probability \emph{on average} over a set of horizons \citep{Armstrong2022}, or by changing the target for inference from the true impulse response function to a smooth projection of this function \citep{Genovese2008}. Finally, a subjectivist Bayesian VAR modeler need only worry about our negative results if their prior on potential misspecification attaches significant weight to MA processes that imply large VAR biases.

\appendix
\begin{appendices}
\numberwithin{equation}{section}
\numberwithin{figure}{section}
\numberwithin{lem}{section}
\numberwithin{prop}{section}
\numberwithin{cor}{section}

\section{Further theoretical results}
\label{app:further_results}

\subsection{Least favorable misspecification}
\label{app:least_fav}

\cref{fig:worstcase_ar1} plots some numerical examples of the least favorable MA polynomial $\alpha^\dagger(L;h,M)=\sum_{\ell=1}^\infty \alpha^\dagger_{\ell,h,M} L^\ell$ discussed in \cref{sec:var_coverage}. We focus here on a univariate local-to-AR(1) model $y_t = \rho y_{t-1} + [1+T^{-1/2}\alpha(L)]\varepsilon_t$, though unreported numerical experiments suggest that the qualitative features mentioned below also apply to multivariate models. Recall that the least favorable MA coefficients depend on the horizon $h$ of interest, while $M$ only influences the overall scale of the coefficients, and not their shape as a function of $\ell$. The figure shows that the shape of the coefficients either takes the form of a hump or of a single zig-zag pattern, with the largest absolute value of the coefficients generally occurring at $\ell=h$. Notice that we can flip the signs of all coefficients without changing the absolute value of the bias.

\begin{figure}[t]
\centering \vspace{0.5\baselineskip}
\includegraphics[width=\textwidth]{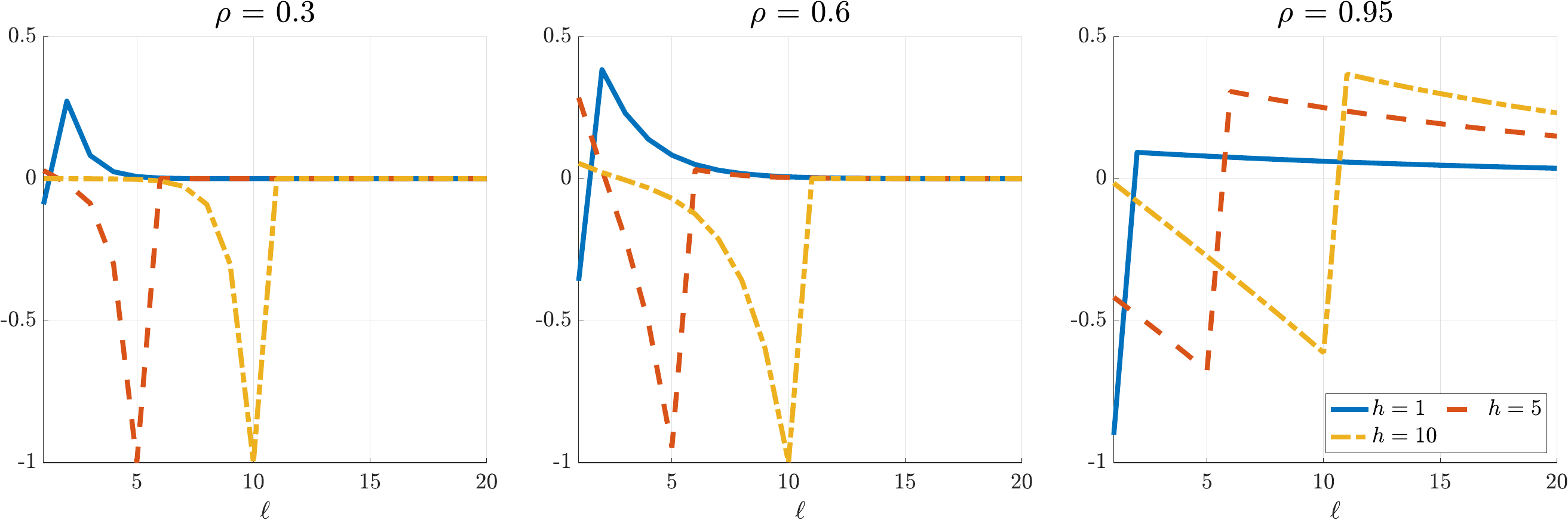}
\caption{Least favorable $\alpha^\dagger(L;h)$ for horizons $h \in \{ 1, 5, 10 \}$ for local-to-AR(1) models with different persistence parameters $\rho$ (left, middle, and right panel).}
\label{fig:worstcase_ar1}
\end{figure}

\subsection{More efficient bias-aware confidence interval}
\label{app:biasaware_opt}

Generalizing the bias-aware VAR confidence interval in \cref{sec:var_biasaware}, consider a bias-aware confidence interval that is centered at the model averaging estimator $\hat{\theta}_h(\omega)=\omega \hat{\beta}_h+(1-\omega)\hat{\delta}_h$ from \cref{thm:model_avg}:
\[\ci_B(\hat{\theta}_h(\omega);M) \equiv \left[\hat{\theta}_h(\omega) \pm \cv_{1-a}\left(\frac{(1-\omega)M\tau}{\sqrt{1+\omega^2\tau^2}}\right)\sqrt{(1+\omega^2\tau^2) \avar(\hat{\delta}_h)/T}\right],\]
where $\tau \equiv \sqrt{\avar(\hat{\beta}_h)/\avar(\hat{\delta}_h)-1}$. This interval equals the conventional LP interval when $\omega=1$ and the bias-aware VAR interval when $\omega=0$.

\begin{cor} \label{thm:coverage_biasaware_opt}
	Impose \cref{asn:model}, $\zeta=1/2$, and $\avar(\hat{\delta}_h)>0$. Then, for any $\omega \in [0,1]$,
	\[\inf_{\alpha(L) \colon \|\alpha(L)\|\leq M}\lim_{T\to\infty} P(\theta_{h,T} \in \ci_B(\hat{\theta}_h(\omega);M)) = 1-a.\]
\end{cor}
\begin{proof}
The result follows from \cref{thm:lp,thm:var,thm:bias_bound}, \cref{thm:asy_var}, and the same calculations as in the proof of \cref{thm:model_avg}.
\end{proof}

Even if we choose the weight $\omega$ to minimize confidence interval length, the resulting bias-aware interval tends to be nearly as long as the LP interval. The length-optimal weight $\omega=\omega^*$ is given by
\[\omega^* \equiv \argmin_{\omega \in [0,1]}\; \cv_{1-a}\left(\frac{(1-\omega)M\tau}{\sqrt{1+\omega^2\tau^2}}\right)\sqrt{1+\omega^2\tau^2}.\]
\cref{fig:weight_opt} shows this optimal weight as a function of $M$ and the relative asymptotic standard deviation of the VAR and LP estimators, while \cref{fig:rel_length_opt} shows the length of the resulting optimal bias-aware confidence interval relative to the length of the conventional LP interval. We see that, for $M \geq 2$, there is little gain from reporting the optimal bias-aware interval rather than the LP interval, regardless of the relative precision of VAR and LP. An additional observation is that, for $M \geq 1.5$, the length-optimal $\omega^*$ is numerically close to the MSE-optimal weight $M^2/(1+M^2)$ derived in \cref{thm:model_avg}.

\begin{figure}[p]
	\centering
	\includegraphics[width=0.75\linewidth]{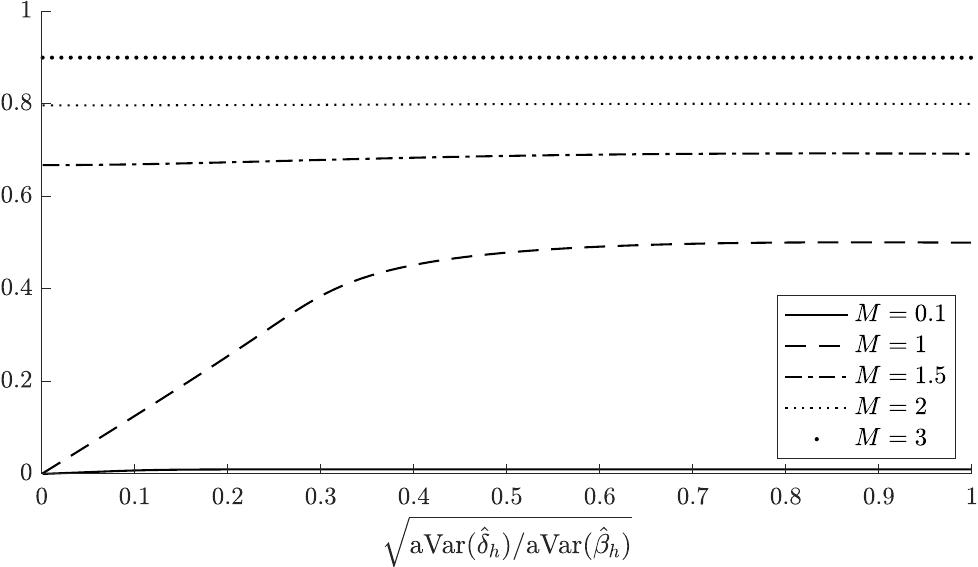}
	\caption{Length-optimal weight on LP in bias-aware confidence interval. Significance level $a=10\%$. Horizontal axis: relative asymptotic standard deviation of VAR vs.\ LP. Different lines: different bounds $M$ on $\|\alpha(L)\|$.} \label{fig:weight_opt}
	
	\bigskip
	
	\includegraphics[width=0.75\linewidth]{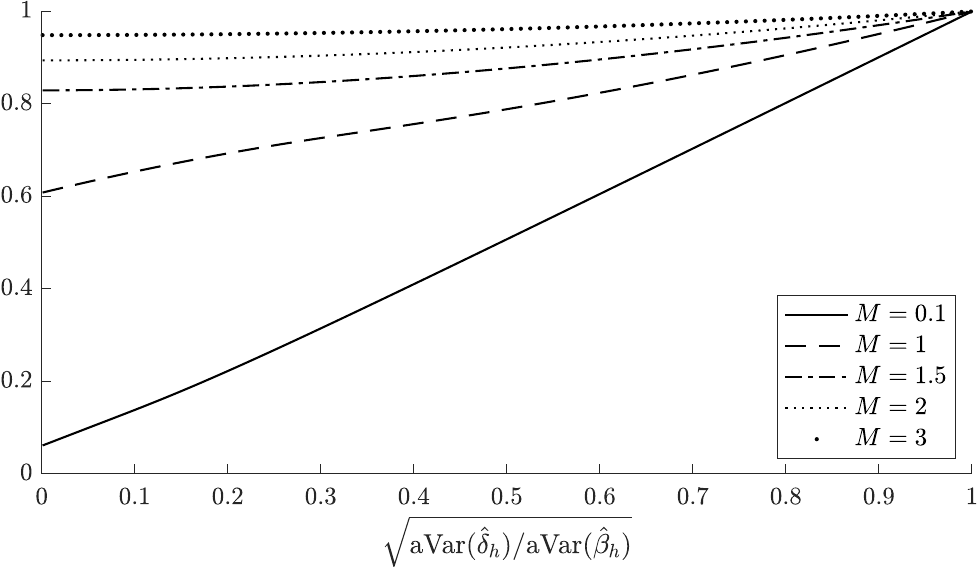}
	\caption{Relative length of optimal bias-aware confidence interval vs.\ conventional LP interval. Significance level $a=10\%$. Horizontal axis: relative asymptotic standard deviation of VAR vs.\ LP. Different lines: different bounds $M$ on $\|\alpha(L)\|$.} \label{fig:rel_length_opt}
	
\end{figure}

\subsection{Covariance structure of LP and VAR estimators}
\label{app:cov}

The following result provides the asymptotic variance-covariance matrix of the LP and VAR estimators in the general multi-dimensional set-up of \cref{sec:joint}. Define $\Psi_{i^*,j^*,h}$ as in \cref{thm:var}, but making the dependence on $(i^*,j^*)$ explicit in the notation.

\begin{cor} \label{thm:asy_var}
Impose \cref{asn:model}, with part (\ref{itm:asn_triang}) holding for all shock indices $j_1^*,\dots,j_k^*$. Then for any $a,b \in \lbrace 1,\dots,k \rbrace$,
\begin{align*}
\acov(\hat{\beta}_{i_a^*,j_a^*,h_a},\hat{\beta}_{i_b^*,j_b^*,h_b}) &= \mathbbm{1}(j_a^*=j_b^*) \sigma_{j_a^*}^{-2} \left(\psi_{a,b} + \sum_{\ell=1}^{\min\lbrace h_a,h_b\rbrace} e_{i_a^*,n}'A^{h_a-\ell}\Sigma(A')^{h_b-\ell}e_{i_b^*,n}\right), \\
\acov(\hat{\delta}_{i_a^*,j_a^*,h_a},\hat{\delta}_{i_b^*,j_b^*,h_b}) &= \mathbbm{1}(j_a^*=j_b^*) \sigma_{j_a^*}^{-2} \psi_{a,b} + \tr\left(\Psi_{i_a^*,j_a^*,h_a}\Sigma \Psi_{i_b^*,j_b^*,h_b}' S^{-1}\right), \\
\acov(\hat{\beta}_{i_a^*,j_a^*,h_a},\hat{\delta}_{i_b^*,j_b^*,h_b}) &= \acov(\hat{\delta}_{i_a^*,j_a^*,h_a},\hat{\delta}_{i_b^*,j_b^*,h_b}),
\end{align*}
where
\[\psi_{a,b} \equiv e_{i_a^*,n}'A^{h_a}\overline{H}_{j_a^*}\overline{D}_{j_a^*}\overline{H}_{j_a^*}'(A')^{h_b}e_{i_b^*,n},\]
and the ``$\acov$'' notation refers to elements of the asymptotic variance-covariance matrix in Equation (C.2) in Supplemental Appendix C.4. In particular, $\acov(\hat{\beta}_{i_a^*,j_a^*,h_a}-\hat{\delta}_{i_a^*,j_a^*,h_a},\hat{\delta}_{i_b^*,j_b^*,h_b})=0$.
\end{cor}
\begin{proof}
See \cref{subsection:proof_asy_var}.
\end{proof}

\section{Proofs}
\label{app:proofs}

Lemmas whose name begins with ``E'' can be found in Supplemental Appendix E.

\subsection{Proof of \texorpdfstring{\cref{thm:lp}}{Proposition \ref{thm:lp}}} \label{subsection:Proof_thm_lp}
Lemma E.1 shows that we can represent
\begin{equation} \label{eqn:lp_repr}
y_{i^*,t+h} = \theta_{h,T}\varepsilon_{j^*,t} + \underline{B}_{h,y}'\underline{y}_{j^*,t} + B_{h,y}' y_{t-1} + \xi_{i^*,h,t} + T^{-\zeta} \Theta_h(L)\varepsilon_t,
\end{equation}
where the expressions for the coefficient matrices and the $1 \times n$ two-sided lag polynomial $\Theta_h(L)=\sum_{\ell=-\infty}^\infty \Theta_{h,\ell} L^\ell$ are given in Lemma E.1.

Let $\hat{x}_{h,t}$ be the residual in a regression of $y_{j^*,t}$ on $\underline{y}_{j^*,t}$ and $y_{t-1}$, using data points $1,2,\dots,T-h$. By definition, $\hat{x}_{h,t}$ is in-sample orthogonal to $\underline{y}_{j^*,t}$ and $y_{t-1}$. Hence,
\begin{align*}
	\hat{\beta}_h &= \frac{\sum_{t=1}^{T-h}y_{i^*,t+h}\hat{x}_{h,t}}{\sum_{t=1}^{T-h}\hat{x}_{h,t}^2} \\
	&= \theta_{h,T} + \frac{\sum_{t=1}^{T-h}(y_{i^*,t+h}-\theta_{h,T}\hat{x}_{h,t} - \underline{B}_{h,y}'\underline{y}_{j^*,t} - B_{h,y}' y_{t-1})\hat{x}_{h,t}}{\sum_{t=1}^{T-h}\hat{x}_{h,t}^2} \quad \text{by orthogonality} \\
	&= \theta_{h,T} + \frac{T^{-1}\sum_{t=1}^{T-h}(y_{i^*,t+h}-\theta_{h,T}\hat{x}_{h,t} - \underline{B}_{h,y}'\underline{y}_{j^*,t} - B_{h,y}' y_{t-1})\hat{x}_{h,t}}{\sigma_{j^*}^2 + o_p(1)} \quad \text{by Lemma E.4(v)}\\
	&= \theta_{h,T} + \frac{T^{-1}\sum_{t=1}^{T-h}(y_{i^*,t+h}-\theta_{h,T}\varepsilon_{j^*,t} - \underline{B}_{h,y}'\underline{y}_{j^*,t} - B_{h,y}' y_{t-1})\hat{x}_{h,t} + O_p(T^{-2\zeta}) +o_p(T^{-1/2})}{\sigma_{j^*}^2 + o_p(1)} \\
	&\;\quad \text{by Lemma E.4(iv)} \\
	&= \theta_{h,T} + \frac{T^{-1}\sum_{t=1}^{T-h}(\xi_{i^*,h,t} + T^{-\zeta} \Theta_h(L)\varepsilon_t)\hat{x}_{h,t} + O_p(T^{-2\zeta}) + o_p(T^{-1/2})}{\sigma_{j^*}^2 + o_p(1)} \quad \text{by \eqref{eqn:lp_repr}} \\
	&= \theta_{h,T} + \frac{T^{-1}\sum_{t=1}^{T-h}(\xi_{i^*,h,t} + T^{-\zeta} \Theta_h(L)\varepsilon_t)\varepsilon_{j^*,t} + O_p(T^{-2\zeta}) + o_p(T^{-1/2})}{\sigma_{j^*}^2 + o_p(1)} \\
	&\;\quad \text{by Lemma E.4(iii) and (vi)} \\
	&= \theta_{h,T} + \frac{T^{-1}\sum_{t=1}^{T-h}\xi_{i^*,h,t}\varepsilon_{j^*,t} + O_p(T^{-2\zeta}) + o_p(T^{-1/2})}{\sigma_{j^*}^2 + o_p(1)} \quad \text{by Lemma E.1},
\end{align*}
and the result follows. \qed

\subsection{Proof of \texorpdfstring{\cref{thm:var}}{Proposition \ref{thm:var}}} \label{subsection:Proof_thm_var}
Note first that
\begin{align*}
	\hat{\delta}_h-e_{i^*,n}'A^h H_{\bullet, j^*} &= e_{i^*,n}'\hat{A}^h\hat{\nu}-e_{i^*,n}'A^h H_{\bullet, j^*} \\
	&= e_{i^*,n}'\hat{A}^h H_{\bullet, j^*} - e_{i^*,n}'A^h H_{\bullet, j^*}  + e_{i^*,n}'\hat{A}^h (\hat{\nu}-H_{\bullet, j^*}). 
\end{align*}
Lemma E.2 shows that $\hat{A}-A = O_p(T^{-\zeta}+T^{-1/2})$. By \citet[Table 7, p.\ 208]{Magnus2007},
\[\left( \frac{\partial (e_{i^*,n}'A^h H_{\bullet, j^*})}{\partial \ve(A)} \right)' = (H_{\bullet, j^*}' \otimes e_{i^*,n}')\left(\sum_{\ell=1}^h (A')^{h-\ell} \otimes A^{\ell-1}\right) = \sum_{\ell=1}^h H_{\bullet, j^*}'(A')^{h-\ell} \otimes e_{i^*,n}'A^{\ell-1},\]
so
\begin{align*}
	\hat{\delta}_h-e_{i^*,n}'A^h H_{\bullet, j^*} 
	&= \left(\sum_{\ell=1}^h H_{\bullet, j^*}'(A')^{h-\ell} \otimes e_{i^*,n}'A^{\ell-1}\right)\ve(\hat{A}-A) + e_{i^*,n}'A^h(\hat{\nu}-H_{\bullet, j^*}) + o_p(T^{-\zeta}+T^{-1/2}) \\
	&= \sum_{\ell=1}^h e_{i^*,n}'A^{\ell-1}(\hat{A}-A)A^{h-\ell}H_{\bullet, j^*} + e_{i^*,n}'A^h(\hat{\nu}-H_{\bullet, j^*}) + o_p(T^{-\zeta}+T^{-1/2})\\
	&= \tr\left\lbrace \Psi_h (\hat{A}-A)\right\rbrace + e_{i^*,n}'A^h(\hat{\nu}-H_{\bullet, j^*}) + o_p(T^{-\zeta}+T^{-1/2}),
\end{align*}
where  $\Psi_h \equiv \sum_{\ell=1}^h A^{h-\ell}H_{\bullet, j^*} e_{i^*,n}'A^{\ell-1}$. Lemma E.2 further implies that
 \begin{align*}
 \tr\left\lbrace \Psi_h (\hat{A}-A) \right\rbrace &= T^{-\zeta}  \tr\left\lbrace S^{-1} \Psi_h H\sum_{\ell=1}^{\infty} \alpha_{\ell}DH'(A')^{\ell-1} \right\rbrace  \\
 &+ \tr\left\lbrace S^{-1} \Psi_h HT^{-1}\sum_{t=1}^T \varepsilon_t \tilde{y}_{t-1}'\right\rbrace + o_p(T^{-\zeta}),
 \end{align*}
where $S$ was defined in \cref{asn:model}. Lemma E.3 shows that
\[\hat{\nu}-H_{\bullet, j^*} = \frac{1}{\sigma_{j^*}^2}T^{-1}\sum_{t=1}^T \xi_{0,t} \varepsilon_{j^*,t} + o_p(T^{-\zeta}+T^{-1/2}).\]
Using the definition of $\theta_{h,T}$ and re-arranging terms gives the desired result. \qed

\subsection{Proof of \texorpdfstring{\cref{thm:lag}}{Corollary \ref{thm:lag}}} \label{subsection:proof_thm_lag}
Use the notation $E^*(z \mid w) = \cov(z,w)\var(w)^{-1}w$ for mean-square projection. Then
\begin{align*}
	\sigma_{j^*}^2 \Psi_h' S^{-1}\tilde{y}_{t-1} &= \bigg(\sum_{\ell=1}^h (A')^{h-\ell} e_{i^*,n} \underbrace{\sigma_{j^*}^2 H_{\bullet,j^*}'(A')^{\ell-1}}_{=\cov(\varepsilon_{j^*,t-\ell},\tilde{y}_{t-1})}\bigg) S^{-1}\tilde{y}_{t-1} \\
	&= \sum_{\ell=1}^h (A')^{h-\ell} e_{i^*,n} E^*(\varepsilon_{j^*,t-\ell} \mid \tilde{y}_{t-1}) \\
	&= \sum_{\ell=1}^h (A')^{h-\ell} e_{i^*,n} \varepsilon_{j^*,t-\ell},
\end{align*}
where the last equality uses $\varepsilon_{j^*,t-\ell} \in \sspan(\tilde{\check{y}}_{t-1},\dots,\tilde{\check{y}}_{t-p})$ for $\ell=1,\dots,h$. Thus,
\begin{align*}
\var(\varepsilon_t'H' \Psi_h' S^{-1}\tilde{y}_{t-1}) &= \var\left(\frac{1}{\sigma_{j^*}^2}\sum_{\ell=1}^h \varepsilon_{j^*,t-\ell}\varepsilon_t'H' (A')^{h-\ell} e_{i^*,n}\right) \\
&= \frac{1}{\sigma_{j^*}^4}\sum_{\ell=1}^h\var\left(\varepsilon_{j^*,t-\ell}\varepsilon_t'H' (A')^{h-\ell} e_{i^*,n}\right) \\
&= \frac{1}{\sigma_{j^*}^4} \sum_{\ell=1}^h E(\varepsilon_{j^*,t-\ell}^2) \var(\varepsilon_t'H' (A')^{h-\ell} e_{i^*,n}) \\
&= \frac{1}{\sigma_{j^*}^2} \var\left(e_{i^*,n}'\sum_{\ell=1}^h A^{h-\ell}H\varepsilon_{t+\ell} \right).
\end{align*}
It now follows as in the proof of \cref{thm:asy_var} that $\avar(\hat{\beta}_h)=\avar(\hat{\delta}_h)$. Then \cref{thm:bias_bound} implies that $\abias(\hat{\delta}_h)=0$. \qed

\subsection{Proof of \texorpdfstring{\cref{thm:mse}}{Corollary \ref{thm:mse}}} \label{subsection:proof_thm_mse}
By \cref{thm:bias_bound}, $\sup_{\alpha(L) \colon \|\alpha(L)\|\leq M} \abias(\hat{\delta}_h;\alpha(L))^2 = M^2\lbrace\avar(\hat{\beta}_h)-\avar(\hat{\delta}_h)\rbrace$. The result follows. \qed

\subsection{Proof of \texorpdfstring{\cref{thm:model_avg}}{Corollary \ref{thm:model_avg}}} \label{subsection:proof_thm_model_avg}
Write $\hat{\theta}_h(\omega) = \hat{\delta}_h + \omega(\hat{\beta}_h-\hat{\delta}_h)$. By \cref{thm:asy_var}, the two terms are asymptotically independent of each other, and the second term has asymptotic variance $\omega^2\lbrace \avar(\hat{\beta}_h)-\avar(\hat{\delta}_h) \rbrace$. Hence,
\[\aMSE(\hat{\theta}_h(\omega)) = \lbrace (1-\omega)\abias(\hat{\delta}_h)\rbrace^2 + \avar(\hat{\delta}_h) + \omega^2\lbrace\avar(\hat{\beta}_h)-\avar(\hat{\delta}_h)\rbrace.\]
By \cref{thm:bias_bound}, the supremum of the above expression over $\alpha(L)$ satisfying $\|\alpha(L)\| \leq M$ equals
\[(1-\omega)^2 M^2\lbrace\avar(\hat{\beta}_h)-\avar(\hat{\delta}_h)\rbrace + \avar(\hat{\delta}_h) + \omega^2\lbrace\avar(\hat{\beta}_h)-\avar(\hat{\delta}_h)\rbrace.\]
To find the $\omega$ that minimizes the above expression, we can equivalently minimize the function $(1-\omega)^2M^2+\omega^2$. The result follows. \qed

\subsection{Proof of \texorpdfstring{\cref{thm:jointprob_worstcase}}{Corollary \ref{thm:jointprob_worstcase}}} \label{subsection:proof_thm_jointprob_worstcase}
\cref{thm:bias_bound} implies that the absolute relative VAR bias $|b_h|$ can be made to take any value in $[0,\infty)$ as $\alpha(L)$ varies over the set of all absolutely summable lag polynomials. The corollary then follows from \cref{thm:coverage,thm:hausman,thm:asy_var}. \qed

\subsection{Proof of \texorpdfstring{\cref{thm:asy_var}}{Corollary A.2}}
\label{subsection:proof_asy_var}
We first use \cref{thm:lp} to compute $\acov(\hat{\beta}_{i_a^*,j_a^*,h_a},\hat{\beta}_{i_b^*,j_b^*,h_b})$. Define $\xi_{j^*,h,t} = (\xi_{1,j^*,h,t},\dots,\xi_{n,j^*,h,t})'$ as in \cref{thm:lp}, but making the dependence on both $i^*$ and $j^*$ explicit in the notation. Observe that
\[E[\xi_{i_a^*,j_a^*,h_a,t}\varepsilon_{j_a^*,t} \xi_{i_b^*,j_b^*,h_b,s}\varepsilon_{j_b^*,s}]=0 \quad \text{for all } s \neq t.\]
Hence,
\[\acov(\hat{\beta}_{i_a^*,j_a^*,h_a},\hat{\beta}_{i_b^*,j_b^*,h_b}) = \frac{1}{\sigma_{j_a^*}^2\sigma_{j_b^*}^2} E[\xi_{i_a^*,j_a^*,h_a,t}\varepsilon_{j_a^*,t} \xi_{i_b^*,j_b^*,h_b,t}\varepsilon_{j_b^*,t}].\]
If $j_a^* < j_b^*$, then $\varepsilon_{j_a^*,t}$ is independent of all the other terms in the above expectation, so the expectation equals zero; similarly if $j_a^* > j_b^*$. Now consider the case $j_a^*=j_b^*$:
\begin{align*}
\acov(\hat{\beta}_{i_a^*,j_a^*,h_a},\hat{\beta}_{i_b^*,j_a^*,h_b}) &= \frac{1}{\sigma_{j_a^*}^4} E[\xi_{i_a^*,j_a^*,h_a,t} \xi_{i_b^*,j_a^*,h_b,t}\varepsilon_{j_a^*,t}^2] \\
&= \frac{1}{\sigma_{j_a^*}^4} E[\xi_{i_a^*,j_a^*,h_a,t} \xi_{i_b^*,j_a^*,h_b,t}]E[\varepsilon_{j_a^*,t}^2] \\
&= \frac{1}{\sigma_{j_a^*}^2} E[\xi_{i_a^*,j_a^*,h_a,t} \xi_{i_b^*,j_a^*,h_b,t}] \\
&= \frac{1}{\sigma_{j_a^*}^2} \Bigg(E[e_{i_a^*,n}'A^{h_a}\overline{H}_{j_a^*}\overline{\varepsilon}_{j_a^*,t}\overline{\varepsilon}_{j_a^*,t}'\overline{H}_{j_a^*}'(A')^{h_b}e_{i_b^*,n}] \\
&\qquad\qquad + E\left[e_{i_a^*,n}'\sum_{\ell_1=1}^{h_a}\sum_{\ell_2=1}^{h_b}A^{h_a-\ell_1}H\varepsilon_{t+\ell_1}\varepsilon_{t+\ell_2}'H'(A')^{h_b-\ell_2} e_{i_b^*,n}\right] \Bigg) \\
&= \frac{1}{\sigma_{j_a^*}^2}\left(\psi_{a,b} + \sum_{\ell=1}^{\min\lbrace h_a,h_b\rbrace} e_{i_a^*,n}'A^{h_a-\ell}\Sigma (A')^{h_b-\ell}e_{i_b^*,n} \right),
\end{align*}
as claimed.

We now derive $\acov(\hat{\delta}_{i_a^*,j_a^*,h_a},\hat{\delta}_{i_b^*,j_b^*,h_b})$ using \cref{thm:var}. Observe that the vector process $(\varepsilon_t' \otimes \tilde{y}_{t-1}',\xi_{j_a^*,0,t}'\varepsilon_{j_a^*,t},\xi_{j_b^*,0,t}'\varepsilon_{j_b^*,t})'$ is a martingale difference sequence with respect to the filtration generated by $\lbrace \varepsilon_t \rbrace$. Moreover, $E[(\varepsilon_t \otimes \tilde{y}_{t-1})\xi_{j^*,0,t}'\varepsilon_{j^*,t}]=0$ for any $j^*$. Hence,
\begin{align*}
\acov(\hat{\delta}_{i_a^*,j_a^*,h_a},\hat{\delta}_{i_b^*,j_b^*,h_b}) &= E\left[\tr\left(S^{-1}\Psi_{i_a^*,j_a^*,h_a} H \varepsilon_t \tilde{y}_{t-1}'\right)\tr\left(S^{-1}\Psi_{i_b^*,j_b^*,h_b} H \varepsilon_t \tilde{y}_{t-1}'\right) \right] \\
&\qquad + \frac{1}{\sigma_{j_a^*}^2\sigma_{j_b^*}^2}E\left[e_{i_a^*,n}'A^{h_a}\xi_{j_a^*,0,t}\xi_{j_b^*,0,t}'(A')^{h_b}e_{i_b^*,n}\varepsilon_{j_a^*,t}\varepsilon_{j_b^*,t} \right].
\end{align*}
The second term on the right-hand side above equals $\mathbbm{1}(j_a^*=j_b^*)\sigma_{j_a^*}^{-2}\psi_{a,b}$, by similar arguments as in the earlier LP calculation. The first term on the right-hand side above equals
\begin{align*}
&E\left[\tilde{y}_{t-1}'S^{-1}\Psi_{i_a^*,j_a^*,h_a} H \varepsilon_t \varepsilon_t' H' \Psi_{i_b^*,j_b^*,h_b}' S^{-1} \tilde{y}_{t-1}  \right] \\
&= \tr\left(E\left[\tilde{y}_{t-1}\tilde{y}_{t-1}'S^{-1}\Psi_{i_a^*,j_a^*,h_a} H \varepsilon_t \varepsilon_t' H' \Psi_{i_b^*,j_b^*,h_b}' S^{-1}  \right] \right) \\
&= \tr\left(E\left[\tilde{y}_{t-1}\tilde{y}_{t-1}'\right]S^{-1}\Psi_{i_a^*,j_a^*,h_a} H E\left[\varepsilon_t \varepsilon_t'\right] H' \Psi_{i_b^*,j_b^*,h_b}' S^{-1} \right) \\
&= \tr\left(S S^{-1}\Psi_{i_a^*,j_a^*,h_a} H D H' \Psi_{i_b^*,j_b^*,h_b}' S^{-1} \right) \\
&= \tr\left(\Psi_{i_a^*,j_a^*,h_a} \Sigma \Psi_{i_b^*,j_b^*,h_b}' S^{-1}\right),
\end{align*}
as claimed.

Finally, we compute $\acov(\hat{\beta}_{i_a^*,j_a^*,h_a},\hat{\delta}_{i_b^*,j_b^*,h_b})$ using \cref{thm:lp,thm:var}. Using arguments similar to above, we obtain
\begin{align*}
&\acov(\hat{\beta}_{i_a^*,j_a^*,h_a},\hat{\delta}_{i_b^*,j_b^*,h_b}) \\
&= \frac{1}{\sigma_{j_a^*}^2} \sum_{s=-\infty}^\infty E\left[e_{i_a^*,n}'\sum_{\ell=1}^{h_a} A^{h_a-\ell}H\varepsilon_{t+\ell} \varepsilon_{j_a^*,t} \tr\left(S^{-1}\Psi_{i_b^*,j_b^*,h_b} H \varepsilon_{t+s} \tilde{y}_{t+s-1}'\right) \right] + \mathbbm{1}(j_a^*=j_b^*)\sigma_{j_a^*}^{-2}\psi_{a,b}.
\end{align*}
The first term on the left-hand side above equals
\begin{align*}
&\frac{1}{\sigma_{j_a^*}^2}\sum_{\ell=1}^{h_a} \sum_{s=-\infty}^\infty E\left[ e_{i_a^*,n}' A^{h_a-\ell}H\varepsilon_{t+\ell}\varepsilon_{t+s}' H' \Psi_{i_b^*,j_b^*,h_b}' S^{-1} \tilde{y}_{t+s-1}\varepsilon_{j_a^*,t} \right] \\
&= \frac{1}{\sigma_{j_a^*}^2}\sum_{\ell=1}^{h_a} E\left[e_{i_a^*,n}' A^{h_a-\ell}H\varepsilon_{t+\ell}\varepsilon_{t+\ell}' H' \Psi_{i_b^*,j_b^*,h_b}' S^{-1} \tilde{y}_{t+\ell-1}\varepsilon_{j_a^*,t} \right] \\
&= \frac{1}{\sigma_{j_a^*}^2}\sum_{\ell=1}^{h_a} e_{i_a^*,n}' A^{h_a-\ell}HE[\varepsilon_{t+\ell}\varepsilon_{t+\ell}'] H' \Psi_{i_b^*,j_b^*,h_b}' S^{-1} E[\tilde{y}_{t+\ell-1}\varepsilon_{j_a^*,t}] \\
&= \frac{1}{\sigma_{j_a^*}^2}\sum_{\ell=1}^{h_a} e_{i_a^*,n}' A^{h_a-\ell}HD H' \Psi_{i_b^*,j_b^*,h_b}' S^{-1} A^{\ell-1}H_{\bullet,j_a^*}\sigma_{j_a^*}^2 \\
&= \tr\bigg( \underbrace{\sum_{\ell=1}^{h_a} A^{\ell-1}H_{\bullet,j_a^*} e_{i_a^*,n}' A^{h_a-\ell}}_{=\sum_{\ell=1}^{h_a} A^{h_a-\ell}H_{\bullet,j_a^*} e_{i_a^*,n}' A^{\ell-1}=\Psi_{i_a^*,j_a^*,h_a}}\Sigma \Psi_{i_b^*,j_b^*,h_b}' S^{-1} \bigg).
\end{align*}
It follows that $\acov(\hat{\beta}_{i_a^*,j_a^*,h_a},\hat{\delta}_{i_b^*,j_b^*,h_b})=\acov(\hat{\delta}_{i_a^*,j_a^*,h_a},\hat{\delta}_{i_b^*,j_b^*,h_b})$, as claimed. \qed

\clearpage
\phantomsection
\addcontentsline{toc}{section}{References}
\bibliography{ref}

@incollection{Stock2016,
	author = {Stock, James H and Watson, Mark W},
	booktitle = {Handbook of Macroeconomics},
	pages = {415--525},
	publisher = {Elsevier},
	title = {Dynamic Factor Models, Factor-Augmented Vector Autoregressions, and Structural Vector Autoregressions in Macroeconomics},
	volume = {2},
	year = {2016},
	chapter = 8,
	editor = {John B. Taylor and Harald Uhlig}
}

@article{Kaenzig2021,
	author = {K{\"a}nzig, Diego R},
	journal = {American Economic Review},
	number = {4},
	pages = {1092--1125},
	publisher = {American Economic Association 2014 Broadway, Suite 305, Nashville, TN 37203},
	title = {The Macroeconomic Effects of Oil Supply News: Evidence from {OPEC} Announcements},
	volume = {111},
	year = {2021}}

@article{Inoue2002,
	author = {Inoue, Atsushi and Kilian, Lutz},
	journal = {International Economic Review},
	number = {2},
	pages = {309--331},
	publisher = {Wiley Online Library},
	title = {Bootstrapping Smooth Functions of Slope Parameters and Innovation Variances in {VAR($\infty$)} Models},
	volume = {43},
	year = {2002}}

@article{Chari2008,
	author = {Chari, Varadarajan V and Kehoe, Patrick J and McGrattan, Ellen R},
	journal = {Journal of Monetary Economics},
	number = {8},
	pages = {1337--1352},
	publisher = {Elsevier},
	title = {Are structural {VARs} with long-run restrictions useful in developing business cycle theory?},
	volume = {55},
	year = {2008}}

@article{Ai2007,
	author = {Chunrong Ai and Xiaohong Chen},
	doi = {https://doi.org/10.1016/j.jeconom.2007.01.013},
	issn = {0304-4076},
	journal = {Journal of Econometrics},
	note = {Special issue on ``Semiparametric methods in econometrics''},
	number = {1},
	pages = {5-43},
	title = {Estimation of possibly misspecified semiparametric conditional moment restriction models with different conditioning variables},
	url = {https://www.sciencedirect.com/science/article/pii/S0304407607000061},
	volume = {141},
	year = {2007}}

@article{Alessi2011,
	author = {Alessi, Lucia and Barigozzi, Matteo and Capasso, Marco},
	doi = {10.1111/j.1751-5823.2011.00131.x},
	journal = {International Statistical Review},
	number = {1},
	pages = {16--47},
	title = {Non-Fundamentalness in Structural Econometric Models: A Review},
	volume = {79},
	year = {2011}}

@article{Armstrong2021,
	author = {Armstrong, Timothy B. and Koles\'{a}r, Michal},
	doi = {https://doi.org/10.3982/QE1609},
	eprint = {https://onlinelibrary.wiley.com/doi/pdf/10.3982/QE1609},
	journal = {Quantitative Economics},
	number = {1},
	pages = {77-108},
	title = {Sensitivity analysis using approximate moment condition models},
	url = {https://onlinelibrary.wiley.com/doi/abs/10.3982/QE1609},
	volume = {12},
	year = {2021}}

@article{Armstrong2022,
	author = {Armstrong, Timothy B. and Koles\'{a}r, Michal and Plagborg-M{\o}ller, Mikkel},
	doi = {10.3982/ECTA18597},
	issn = {0012-9682},
	journal = {Econometrica},
	language = {en},
	number = {6},
	pages = {2567--2602},
	title = {Robust Empirical {Bayes} Confidence Intervals},
	url = {https://www.econometricsociety.org/doi/10.3982/ECTA18597},
	urldate = {2022-11-15},
	volume = {90},
	year = {2022}}

@article{Bang2005,
	author = {Heejung Bang and James M. Robins},
	issn = {0006341X, 15410420},
	journal = {Biometrics},
	number = {4},
	pages = {962--972},
	publisher = {[Wiley, International Biometric Society]},
	title = {Doubly Robust Estimation in Missing Data and Causal Inference Models},
	url = {http://www.jstor.org/stable/3695907},
	urldate = {2025-10-03},
	volume = {61},
	year = {2005}}

@article{Belloni2013,
	author = {Alexandre Belloni and Victor Chernozhukov},
	journal = {Bernoulli},
	number = 2,
	pages = {521-547},
	title = {Least squares after model selection in high-dimensional sparse models},
	volume = 19,
	year = 2013}

@article{Braun1993,
	author = {Phillip A. Braun and Stefan Mittnik},
	journal = {Journal of Econometrics},
	number = 3,
	pages = {319-341},
	title = {Misspecifications in vector autoregressions and their effects on impulse responses and variance decompositions},
	volume = 59,
	year = 1993}

@book{Brockwell1991,
	author = {Brockwell, Peter J. and Davis, Richard A.},
	doi = {10.1007/978-1-4419-0320-4},
	edition = {2nd},
	isbn = {978-1-4419-0319-8},
	publisher = {Springer},
	series = {Springer Series in Statistics},
	title = {Time Series: Theory and Methods},
	url = {http://link.springer.com/10.1007/978-1-4419-0320-4},
	year = {1991}}

@article{Chernozhukov2018,
	author = {Chernozhukov, Victor and Chetverikov, Denis and Demirer, Mert and Duflo, Esther and Hansen, Christian and Newey, Whitney and Robins, James},
	doi = {10.1111/ectj.12097},
	journal = {The Econometrics Journal},
	month = feb,
	number = {1},
	pages = {C1--C68},
	title = {Double/debiased machine learning for treatment and structural parameters},
	url = {http://doi.wiley.com/10.1111/ectj.12097},
	volume = {21},
	year = {2018}}

@article{Chernozhukov2022,
	author = {Chernozhukov, Victor and Escanciano, Juan Carlos and Ichimura, Hidehiko and Newey, Whitney K. and Robins, James M.},
	doi = {10.3982/ECTA16294},
	issn = {0012-9682},
	journal = {Econometrica},
	language = {en},
	number = {4},
	pages = {1501--1535},
	title = {Locally Robust Semiparametric Estimation},
	url = {https://www.econometricsociety.org/doi/10.3982/ECTA16294},
	urldate = {2023-03-08},
	volume = {90},
	year = {2022}}

@article{Genovese2008,
	author = {Christopher Genovese and Larry Wasserman},
	doi = {10.1214/07-AOS500},
	journal = {The Annals of Statistics},
	number = {2},
	pages = {875 -- 905},
	publisher = {Institute of Mathematical Statistics},
	title = {Adaptive confidence bands},
	url = {https://doi.org/10.1214/07-AOS500},
	volume = {36},
	year = {2008}}

@techreport{GonzalesCasasus2025,
	author = {Oriol Gonz\'{a}lez-Casas\'{u}s and Frank Schorfheide},
	institution = {National Bureau of Economic Research},
	number = 33474,
	title = {Misspecification-Robust Shrinkage and Selection for {VAR} Forecasts and {IRFs}},
	year = 2025}

@article{Granger1976,
	author = {C. W. J. Granger and M. J. Morris},
	issn = {00359238},
	journal = {Journal of the Royal Statistical Society. Series A (General)},
	number = {2},
	pages = {246--257},
	publisher = {[Royal Statistical Society, Wiley]},
	title = {Time Series Modelling and Interpretation},
	url = {http://www.jstor.org/stable/2345178},
	urldate = {2024-03-15},
	volume = {139},
	year = {1976}}

@article{Hallin1989,
	author = {Marc Hallin and Jean-Francois Ingenbleek and Madan L. Puri},
	doi = {https://doi.org/10.1016/0047-259X(89)90087-0},
	issn = {0047-259X},
	journal = {Journal of Multivariate Analysis},
	number = {1},
	pages = {34-71},
	title = {Asymptotically most powerful rank tests for multivariate randomness against serial dependence},
	url = {https://www.sciencedirect.com/science/article/pii/0047259X89900870},
	volume = {30},
	year = {1989}}

@article{Hallin1988,
	author = {Marc Hallin and Madan L. Puri},
	journal = {The Annals of Statistics},
	number = 1,
	pages = {402-432},
	title = {Optimal Rank-Based Procedures for Time Series Analysis: Testing an {ARMA} Model Against Other {ARMA} Models},
	volume = 16,
	year = 1988}

@article{Hausman1978,
	author = {Hausman, JA},
	journal = {Econometrica},
	number = {6},
	pages = {1251--1271},
	title = {Specification Tests in Econometrics},
	url = {http://www.jstor.org/stable/10.2307/1913827},
	volume = {46},
	year = {1978}}

@article{Herbst2023,
	author = {Edward P. Herbst and Benjamin K. Johannsen},
	doi = {https://doi.org/10.1016/j.jeconom.2024.105655},
	issn = {0304-4076},
	journal = {Journal of Econometrics},
	number = {1},
	pages = {105655},
	title = {Bias in local projections},
	url = {https://www.sciencedirect.com/science/article/pii/S0304407624000010},
	volume = {240},
	year = {2024}}

@article{Inoue2020,
	author = {Inoue, Atsushi and Kilian, Lutz},
	doi = {10.1016/j.jeconom.2019.10.001},
	issn = {03044076},
	journal = {Journal of Econometrics},
	language = {en},
	month = apr,
	number = {2},
	pages = {450--472},
	title = {The uniform validity of impulse response inference in autoregressions},
	url = {https://linkinghub.elsevier.com/retrieve/pii/S0304407619302040},
	urldate = {2020-03-15},
	volume = {215},
	year = {2020}}

@article{Jorda2005,
	author = {Jord\`{a}, {\`{O}}scar},
	doi = {10.1257/0002828053828518},
	journal = {American Economic Review},
	month = {March},
	number = {1},
	pages = {161--182},
	title = {Estimation and Inference of Impulse Responses by Local Projections},
	url = {http://www.aeaweb.org/articles?id=10.1257/0002828053828518},
	volume = {95},
	year = {2005}}

@article{Jorda2023,
	author = {Jord\`{a}, {\`{O}}scar},
	doi = {10.1146/annurev-economics-082222-065846},
	journal = {Annual Review of Economics},
	number = {1},
	pages = {607-631},
	title = {Local Projections for Applied Economics},
	volume = {15},
	year = {2023}}

@article{Kilian1998,
	author = {Kilian, Lutz},
	doi = {10.1162/003465398557465},
	journal = {Review of Economics and Statistics},
	month = may,
	number = {2},
	pages = {218--230},
	title = {Small-sample Confidence Intervals for Impulse Response Functions},
	url = {http://www.mitpressjournals.org/doi/10.1162/003465398557465},
	volume = {80},
	year = {1998}}

@article{Kilian2011,
	author = {Kilian, Lutz and Kim, Yun Jung},
	doi = {10.1162/REST_a_00143},
	issn = {0034-6535},
	journal = {Review of Economics and Statistics},
	month = {nov},
	number = {4},
	pages = {1460--1466},
	title = {How Reliable Are Local Projection Estimators of Impulse Responses?},
	url = {http://www.mitpressjournals.org/doi/10.1162/REST{\_}a{\_}00143},
	volume = {93},
	year = {2011}}

@book{Kilian2017,
	author = {Lutz Kilian and Helmut L\"{u}tkepohl},
	publisher = {Cambridge University Press},
	title = {{Structural Vector Autoregressive Analysis}},
	year = 2017}

@article{Kuersteiner2001,
	author = {Guido M. Kuersteiner},
	doi = {https://doi.org/10.1016/S0304-4076(01)00088-4},
	issn = {0304-4076},
	journal = {Journal of Econometrics},
	number = {2},
	pages = {359-405},
	title = {Optimal instrumental variables estimation for {ARMA} models},
	url = {https://www.sciencedirect.com/science/article/pii/S0304407601000884},
	volume = {104},
	year = {2001}}

@article{Leeb2005,
	author = {Leeb, Hannes and P{\"o}tscher, Benedikt M.},
	doi = {10.1017/S0266466605050036},
	issn = {0266466605},
	journal = {Econometric Theory},
	month = feb,
	number = {01},
	pages = {21--59},
	title = {Model Selection and Inference: Facts and Fiction},
	url = {http://www.journals.cambridge.org/abstract_S0266466605050036},
	volume = {21},
	year = {2005}}

@article{Li2024,
	author = {Dake Li and Mikkel Plagborg-M{\o}ller and Christian K. Wolf},
	doi = {https://doi.org/10.1016/j.jeconom.2024.105722},
	issn = {0304-4076},
	journal = {Journal of Econometrics},
	note = {Themed Issue: Macroeconometrics},
	number = {2},
	pages = {105722},
	title = {Local projections vs.\ {VARs}: Lessons from thousands of {DGPs}},
	url = {https://www.sciencedirect.com/science/article/pii/S030440762400068X},
	volume = {244},
	year = {2024}}

@article{Lutkepohl1984,
	author = {Helmut L\"{u}tkepohl},
	doi = {https://doi.org/10.1016/0165-1765(84)90009-0},
	issn = {0165-1765},
	journal = {Economics Letters},
	number = {4},
	pages = {345-350},
	title = {Linear aggregation of vector autoregressive moving average processes},
	url = {https://www.sciencedirect.com/science/article/pii/0165176584900090},
	volume = {14},
	year = {1984}}

@book{Lutkepohl2005,
	author = {L\"{u}tkepohl, Helmut},
	isbn = {978-3-540-40172-8},
	language = {en},
	publisher = {Springer},
	title = {New Introduction to Multiple Time Series Analysis},
	year = {2005}}

@book{Magnus2007,
	author = {Magnus, J. and Neudecker, H.},
	edition = {3rd},
	publisher = {John Wiley \& Sons},
	series = {Wiley Series in Probability and Statistics},
	title = {Matrix Differential Calculus with Applications in Statistics and Econometrics},
	url = {http://www.janmagnus.nl/misc/mdc2007-3rdedition},
	year = {2007}}

@article{MontielOlea2021,
	author = {Montiel Olea, Jos\'{e} Luis and Plagborg-M{\o}ller, Mikkel},
	doi = {10.3982/ECTA18756},
	issn = {0012-9682},
	journal = {Econometrica},
	language = {en},
	number = {4},
	pages = {1789--1823},
	title = {Local Projection Inference Is Simpler and More Robust Than You Think},
	url = {https://www.econometricsociety.org/doi/10.3982/ECTA18756},
	urldate = {2021-07-19},
	volume = {89},
	year = {2021}}

@article{MPQW2025,
	author = {Jos\'{e} Luis {Montiel Olea} and Mikkel Plagborg-M{\o}ller and Eric Qian and Christian K. Wolf},
	journal = {NBER Macroeconomics Annual},
	note = {Forthcoming},
	title = {Local Projections or {VARs}? {A} Primer for Macroeconomists},
	year = 2025}

@unpublished{Mueller2011,
	author = {Ulrich K. M\"{u}ller and James H. Stock},
	note = {Manuscript, Princeton University},
	title = {Forecasts in a Slightly Misspecified Finite Order {VAR}},
	year = 2011}

@article{Nakamura2018,
	author = {Nakamura, Emi and Steinsson, J{\'o}n},
	doi = {10.1257/jep.32.3.59},
	issn = {0040-0912},
	journal = {Journal of Economic Perspectives},
	month = aug,
	number = {3},
	pages = {59--86},
	title = {Identification in Macroeconomics},
	url = {https://pubs.aeaweb.org/doi/10.1257/jep.32.3.59},
	volume = {32},
	year = {2018}}

@article{Plagborg2021,
	author = {Plagborg-M{\o}ller, Mikkel and Wolf, Christian K.},
	doi = {10.3982/ECTA17813},
	issn = {0012-9682},
	journal = {Econometrica},
	language = {en},
	number = {2},
	pages = {955--980},
	title = {Local Projections and {VARs} Estimate the Same Impulse Responses},
	url = {https://www.econometricsociety.org/doi/10.3982/ECTA17813},
	volume = {89},
	year = {2021}}

@article{Pope1990,
	author = {Pope, Alun Lloyd},
	doi = {10.1111/j.1467-9892.1990.tb00056.x},
	journal = {Journal of Time Series Analysis},
	month = may,
	number = {3},
	pages = {249--258},
	title = {Biases of Estimators in Multivariate Non-{Gaussian} Autoregressions},
	url = {http://doi.wiley.com/10.1111/j.1467-9892.1990.tb00056.x},
	volume = {11},
	year = {1990}}

@incollection{Ramey2016,
	author = {Ramey, Valerie A.},
	booktitle = {Handbook of Macroeconomics},
	chapter = {2},
	doi = {10.1016/bs.hesmac.2016.03.003},
	editor = {Taylor, John B. and Uhlig, Harald},
	pages = {71--162},
	publisher = {Elsevier},
	title = {Macroeconomic Shocks and Their Propagation},
	url = {http://dx.doi.org/10.1016/bs.hesmac.2016.03.003},
	volume = 2,
	year = {2016}}

@article{Robins1992,
	 ISSN = {0006341X, 15410420},
	 URL = {http://www.jstor.org/stable/2532304},
	 author = {James M. Robins and Steven D. Mark and Whitney K. Newey},
	 journal = {Biometrics},
	 number = {2},
	 pages = {479--495},
	 publisher = {International Biometric Society},
	 title = {Estimating Exposure Effects by Modelling the Expectation of Exposure Conditional on Confounders},
	 urldate = {2025-10-30},
	 volume = {48},
	 year = {1992}
}

@article{Robins1995,
	author = {James M. Robins and Andrea Rotnitzky},
	issn = {01621459, 1537274X},
	journal = {Journal of the American Statistical Association},
	number = {429},
	pages = {122--129},
	publisher = {[American Statistical Association, Taylor & Francis, Ltd.]},
	title = {Semiparametric Efficiency in Multivariate Regression Models with Missing Data},
	url = {http://www.jstor.org/stable/2291135},
	urldate = {2025-10-03},
	volume = {90},
	year = {1995}}

@incollection{Rothenberg1984,
	author = {Thomas J. Rothenberg},
	booktitle = {Handbook of Econometrics},
	chapter = {15},
	doi = {https://doi.org/10.1016/S1573-4412(84)02007-9},
	editor = {Zvi Griliches and Michael D. Intriligator},
	issn = {1573-4412},
	pages = {881-935},
	publisher = {Elsevier},
	title = {Approximating the distributions of econometric estimators and test statistics},
	url = {https://www.sciencedirect.com/science/article/pii/S1573441284020079},
	volume = {2},
	year = {1984}}

@article{Schorfheide2005,
	author = {Frank Schorfheide},
	doi = {https://doi.org/10.1016/j.jeconom.2004.08.009},
	issn = {0304-4076},
	journal = {Journal of Econometrics},
	number = {1},
	pages = {99--136},
	title = {{VAR} forecasting under misspecification},
	url = {http://www.sciencedirect.com/science/article/pii/S0304407604001514},
	volume = {128},
	year = {2005}}

@article{Sims1980,
	author = {Sims, Christopher A},
	journal = {Econometrica},
	number = 1,
	pages = {1--48},
	publisher = {JSTOR},
	title = {Macroeconomics and Reality},
	volume = 48,
	year = {1980}}

@article{Stock2018,
	author = {James H. Stock and Mark W. Watson},
	journal = {Economic Journal},
	number = 610,
	pages = {917--948},
	title = {Identification and Estimation of Dynamic Causal Effects in Macroeconomics Using External Instruments},
	volume = 128,
	year = 2018}

@unpublished{Xu2023,
	author = {Ke-Li Xu},
	note = {Manuscript, Indiana University Bloomington},
	title = {Local Projection Based Inference under General Conditions},
	year = 2023}

@article{robins2000profile,
  	title={On Profile Likelihood: Comment},
	  author={Robins, James M and Rotnitzky, Andrea and van der Laan, Mark},
	  journal={Journal of the American Statistical Association},
	  volume={95},
  	number={450},
 	 pages={477--482},
 	 year={2000}
}

@article{robins_rotnitzky,
	title={Comments on ``Inference for Semiparametric Models: Some Questions and an Answer'', by P.J.\ Bickel and J.\ Kwon},
	author={Robins, J.M. and Rotnitzky, Andrea},
	 journal={Statistica Sinica},
 	pages={920--936},
	 year={2001},
	 volume = 11,
	 number = 4
 }

\end{appendices}

\end{document}


\begin{frontmatter}

\title{Supplement to ``Doubly Robust Local Projections
and Some Unpleasant VARithmetic''}
\runtitle{Supplement to ``Doubly Robust Local Projections\dots''}

\begin{aug}
%
%
%
\author[add1]{\fnms{Jos\'{e} Luis}~\snm{Montiel Olea}\ead[label=e1]{montiel.olea@gmail.com}}
\author[add2]{\fnms{Mikkel}~\snm{Plagborg-M{\o}ller}\ead[label=e2]{mikkelpm@uchicago.edu}}
\author[add3]{\fnms{Eric}~\snm{Qian}\ead[label=e3]{ericqian@princeton.edu}}
\author[add4]{\fnms{Christian K.}~\snm{Wolf}\ead[label=e4]{ckwolf@mit.edu}}
\address[add1]{%
\orgdiv{Department of Economics},
\orgname{Cornell University}}

\address[add2]{%
\orgdiv{Department of Economics},
\orgname{University of Chicago}}

\address[add3]{%
\orgdiv{Department of Economics},
\orgname{Princeton University}}

\address[add4]{%
\orgdiv{Department of Economics},
\orgname{Massachusetts Institute of Technology}}
\end{aug}


\end{frontmatter}

\begin{appendix}
\setcounter{section}{2}

\section{Further theoretical results}

\subsection{Heteroskedasticity}
\label{app:hetero}

The conclusions of Propositions 3.1 and 3.2 do not require shock independence, either cross-sectionally or across time. In particular, our proofs of these propositions and the auxiliary lemmas in \cref{app:proof_details} replace Assumption 2.1(i) by the following:\footnote{It is an interesting topic for future research to investigate whether the other results in Sections 3 and 4 also hold under the weaker condition.}

\begin{asn} \label{asn:mds}
$\varepsilon_t$ is a strictly stationary martingale difference sequence with respect to its natural filtration; $\var(\varepsilon_t) = D \equiv \diag(\sigma_1^2,\dots,\sigma_m^2)$; $\sigma_j>0$ for all $j=1,\dots,m$; $E[\|\varepsilon_t\|^4]<\infty$; and $\sum_{\ell=1}^\infty \sum_{\tau=1}^\infty \|\cov(\varepsilon_t \otimes \varepsilon_t, \varepsilon_{t-\ell} \otimes \varepsilon_{t-\tau})\|<\infty$.
\end{asn}

\cref{asn:mds} strictly weakens Assumption 2.1(i) by allowing the shocks to be conditionally heteroskedastic and by weakening mutual shock independence to orthogonality. The last part of \cref{asn:mds} is a vector version of the fourth-order cumulant summability condition of \citet[Assumption A1]{Kuersteiner2001}. It restricts the higher-order dependence of the shocks, consistent with many stationary models of conditional heteroskedasticity.\footnote{One sufficient condition is finite dependence, i.e., there exists an integer $K$ such that $\lbrace \varepsilon_s \rbrace_{s \geq t+K}$ is independent of $\lbrace \varepsilon_s \rbrace_{s \leq t}$. Another set of sufficient conditions is that $E(\varepsilon_t \mid \lbrace \varepsilon_s \rbrace_{s \neq t})=0$ and $\lbrace \varepsilon_t \otimes \varepsilon_t \rbrace$ has absolutely summable autocovariance function. See also \citet[Remark 2, p.\ 362]{Kuersteiner2001}.}

\subsection{External instruments and proxies}
\label{app:proxy}

Our framework can accommodate identification via an external instrument (also known as a proxy) by a simple reparametrization. To see this, let the proxy be ordered first in $y_t$. Set $j^*=1$, and replace Assumption 2.1(iii) with the following:
\begin{asn} \label{asn:proxy}
The first column of $A$ consists of zeros, except possibly the first element; the first row of $H$ equals $(1,0,0,\dots,0,1)$; the last column of $H$ consists of zeros, except the first element; and the last column of $\alpha(L)$ consists of zeros.
\end{asn}
\noindent This assumption imposes the following restrictions:
\begin{itemize}
\item The proxy $y_{1,t}$ equals
\[y_{1,t} = A_{1,\bullet} y_{t-1} + \varepsilon_{1,t} + \varepsilon_{m,t} + T^{-\zeta}[\alpha_{1,\bullet}(L) + \alpha_{m,\bullet}(L)]\varepsilon_t,\]
which generalizes Assumption 4 in \citet{Plagborg2021} to allow for local contamination by lagged shocks. The last shock $\varepsilon_{m,t}$ is viewed as measurement error or noise ($v_t$ in the notation of \citealp{Plagborg2021}).
\item The dynamics of $(y_{2,t},\dots,y_{n,t})$ (i.e., with the proxy excluded) follow a VAR(1), up to the local misspecification in the form of lags of $(\varepsilon_{1,t},\dots,\varepsilon_{m-1,t})$.
\item The proxy measurement error $\varepsilon_{m,t}$ is orthogonal to all leads and lags of $(y_{2,t},\dots,y_{n,t})$.
\end{itemize}
Now transform the shocks from $\varepsilon_t$ to $\tilde{\varepsilon}_t=(\tilde{\varepsilon}_{1,t},\dots,\tilde{\varepsilon}_{m,t})'$ as follows:
\[\tilde{\varepsilon}_{1,t} \equiv \varepsilon_{1,t} + \varepsilon_{m,t}; \quad \tilde{\varepsilon}_{j,t} \equiv \varepsilon_{j,t} \; \text{ for } j=2,\dots,m-1; \quad \tilde{\varepsilon}_{m,t} \equiv  \varepsilon_{m,t} - \frac{\sigma_m^2}{\sigma_1^2+\sigma_m^2}(\varepsilon_{1,t}+\varepsilon_{m,t}).\]
By construction, the elements of $\tilde{\varepsilon}_t$ are mutually orthogonal. Note that
\[\varepsilon_{1,t} = \frac{\sigma_1^2}{\sigma_1^2+\sigma_m^2}\tilde{\varepsilon}_{1,t} - \tilde{\varepsilon}_{m,t};\quad 
\varepsilon_{j,t} = \tilde{\varepsilon}_{j,t} \; \text{ for } j=2,\dots,m-1; \quad \varepsilon_{m,t} = \frac{\sigma_m^2}{\sigma_1^2+\sigma_m^2}\tilde{\varepsilon}_{1,t} + \tilde{\varepsilon}_{m,t};\]
write $Q$ for the $m \times m$ matrix such that $\varepsilon_t = Q\tilde{\varepsilon}_t$. We can then re-express the VARMA$(1,\infty)$ model for $y_t$ in Equation (2.1) as
\begin{equation} \label{eqn:system_reparam}
y_t = Ay_{t-1} + \tilde{H}[I + T^{-\zeta}\tilde{\alpha}(L)]\tilde{\varepsilon}_t,
\end{equation}
where
\[\tilde{H} \equiv HQ = \begin{pmatrix}
1 & 0 & 0 & \cdots & 0 & 0 \\
\frac{\sigma_1^2}{\sigma_1^2+\sigma_m^2}H_{2,1} & H_{2,2} & H_{2,3} & \cdots & H_{2,m-1} & -H_{2,1} \\
\vdots & & & & &  \vdots \\
\frac{\sigma_1^2}{\sigma_1^2+\sigma_m^2}H_{n,1} & H_{n,2} & H_{n,3} & \cdots & H_{n,m-1} & -H_{n,1} \\
\end{pmatrix},\]
and $\tilde{\alpha}(L) \equiv Q^{-1}\alpha(L)Q$. Under \cref{asn:proxy} above, the impulse responses of $(y_{2,t},\dots,y_{n,t})$ with respect to $\tilde{\varepsilon}_{1,t}$ in the reparametrized system \eqref{eqn:system_reparam} equal $\sigma_1^2/(\sigma_1^2+\sigma_m^2)$ times the impulse responses of $(y_{2,t},\dots,y_{n,t})$ with respect to $\varepsilon_{1,t}$ in the original parametrization in Equation (2.1). This follows by inspection of the elements $\tilde{H}_{i,1}$ for $i \geq 2$, the assumption that $A_{i,1}=0$ for $i \geq 2$, and the fact that the first column of $\tilde{H}\tilde{\alpha}(L)$ equals $H\alpha(L)Q_{\bullet,1} = \frac{\sigma_1^2}{\sigma_1^2+\sigma_m^2} H\alpha_{\bullet,1}(L)$ (here we use the assumption that the last column of $\alpha(L)$ is zero).

If the original shocks (and measurement error) $\varepsilon_t$ satisfy \cref{asn:mds}, then the transformed shocks $\tilde{\varepsilon}_t$ do also, provided that the fourth-order cumulant condition holds for the transformed shocks (recall that \cref{asn:mds} only requires the shocks to be mutually orthogonal, not independent). The transformed system \eqref{eqn:system_reparam} therefore satisfies \cref{asn:mds} and Assumption 2.1(ii)--(v), with $\tilde{H}$ and $\tilde{\alpha}(L)$ in place of $H$ and $\alpha(L)$. Hence, Propositions 3.1 and 3.2 apply. We conclude that LPs on the proxy $y_{1,t}$---and recursive VARs with the proxy ordered first---consistently estimate the true impulse responses of $(y_{2,t},\dots,y_{n,t})$ with respect to $\varepsilon_{1,t}$, up to the scale factor $\frac{\sigma_1^2}{\sigma_1^2+\sigma_m^2}$. This scale factor, which is the same across all response variables $i^*$ and horizons $h$, reflects attenuation bias caused by the measurement error in the proxy. It gets canceled out if one reports relative (i.e., unit-effect-normalized) impulse responses, as explained by \citet[Section 3.3]{Plagborg2021}. Of course, though the VAR estimator is consistent, it suffers from asymptotic bias of order $T^{-\zeta}$, while the LP estimator has bias of the smaller order $T^{-2\zeta}$.

In summary, the main results in the paper carry over to identification via an external instrument/proxy. A caveat is that the instrument/proxy must be strong, in the sense that $\sigma_1^2=\var(\varepsilon_{1,t})$ is not close to zero; otherwise, we will end up dividing by a number close to zero when computing relative impulse responses.

\subsection{Data-dependent lag selection}
\label{app:info_crit}

Here we argue that local projection inference is more robust to data-dependent lag selection errors than conventional VAR inference. The following proposition establishes the properties of conventional information criteria in our class of DGPs.

\begin{prop} \label{thm:infocrit}
Assume that the $\check{n}$-dimensional process $\lbrace \check{y}_t \rbrace$ is a stationary solution of the local-to-SVAR($p_0$) model in Equation (2.2). Assume also that $\check{A}_{p_0} \neq 0$ so that $p_0$ is the minimal true autoregressive lag order, $\check{H}D\check{H}'$ is positive definite, and the process (written in companion form as on p.\ 7 of the main paper) satisfies Assumption 2.1.

Let $\hat{\check{\Sigma}}(p)$ denote the $\check{n} \times \check{n}$ sample residual variance-covariance matrix from a least-squares VAR($p$) regression on the data $(\check{y}_1,\dots,\check{y}_T)$. Let $\lbrace g_T \rbrace_T$ be a deterministic scalar sequence. Fix a maximal lag length $K \geq p_0$. Suppose we select the lag length by minimizing an information criterion:
\[\hat{p} \equiv \argmin_{0 \leq p \leq K}\; \left\lbrace \log\det\left(\hat{\check{\Sigma}}(p)\right) + p \times g_T \right\rbrace.\]
Then the following statements hold:
\begin{enumerate}[i)]
\item \label{itm:infocrit_under} If $g_T \to 0$, then $P(\hat{p} < p_0) \to 0$ as $T \to \infty$.
\item \label{itm:infocrit_over} If $Tg_T \to \infty$ and $\zeta \geq 1/2$, then $P(\hat{p} > p_0) \to 0$ as $T \to \infty$.
\end{enumerate}
\end{prop}

The proposition implies that, when applied to a VAR in the data $\lbrace \check{y}_t \rbrace$, both the BIC ($g_T = \check{n}^2(\log T)/T$) and AIC ($g_T=2\check{n}^2/T$) select $p_0$ or more lags with high probability in large samples, and in fact the BIC selects exactly $p_0$ lags asymptotically when $\zeta \geq 1/2$.\footnote{For additional intuition, consider the critical case $\zeta=1/2$ where the moving average coefficients are of order $T^{-1/2}$. Under smoothness assumptions on the density of the shocks $\varepsilon_t$, the local-to-SVAR($p_0$) DGP is then contiguous to an exact SVAR($p_0$) DGP \citep[for formal results, see][]{Hallin1988,Hallin1989}. Since the BIC is consistent for $p_0$ in latter DGP, contiguity implies that the BIC also selects $p_0$ lags with probability approaching 1 in the former DGP.} Hence, in the latter case, it follows from Corollary 3.1 that it is pointwise asymptotically valid to report a local projection confidence interval that controls for $\hat{p}$ lags of the data, where $\hat{p}$ is selected by applying the BIC to auxiliary VAR regressions as defined above.

Unlike LP inference, VAR inference is sensitive to minor model selection errors. \citet{Leeb2005} show that VAR confidence intervals with lag length selected by BIC or AIC fail to control coverage uniformly over the VAR parameter space. The breakdown in performance happens for (a sequence of) DGPs that satisfy a VAR($\tilde{p}_0$) model where the first $p_0$ lags have large coefficients and the remaining $\tilde{p}_0-p_0$ lags have small coefficients of order $T^{-1/2}$; this implies a VARMA($p_0,\infty$) representation where the moving average coefficients are of order $T^{-1/2}$. Because the BIC or AIC cannot reliably detect the small coefficients, they will tend to select fewer than $\tilde{p}_0$ lags (cf.\ \cref{thm:infocrit}). The small coefficients on the omitted lags impart an asymptotic bias in the VAR estimator that can cause large coverage distortions for the associated confidence interval, consistent with Corollary 3.1. However, if in this context the BIC-selected lag length is instead used for \emph{local projection} inference, the bias imparted by the omitted lags is much smaller than for the VAR estimator and is in fact asymptotically negligible: this is precisely the message of Proposition 3.1. Though our arguments fall short of proving that local projection inference with data-dependent lag length is \emph{uniformly} valid over some parameter space, they nevertheless show that LP inference remains valid under the types of drifting parameter sequences that \citet{Leeb2005} show are responsible for the non-uniformity of VAR inference.

\subsubsection{Proof of \texorpdfstring{\cref{thm:infocrit}}{Proposition C.1}}

The proof follows standard arguments for exact VAR models, see \citet[Chapter 4.3.2]{Lutkepohl2005} and references therein. We merely show that extra terms induced by the vanishing moving average process are asymptotically negligible under our assumptions. Let $\|\cdot\|$ denote the Frobenius norm.

\cref{thm:ytilde} implies that $T^{-1}\sum_{t=1}^T \check{y}_t \check{y}_{t-\ell}' = T^{-1}\sum_{t=1}^T \tilde{\check{y}}_t \tilde{\check{y}}_{t-\ell}' + O_p(T^{-\zeta})$ for any $\ell \geq 0$, where $\lbrace \tilde{\check{y}}_t \rbrace$ is a process that satisfies the VAR($p_0)$ model with no moving average term ($\alpha(L)=0$), see also Corollary 3.2. It follows from least-squares algebra that the probability limit of $\hat{\check{\Sigma}}(p)$ for any fixed $p$ is the same as it would be in an exact VAR($p_0$) DGP with no moving average term. Statement (\ref{itm:infocrit_under}) of the proposition then follows immediately from standard arguments for lag length estimation in exact VAR models \citep[Proposition 4.2]{Lutkepohl2005}.

To prove statement (\ref{itm:infocrit_over}), assume $p \geq p_0$. It suffices to show that $\hat{\check{\Sigma}}(p) = \hat{\check{\Sigma}}_0 + O_p(T^{-1})$ for some data-dependent matrix $\hat{\check{\Sigma}}_0$ independent of $p$, since this implies\footnote{Actually, in order to apply the delta method to the log determinant, we also use that $\plim \hat{\check{\Sigma}}_0 = \check{H}D\check{H}'$, a matrix that is non-singular by assumption.}
\[P\left(\log\det\left(\hat{\check{\Sigma}}(p)\right) + p g_T > \log\det\left(\hat{\check{\Sigma}}(p_0)\right) + p_0 g_T\right) = P\big((p-p_0)Tg_T > O_p(1)\big) \to 1\]
when $p>p_0$. Let $y_t$ denote the $(\check{n}p)$-dimensional companion form vector obtained by stacking $p$ lags of $\check{y}_t$, as in Equation (2.2). Then $\hat{\check{\Sigma}}(p)$ is the upper left $\check{n} \times \check{n}$ block of the $(\check{n}p) \times (\check{n}p)$ matrix $\hat{\Sigma}=T^{-1}\sum_{t=1}^T \hat{u}_t\hat{u}_t'$ defined in Section 2.2. To finish the proof, we show that $T^{-1}\sum_{t=1}^T \|\hat{u}_t-H\varepsilon_t\|^2 = O_p(T^{-1})$, which by Cauchy-Schwarz implies $\hat{\Sigma} = T^{-1}\sum_{t=1}^T H\varepsilon_t\varepsilon_t'H' + O_p(T^{-1})$, as needed.

Note that when the estimation lag length $p$ weakly exceeds the true autoregressive order $p_0$, all results in our paper apply, as noted in the discussion surrounding Equation (2.2). Hence, using the definition of $u_t$ in \cref{thm:ytilde},
\begin{align*}
\frac{1}{T}\sum_{t=1}^T \|\hat{u}_t-H\varepsilon_t\|^2 &\leq \frac{2}{T}\sum_{t=1}^T \|\underbrace{\hat{u}_t-u_t}_{(A-\hat{A})y_{t-1}}\|^2 + \frac{2}{T}\sum_{t=1}^T \|\underbrace{u_t-H\varepsilon_t}_{T^{-\zeta}H\alpha(L)\varepsilon_t}\|^2 \\
&\leq 2\underbrace{\|\hat{A}-A\|^2}_{O_p((T^{-1/2}+T^{-\zeta})^2)} \underbrace{\frac{1}{T}\sum_{t=1}^T \|y_{t-1}\|^2}_{O_p(1)} + 2T^{-2\zeta} \|H\|^2\underbrace{\frac{1}{T}\sum_{t=1}^T \|\alpha(L)\varepsilon_t\|^2}_{O_p(1)} \\
&= O_p(T^{-1} + T^{-2\zeta}),
\end{align*}
where the three $O_p(\cdot)$ statements in the penultimate line rely on \cref{thm:A,thm:y_second_moment} and Assumption 2.1(v), respectively. When $\zeta \geq 1/2$, the right-hand side above is $O_p(T^{-1})$. \qed

\subsection{Inference on multiple impulse responses}
\label{app:joint}

This subsection generalizes the worst-case bias formula in Proposition 4.1 to the multi-dimensional case and derives the worst-case coverage of the Wald confidence ellipsoid.

\subsubsection{Set-up}
We consider inference on any combination of impulse responses for various horizons $h$, response variables $i^*$, and shocks $j^*$. When referring to impulse responses and estimators of these, we need to make the response variable and shock explicit in the notation. Thus, we write $\theta_{i^*,j^*,h,T}$, $\hat{\beta}_{i^*,j^*,h}$, and $\hat{\delta}_{i^*,j^*,h}$, with the definitions being the same as in Section 2. Let $k$ denote the total number of impulse responses of interest. We refer to the list of impulse responses by the collection of triples $\lbrace (i_a^*,j_a^*,h_a) \rbrace_{a=1}^k$ indexing the response variable, shock variable, and horizon, respectively. Define the $k$-dimensional vectors of true impulse responses and LP and VAR estimators:
\[\boldsymbol{\theta}_T \equiv \begin{pmatrix}
\theta_{i_1^*,j_1^*,h_1,T} \\
\vdots \\
\theta_{i_k^*,j_k^*,h_k,T}
\end{pmatrix},\quad \hat{\boldsymbol{\beta}} \equiv \begin{pmatrix}
\hat{\beta}_{i_1^*,j_1^*,h_1} \\
\vdots \\
\hat{\beta}_{i_k^*,j_k^*,h_k}
\end{pmatrix}, \quad \hat{\boldsymbol{\delta}} \equiv \begin{pmatrix}
\hat{\delta}_{i_1^*,j_1^*,h_1} \\
\vdots \\
\hat{\delta}_{i_k^*,j_k^*,h_k}
\end{pmatrix}.\]
It follows from Propositions 3.1 and 3.2 that, when $\zeta=1/2$,
\begin{equation} \label{eqn:asy_normal_joint}
\sqrt{T}\begin{pmatrix}
\hat{\boldsymbol{\beta}} - \boldsymbol{\theta}_T \\
\hat{\boldsymbol{\delta}} - \boldsymbol{\theta}_T
\end{pmatrix}
\stackrel{d}{\to} N\left( \begin{pmatrix}
0_{k \times 1} \\
\abias(\hat{\boldsymbol{\delta}})
\end{pmatrix} , \begin{pmatrix}
\avar(\hat{\boldsymbol{\beta}}) & \acov(\hat{\boldsymbol{\beta}},\hat{\boldsymbol{\delta}}) \\
\acov(\hat{\boldsymbol{\delta}},\hat{\boldsymbol{\beta}}) & \avar(\hat{\boldsymbol{\delta}})
\end{pmatrix} \right),
\end{equation}
for a $k$-dimensional vector $\abias(\hat{\boldsymbol{\delta}})$ (defined in the proof of \cref{thm:bias_general} below) and $k \times k$ matrices $\avar(\hat{\boldsymbol{\beta}})$, $\avar(\hat{\boldsymbol{\delta}})$, and $\acov(\hat{\boldsymbol{\beta}},\hat{\boldsymbol{\delta}})$ given in Corollary A.2 in Appendix A.3. This corollary also implies that the difference $\hat{\boldsymbol{\beta}}-\hat{\boldsymbol{\delta}}$ is asymptotically independent of $\hat{\boldsymbol{\delta}}$, which is not surprising given the general arguments of \citet{Hausman1978} and the facts that (i) the asymptotic variances of the estimators are the same as in the model with $\alpha(L)=0$ and (ii) the VAR estimator is the quasi-MLE in such a model. It follows that $\avar(\hat{\boldsymbol{\beta}}) \geq \avar(\hat{\boldsymbol{\delta}})$ in the positive semidefinite sense.

\subsubsection{Worst-case bias}
The following result generalizes the univariate worst-case bias formula in Proposition 4.1.

\begin{prop} \label{thm:bias_general}
Impose Assumption 2.1, with part (iii) holding for all shock indices $j_1^*,\dots,j_k^*$, and let $\zeta=1/2$. Let $R$ be a constant matrix with $k$ columns. Then
\[\max_{\alpha(L) \colon \|\alpha(L) \| \leq M} \|R\abias(\hat{\boldsymbol{\delta}})\|^2 = M^2 \lambda_{\max}\! \left( R[\avar(\hat{\boldsymbol{\beta}})-\avar(\hat{\boldsymbol{\delta}})]R' \right),\]
where $\lambda_{\max}(B)$ denotes the largest eigenvalue of the matrix $B$.
\end{prop}

The proposition shows that the worst-case squared norm of the bias of the VAR estimator $R\hat{\boldsymbol{\delta}}$ of $R\boldsymbol{\theta}_T$ is a function of two simple quantities: the bound $M$ on misspecification, and the largest eigenvalue of the difference $\avar(R\hat{\boldsymbol{\beta}})-\avar(R\hat{\boldsymbol{\delta}})$ between the variance-covariance matrices for the LP and VAR estimators. The latter eigenvalue equals $\max_{\|\varsigma\|=1} \lbrace  \avar(\varsigma'R\hat{\boldsymbol{\beta}})-\avar(\varsigma'R\hat{\boldsymbol{\delta}}) \rbrace$, i.e., the \emph{largest} efficiency gain for VAR over LP across all linear combinations (with norm 1) of the estimated parameters. Consequently, the worst-case bias is non-negligible if the VAR offers efficiency gains for \emph{any} linear combination of the parameters of interest, echoing our univariate results. When $R$ is a row vector, then the proposition implies that our conclusions from Section 4.2 extend to inference on any \emph{linear combination} of impulse responses.

\subsubsection{Worst-case coverage of confidence ellipsoid}
We next derive the coverage of the conventional Wald confidence ellipsoid based on the VAR estimator. The  level-($1-a$) confidence ellipsoid is given by
\[\ce(\hat{\boldsymbol{\delta}}) \equiv \left\lbrace \tilde{\boldsymbol{\theta}} \in \mathbb{R}^k \colon T(\hat{\boldsymbol{\delta}}-\tilde{\boldsymbol{\theta}})'\avar(\hat{\boldsymbol{\delta}})^{-1}(\hat{\boldsymbol{\delta}}-\tilde{\boldsymbol{\theta}}) \leq \chi_{1-a,k}^2 \right\rbrace,\]
where $\chi_{1-a,k}^2$ is the $1-a$ quantile of the $\chi^2$ distribution with $k$ degrees of freedom.

\begin{cor} \label{thm:coverage_joint}
Impose Assumption 2.1, with part (iii) holding for all shock indices $j_1^*,\dots,j_k^*$, and let $\zeta=1/2$. Assume also that $\avar(\hat{\boldsymbol{\delta}})$ is non-singular. Then
\[\min_{\alpha(L) \colon \|\alpha(L) \| \leq M} \lim_{T \to \infty} P(\boldsymbol{\theta}_T \in \ce(\hat{\boldsymbol{\delta}})) = F_k\left(\chi_{1-a,k}^2; M^2 \left[\lambda_{\max}(\avar(\hat{\boldsymbol{\beta}})\avar(\hat{\boldsymbol{\delta}})^{-1})-1\right] \right),\]
where $F_k(x;c)$ is the cumulative distribution function, evaluated at point $x \geq 0$, of a non-central $\chi^2$ distribution with $k$ degrees of freedom and non-centrality parameter $c \geq 0$.
\end{cor}

The worst-case coverage probability of the VAR confidence ellipsoid depends on three scalars: the bound $M$ on misspecification, the dimension $k$ of the ellipsoid, and the ``multivariate relative standard error''
\[\sqrt{\lambda_{\min}(\avar(\hat{\boldsymbol{\delta}})\avar(\hat{\boldsymbol{\beta}})^{-1})} = [\lambda_{\max}(\avar(\hat{\boldsymbol{\beta}})\avar(\hat{\boldsymbol{\delta}})^{-1})]^{-1/2} = \min_{\varsigma \in \mathbb{R}^k} \sqrt{\avar(\varsigma'\hat{\boldsymbol{\delta}})/\avar(\varsigma'\hat{\boldsymbol{\beta}})}.\]
Again, the worst-case coverage distortion is an increasing function of the \emph{largest} efficiency gain for VAR over LP across \emph{all} linear combinations of the impulse responses. Since VAR impulse response estimates are often highly correlated across horizons, this suggests that the VAR undercoverage can in fact be particularly severe in the multivariate case. Numerical calculations (available upon request from the authors) show that the coverage distortions can be severe even when $M=1$, regardless of the dimension $k$.

\subsubsection{Proof of \texorpdfstring{\cref{thm:bias_general}}{Proposition C.2}}
Define $\tilde{\alpha}_\ell = D^{-1/2}\alpha_\ell D^{1/2}$ for all $\ell \geq 1$. Notice that $\|\alpha(L)\|^2 = \sum_{\ell=1}^\infty \|\tilde{\alpha}_\ell\|^2$.

By Proposition 3.2, we have $\abias(\hat{\delta}_{i^*,j^*,h}) = \sum_{\ell=1}^\infty \tr(\Xi_{i^*,j^*,h,\ell}\tilde{\alpha}_\ell)$, where
\[\Xi_{i^*,j^*,h,\ell} \equiv D^{1/2}H'(A')^{\ell-1}S^{-1}\Psi_{i^*,j^*,h}HD^{1/2} - \mathbbm{1}(\ell \leq h)D^{-1/2}e_{j^*,m}e_{i^*,n}'A^{h-\ell}HD^{1/2}.\]
Since $\tr(\Xi_{i^*,j^*,h,\ell}\tilde{\alpha}_\ell) = \ve(\Xi_{i^*,j^*,h,\ell})'\ve(\tilde{\alpha}_\ell')$, we can write
\[\abias(\hat{\boldsymbol{\delta}}) = \sum_{\ell=1}^\infty \Upsilon_\ell \ve(\tilde{\alpha}_\ell'),\]
where
\[\Upsilon_{\ell} \equiv \Big(\ve(\Xi_{i_1^*,j_1^*,h_1,\ell}),\dots,\ve(\Xi_{i_k^*,j_k^*,h_k,\ell})\Big)'
\in \mathbb{R}^{k \times m^2}.\]
Hence,
\[\max_{\alpha(L) \colon \|\alpha(L)\| \leq M} \|R\abias(\hat{\boldsymbol{\delta}})\|^2 = \max_{\lbrace \tilde{\alpha}_\ell \rbrace_{\ell=1}^\infty \colon \sum_{\ell=1}^\infty \|\tilde{\alpha}_\ell'\|^2 \leq M^2} \left\|\sum_{\ell=1}^\infty R \Upsilon_\ell \ve(\tilde{\alpha}_\ell')\right\|^2.\]
\cref{thm:bias_general_aux1} below shows that the final expression above equals $M^2\lambda_{\max}(\sum_{\ell=1}^\infty R\Upsilon_\ell \Upsilon_\ell' R')$ (the lemma only explicitly considers the case $M=1$, but the general case then follows from the homogeneity of degree 1 of the norm). Finally, \cref{thm:bias_general_aux2} below shows that $\sum_{\ell=1}^\infty \Upsilon_\ell \Upsilon_\ell' = \avar(\hat{\boldsymbol{\beta}})-\avar(\hat{\boldsymbol{\delta}})$. This completes the proof of the proposition. The proof of \cref{thm:bias_general_aux1} shows that the maximum above is achieved when $\ve(\tilde{\alpha}_\ell') \propto \Upsilon_\ell' v$ (with the constant of proportionality being independent of $\ell$ and chosen to satisfy the norm constraint), where $v$ is the eigenvector corresponding to the largest eigenvalue of $R[\avar(\hat{\boldsymbol{\beta}})-\avar(\hat{\boldsymbol{\delta}})]R'$. In the univariate case $k=1$, this reduces to expression (4.1) in Section 4.2.
\qed

\begin{lem} \label{thm:bias_general_aux1}
Let $\mathcal{X}$ denote the set of sequences $\lbrace x_\ell \rbrace_{\ell=1}^\infty$ of $m \times m$ matrices $x_\ell$ satisfying $\sum_{\ell=1}^\infty \|x_\ell\|^2 \leq 1$. Let $\lbrace L_\ell \rbrace_{\ell=1}^\infty$ be a sequence of $r \times m^2$ matrices $L_\ell$ satisfying $\sum_{\ell=1}^\infty \|L_\ell\|^2<\infty$. Then
\begin{equation} \label{eqn:bias_general_aux1}
\max_{\lbrace x_\ell \rbrace_{\ell=1}^\infty \in \mathcal{X}} \left\|\sum_{\ell=1}^\infty L_\ell \ve(x_\ell) \right\|^2 = \lambda_{\max}\!\left(\sum_{\ell=1}^\infty L_\ell L_\ell' \right).
\end{equation}
\end{lem}

\begin{proof}
A short proof using abstract functional analysis is available upon request from the authors. Below we provide a more elementary proof.

The statement of the lemma is obvious if $\sum_{\ell=1}^\infty \|L_\ell\|^2=0$, in which case both sides of the above display equal 0. Hence, we may assume that the series $V \equiv \sum_{\ell=1}^\infty L_\ell L_\ell'$ converges to a non-zero matrix. Let $v$ be the unit-length eigenvector corresponding to the largest eigenvalue $\lambda \equiv \lambda_{\max}(V) \in (0,\infty)$ of $V$.

The purported maximum \eqref{eqn:bias_general_aux1} is achieved by the sequence $\lbrace x_\ell^* \rbrace$ given by $\ve(x_\ell^*) = \lambda^{-1/2} L_\ell'v$:
\[\left\|\sum_{\ell=1}^\infty L_\ell \ve(x_\ell^*) \right\|^2 = \left\|\lambda^{-1/2}\sum_{\ell=1}^\infty L_\ell L_\ell'v \right\|^2 = \lambda^{-1}\left\|Vv \right\|^2 = \lambda^{-1}\left\|\lambda v \right\|^2 = \lambda\|v\|^2 = \lambda,\]
and
\[\sum_{\ell=1}^\infty \|x_\ell^*\|^2 = \sum_{\ell=1}^\infty \ve(x_\ell^*)'\ve(x_\ell^*) = \lambda^{-1}v'\sum_{\ell=1}^\infty L_\ell L_\ell'v = \lambda^{-1}v'Vv = \lambda^{-1}\lambda = 1.\]
We complete the proof by showing that the left-hand side of \eqref{eqn:bias_general_aux1} is bounded above by the right-hand side. Let $K$ be an arbitrary positive integer. Then
\[\max_{\lbrace x_\ell \rbrace_{\ell=1}^\infty \in \mathcal{X}} \left\|\sum_{\ell=1}^\infty L_\ell \ve(x_\ell) \right\| \leq \max_{\lbrace x_\ell \rbrace_{\ell=1}^\infty \in \mathcal{X}} \left\|\sum_{\ell=1}^K L_\ell \ve(x_\ell) \right\| + \max_{\lbrace x_\ell \rbrace_{\ell=1}^\infty \in \mathcal{X}} \left\|\sum_{\ell=K+1}^\infty L_\ell \ve(x_\ell) \right\|.\]
The second term on the right-hand side is bounded above by $(\sum_{\ell=K+1}^\infty \|L_\ell\|^2)^{1/2}$ by Cauchy-Schwarz. As for the first term, standard results for the eigenvalues of finite-dimensional matrices yield
\begin{align*}
&\max_{\lbrace x_\ell \rbrace_{\ell=1}^\infty \in \mathcal{X}} \left\|\sum_{\ell=1}^K L_\ell \ve(x_\ell) \right\|^2 = \max_{x \in \mathbb{R}^{Km^2}\colon \|x\| \leq 1} \left\| \begin{pmatrix}
L_1 & L_2 & \cdots & L_K
\end{pmatrix}x \right\|^2 \\
&= \lambda_{\max}\!\left(\begin{pmatrix}
L_1 & \cdots & L_K
\end{pmatrix}'\begin{pmatrix}
L_1 & \cdots & L_K
\end{pmatrix} \right) = \lambda_{\max}\!\left(\begin{pmatrix}
L_1 & \cdots & L_K
\end{pmatrix} \begin{pmatrix}
L_1 & \cdots & L_K
\end{pmatrix}' \right) = \lambda_{\max}\!\left(\sum_{\ell=1}^K L_\ell L_\ell'\right).
\end{align*}
We have shown
\[\max_{\lbrace x_\ell \rbrace_{\ell=1}^\infty \in \mathcal{X}} \left\|\sum_{\ell=1}^\infty L_\ell \ve(x_\ell) \right\| \leq \left(\lambda_{\max}\left(\sum_{\ell=1}^K L_\ell L_\ell'\right)\right)^{1/2} + \left(\sum_{\ell=K+1}^\infty \|L_\ell\|^2\right)^{1/2}.\]
Now let $K \to \infty$. Since $\sum_{\ell=1}^\infty L_\ell L_\ell'$ is a convergent series, the first term on the right-hand side above converges to $\lambda^{1/2}$ by continuity of eigenvalues, while the second term converges to 0. This establishes the required bound.
\end{proof}

\begin{lem} \label{thm:bias_general_aux2}
Under the assumptions of \cref{thm:bias_general}, and using the notation in the proof of that proposition, we have
\[\sum_{\ell=1}^\infty \Upsilon_\ell \Upsilon_\ell' = \avar(\hat{\boldsymbol{\beta}})-\avar(\hat{\boldsymbol{\delta}}).\]
\end{lem}

\begin{proof}
By definition of $\Upsilon_\ell$, it suffices to show that, for any indices $a,b \in \lbrace 1,\dots, k\rbrace$,
\begin{equation} \label{eqn:bias_general_aux2}
\sum_{\ell=1}^\infty \ve(\Xi_{i_a^*,j_a^*,h_a,\ell})'\ve(\Xi_{i_b^*,j_b^*,h_b,\ell}) = \acov(\hat{\beta}_{i_a^*,j_a^*,h_a},\hat{\beta}_{i_b^*,j_b^*,h_b}) - \acov(\hat{\delta}_{i_a^*,j_a^*,h_a},\hat{\delta}_{i_b^*,j_b^*,h_b}).
\end{equation}
Multiplying out terms, we find that the left-hand side above equals
\begin{align*}
\sum_{\ell=1}^\infty \tr(\Xi_{i_a^*,j_a^*,h_a,\ell}' \Xi_{i_b^*,j_b^*,h_b,\ell}) &= \sum_{\ell=1}^\infty \tr\left(A^{\ell-1}\Sigma(A')^{\ell-1}S^{-1}\Psi_{i_b^*,j_b^*,h_b}\Sigma \Psi_{i_a^*,j_a^*,h_a}' S^{-1}\right) \\
&\qquad - \sum_{\ell=1}^{h_b} \tr\left(A^{\ell-1}H_{\bullet,j_b^*}e_{i_b^*,n}'A^{h_b-\ell}\Sigma \Psi_{i_a^*,j_a^*,h_a}'S^{-1} \right) \\
&\qquad - \sum_{\ell=1}^{h_a} \tr\left(A^{\ell-1}H_{\bullet,j_a^*}e_{i_a^*,n}'A^{h_a-\ell}\Sigma \Psi_{i_b^*,j_b^*,h_b}'S^{-1} \right) \\
&\qquad + \mathbbm{1}(j_a^*=j_b^*) \sigma_{j_a^*}^{-2} \sum_{\ell=1}^{\min\lbrace h_a,h_b\rbrace} e_{i_b^*,n}'A^{h_b-\ell}\Sigma(A')^{h_a-\ell}e_{i_a^*,n}.
\end{align*}
We now evaluate each of the four terms on the right-hand side above. The \emph{first} term equals
\[\tr\Bigg(\underbrace{\sum_{\ell=1}^\infty A^{\ell-1}\Sigma(A')^{\ell-1}}_{=S}S^{-1}\Psi_{i_b^*,j_b^*,h_b}\Sigma \Psi_{i_a^*,j_a^*,h_a}' S^{-1}\Bigg) = \tr\left(\Psi_{i_b^*,j_b^*,h_b}\Sigma \Psi_{i_a^*,j_a^*,h_a}' S^{-1}\right).\]
The \emph{second} term (in the earlier display) equals
\[-\tr\Bigg( \underbrace{\sum_{\ell=1}^{h_b} A^{\ell-1}H_{\bullet,j_b^*}e_{i_b^*,n}'A^{h_b-\ell}}_{=\sum_{\ell=1}^{h_b} A^{h_b-\ell}H_{\bullet,j_b^*}e_{i_b^*,n}'A^{\ell-1}=\Psi_{i_b^*,j_b^*,h_b}}\Sigma \Psi_{i_a^*,j_a^*,h_a}'S^{-1} \Bigg) = -\tr\left(\Psi_{i_b^*,j_b^*,h_b}\Sigma \Psi_{i_a^*,j_a^*,h_a}' S^{-1}\right),\]
and the \emph{third} term (in the earlier display) also equals this quantity by a symmetric calculation. In conclusion, we have shown
\begin{align*}
&\sum_{\ell=1}^\infty \tr(\Xi_{i_a^*,j_a^*,h_a,\ell}' \Xi_{i_b^*,j_b^*,h_b,\ell}) \\
&= \mathbbm{1}(j_a^*=j_b^*) \sigma_{j_a^*}^{-2} \sum_{\ell=1}^{\min\lbrace h_a,h_b\rbrace} e_{i_b^*,n}'A^{h_b-\ell}\Sigma(A')^{h_a-\ell}e_{i_a^*,n} - \tr\left(\Psi_{i_b^*,j_b^*,h_b}\Sigma \Psi_{i_a^*,j_a^*,h_a}' S^{-1}\right).
\end{align*}
The desired result \eqref{eqn:bias_general_aux2} now follows from Corollary A.2.
\end{proof}

\subsubsection{Proof of \texorpdfstring{\cref{thm:coverage_joint}}{Corollary C.1}}
The result follows straightforwardly from \eqref{eqn:asy_normal_joint} if we can show that the maximal non-centrality parameter equals
\[\max_{\alpha(L) \colon \|\alpha(L) \| \leq M} \abias(\hat{\boldsymbol{\delta}})'\avar(\hat{\boldsymbol{\delta}})^{-1}\abias(\hat{\boldsymbol{\delta}}) = M^2 \left[\lambda_{\max}(\avar(\hat{\boldsymbol{\beta}})\avar(\hat{\boldsymbol{\delta}})^{-1})-1\right].\]
But this follows from applying \cref{thm:bias_general} with $R = \avar(\hat{\boldsymbol{\delta}})^{-1/2}$, since
\begin{align*}
&\lambda_{\max}\! \left( \avar(\hat{\boldsymbol{\delta}})^{-1/2}[\avar(\hat{\boldsymbol{\beta}})-\avar(\hat{\boldsymbol{\delta}})]\avar(\hat{\boldsymbol{\delta}})^{-1/2\prime} \right) \\
&= \lambda_{\max}\! \left( \avar(\hat{\boldsymbol{\delta}})^{-1/2}\avar(\hat{\boldsymbol{\beta}})\avar(\hat{\boldsymbol{\delta}})^{-1/2\prime} - I_k\right) \\
&= \lambda_{\max}\!\left( \avar(\hat{\boldsymbol{\beta}})\avar(\hat{\boldsymbol{\delta}})^{-1} \right) - 1. \qed
\end{align*}

\section{Simulation details and further results}
\label{app:oil}

We here provide supplementary details for the simulation study reported in Section 5.2.

\subsection{Implementation details}
All inference procedures (correctly) assume homoskedastic shocks. The VAR is estimated with an intercept, and confidence intervals are constructed either using the delta method or the recursive residual bootstrap; following the recommendation of \citet{Inoue2020}, we report the Efron bootstrap confidence interval. For LPs, we include the shock measure as the only contemporaneous regressor, control for an intercept and the same $p$ lags of all observables as in the VAR, and report the OLS coefficient on the shock. Confidence intervals are constructed either using homoskedastic OLS standard errors or by bootstrapping an auxiliary VAR as in \citet{MontielOlea2021}, but using a recursive residual bootstrap instead of a wild bootstrap; we follow the latter paper and report the percentile-$t$ bootstrap confidence interval. We use 1,000 bootstrap draws, and the maximal lag length considered for the AIC is $24$.

\subsection{Further results}
\cref{fig:oil_simulations_lags} shows coverage probabilities and median confidence interval lengths for longer estimation lag lengths $p \in \lbrace 15,18\rbrace$. The results are consistent with the asymptotic theory: the longer the estimation lag length, the less severe VAR undercoverage, with correct coverage ensured through equivalence with LP. Of course, since the true DGP is a VAR($18$), the $p=18$ VAR estimator is more efficient than LP at longer horizons (middle panel). The bottom panel of the figure shows the root MSE for VAR and LP, with lag length selected by AIC. By this measure, VAR outperforms LP, indicating that while the VAR bias is large enough to seriously compromise inference (as shown earlier), it is not so large as to threaten the VAR estimator's status as a useful point estimator.

\begin{figure}[p]
\centering
{\textsc{\small Coverage and length for lag length $p = 15$}} \\
\vspace{0.6\baselineskip}
\begin{subfigure}{\textwidth}
\centering
\includegraphics[width=0.825\textwidth]{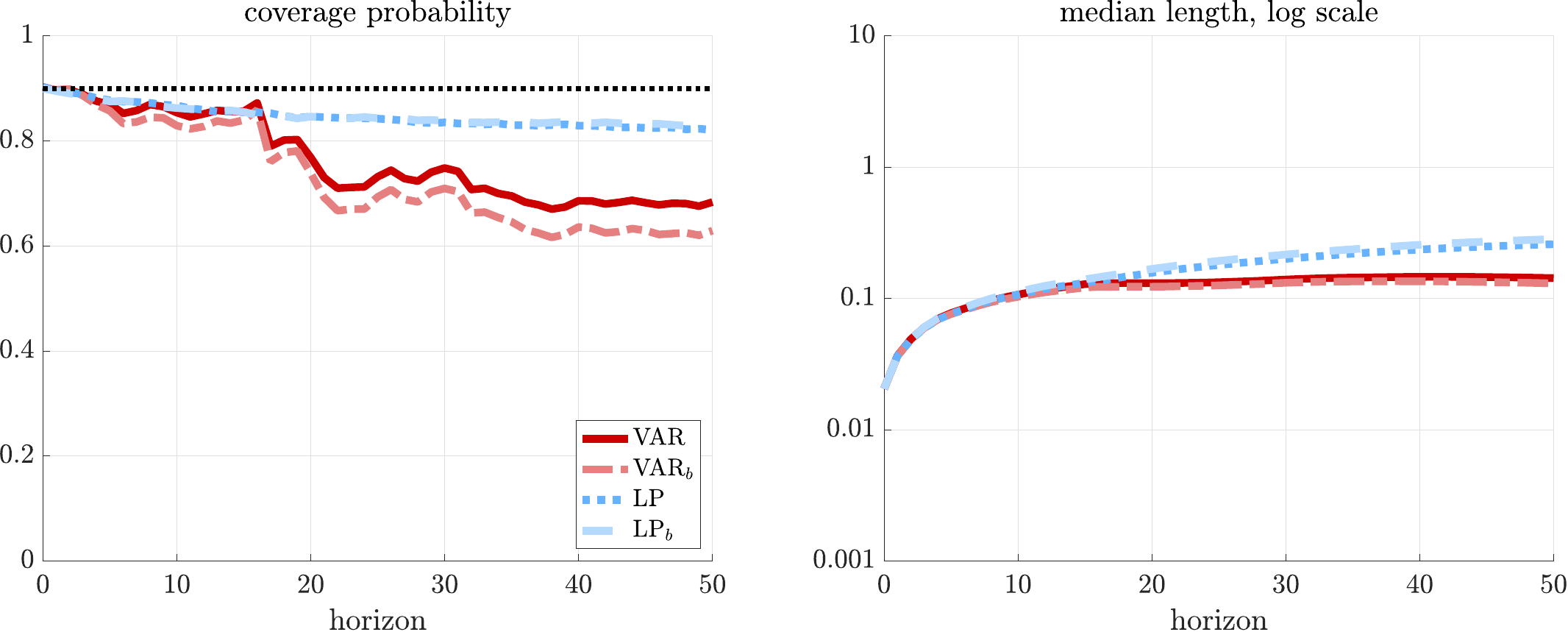}\vspace{0.2cm}
\end{subfigure} 
\centering
{\textsc{\small Coverage and length for lag length $p = 18$}} \\
\vspace{0.6\baselineskip}
\begin{subfigure}{\textwidth}
\centering
\includegraphics[width=0.825\textwidth]{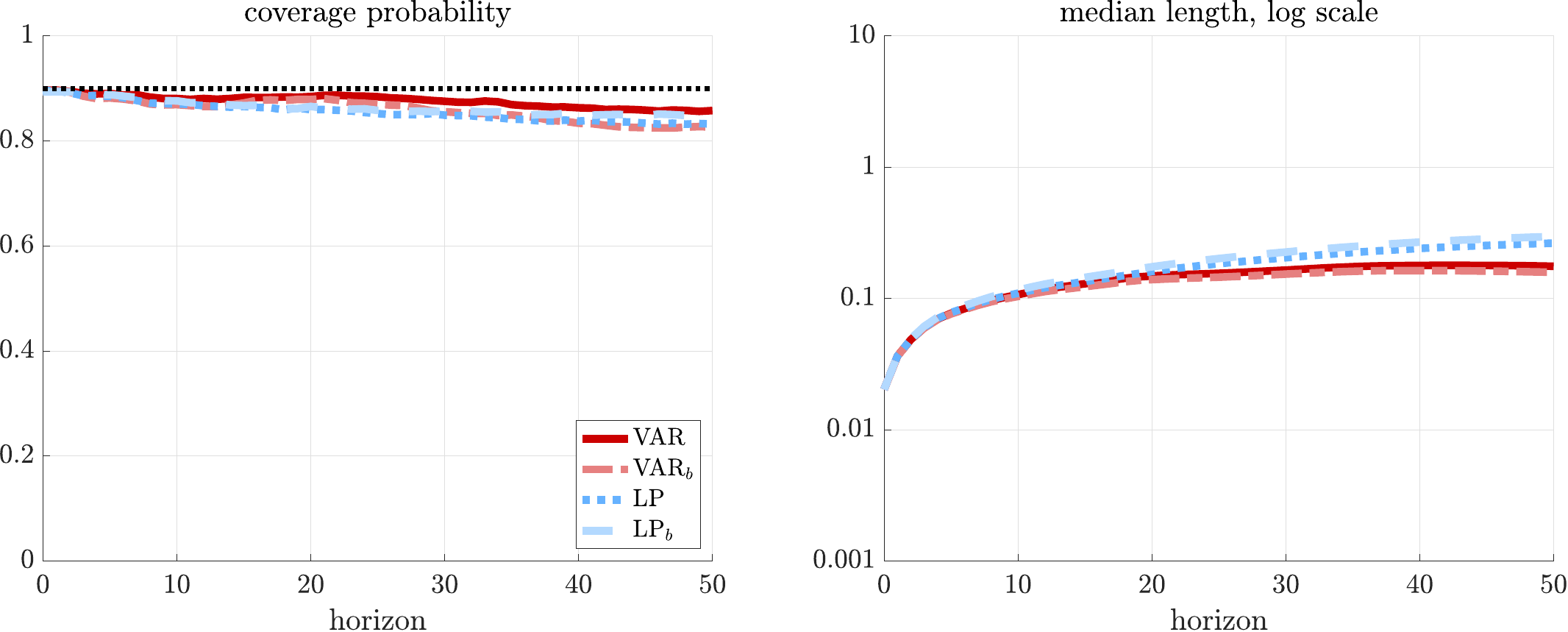}\vspace{0.2cm} \\
\centering
{\textsc{\small RMSE for lag length selected by AIC}} \\
\vspace{0.6\baselineskip}
\includegraphics[width=0.5\textwidth]{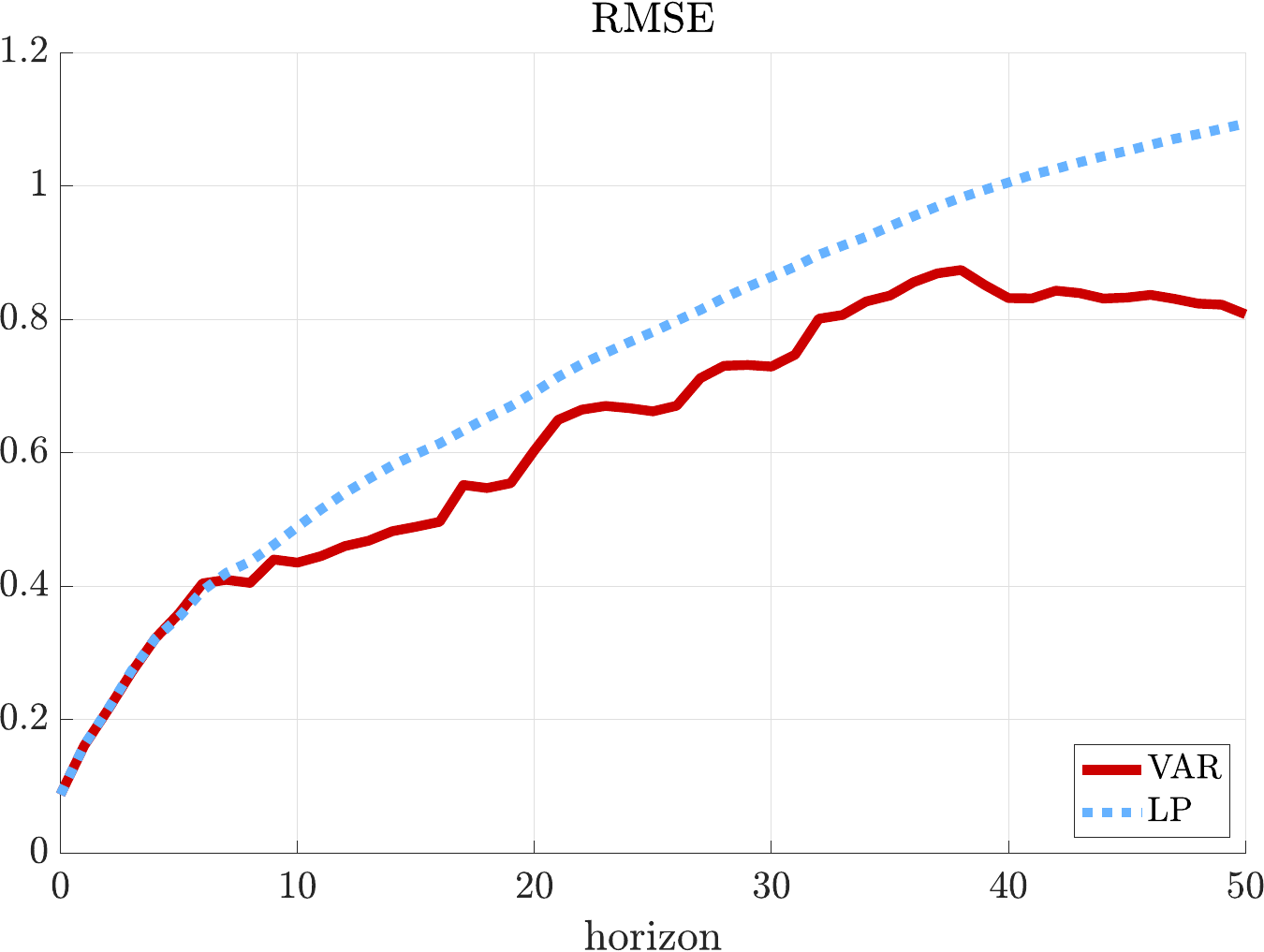}
\end{subfigure} 
\caption{Top two panels: coverage probability (left) and median length (right) for VAR (red) and LP (blue) 90\% confidence intervals computed via the delta method or bootstrap (the latter are indicated with subscript ``b'' in the legends). Bottom panel: root MSE of estimators. Lag length: $p = 15$ in the top panel, $p = 18$ in the middle panel, and $p$ selected by AIC in the bottom panel.}
\label{fig:oil_simulations_lags}
\end{figure}

\section{Further proofs}
\label{app:proof_details}

We impose \cref{asn:mds} and Assumption 2.1(ii)--(v) throughout; as discussed in \cref{app:hetero}, none of the proofs below require Assumption 2.1(i). Let $\|B\|$ denote the Frobenius norm of any matrix $B$. It is well known that this norm is sub-multiplicative: $\|BC\| \leq \|B\|\cdot\|C\|$. Let $I_n$ denote the $n \times n$ identity matrix, $0_{m \times n}$ the $m \times n$ matrix of zeros, and $e_{i,n}$ the $n$-dimensional unit vector with a 1 in the $i$-th position. Recall from Assumption 2.1 the definitions $D\equiv \var(\varepsilon_t)= \diag(\sigma_1^2,\dots,\sigma_m^2)$, $\tilde{y}_t \equiv (I_n-A L)^{-1}H\varepsilon_t = \sum_{s=0}^\infty A^s H \varepsilon_{t-s}$, and $S \equiv \var(\tilde{y}_t)$.

\subsection{Main lemmas}

\begin{lem} \label{thm:lp_coefficients_aux}
For any $i^* \in \{1, \ldots, n\}$ and $j^* \in \{1,\ldots, m\}$, we have  
\begin{equation*} \label{eqn:lp_repr_aux0}
y_{i^*,t+h} = \theta_{h,T}\varepsilon_{j^*,t} + \underline{B}_{h,i^*,j^*}'\underline{y}_{j^*,t} + B_{h,i^*,j^*}' y_{t-1} + \xi_{i^*,h,t} + T^{-\zeta} \Theta_h(L)\varepsilon_t,
\end{equation*}
where 
\begin{align*}
\theta_{h,T} &\equiv e_{i^*,n}'(A^h H + T^{-\zeta}\sum_{\ell=1}^h A^{h-\ell}H \alpha_\ell)e_{j^*,m}, \\
\underline{B}_{h,i^*,j^*}' &\equiv e'_{i^*,n}A^h \underline{H}_{j^*} H_{11}^{-1}, \\
B'_{h,i^*,j^*} &\equiv  e'_{i^*,n}\left [ A^{h+1} -  A^h \underline{H}_{j^*} H_{11}^{-1} \underline{I}_{j^*}A \right],  \\
\xi_{i^*,h,t}  &\equiv  e'_{i^*,n}A^h \overline{H}_{j^*}\overline{\varepsilon}_{j^*,t} + \sum_{\ell=1}^h  e'_{i^*,n}A^{h-\ell} H\varepsilon_{t+\ell},
\end{align*}
and $\Theta_h(L) = \sum_{\ell=-\infty}^{\infty} \Theta_{h,\ell} L^{\ell}$ is an absolutely summable, $1 \times n$ two-sided lag polynomial with the $j^*$-th element of $\Theta_{h,0}$ equal to zero. Moreover, 
\[T^{-1}\sum_{t=1}^{T-h} (\Theta_h(L)\varepsilon_t) \varepsilon_{j^*,t} = O_p(T^{-1/2}).\]
\end{lem} 

\begin{proof}
Iteration on the model in Equation (2.1) yields
\begin{align}
y_{t+h} &=  A^{h+1}y_{t-1} + \sum_{\ell=0}^h A^{h-\ell} (H\varepsilon_{t+\ell} + T^{-\zeta}H\alpha(L)\varepsilon_{t+\ell}). \label{eqn:iterate_aux}
\end{align}
As in Section 2.2, let $\underline{y}_{j^*,t} \equiv (y_{1,t},\dots,y_{j^*-1,t})'$ denote the variables ordered before $y_{j^*,t}$ (if any). Analogously, let $\overline{y}_{j^*,t} \equiv (y_{j^*+1,t}, \ldots y_{n,t})'$  denote the variables ordered after $y_{j^*,t}$.

Using Assumption 2.1(iii), partition
\[H = (\underline{H}_{j^*},H_{\bullet, j^*},\overline{H}_{j^*}) = \begin{pmatrix}
H_{11} & 0 & 0 \\
H_{21} & H_{22} & 0 \\
H_{31} & H_{32} & H_{33}
\end{pmatrix}\]
conformably with the vector $y_t=(\underline{y}_{j^*,t}',y_{j^*,t},\overline{y}_{j^*,t}')'$. Let $\underline{I}_{j^*}$ denote the first $j^*-1$ rows of the $n \times n$ identity matrix. Using the definition of $y_{t}$ in Equation (2.1),
\[\underline{y}_{j^*,t} = \underline{I}_{j^*}Ay_{t-1} + H_{11}\underline{\varepsilon}_{j^*,t} + T^{-\zeta}H_{11}\underline{I}_{j^*}\alpha(L)\varepsilon_t,\]
where $\underline{\varepsilon}_{j^*,t} = \underline{I}_{j^*} \varepsilon_{t}$. Using the previous equation to solve for $\underline{\varepsilon}_{j^*,t}$ we get
\begin{equation} \label{eqn:lp_coefficients_aux1} 
\underline{\varepsilon}_{j^*,t} = H_{11}^{-1}(\underline{y}_{j^*,t}-\underline{I}_{j^*}Ay_{t-1}-T^{-\zeta}H_{11}\underline{I}_{j^*}\alpha(L)\varepsilon_t).
\end{equation}
Expanding the terms in \eqref{eqn:iterate_aux} we get:
\begin{align*}
y_{t+h} &=A^{h+1}y_{t-1}  + A^h H \varepsilon_{t} + T^{-\zeta}A^h H\alpha(L)\varepsilon_t + \sum_{\ell=1}^h A^{h-\ell} (H\varepsilon_{t+\ell} + T^{-\zeta}H\alpha(L)\varepsilon_{t+\ell}) \\
&= A^{h+1}y_{t-1} + \left( A^h \underline{H}_{j^*} \underline{\varepsilon}_{j^*,t} +  A^h H_{\bullet, j^*} \varepsilon_{j^*,t}   +  A^h \overline{H}_{j^*}\overline{\varepsilon}_{j^*,t} \right) \\
&\quad + T^{-\zeta}A^h H\alpha(L)\varepsilon_t + \sum_{\ell=1}^h A^{h-\ell} (H\varepsilon_{t+\ell} + T^{-\zeta}H\alpha(L)\varepsilon_{t+\ell})\\
&= A^{h+1}y_{t-1} + A^h \underline{H}_{j^*} H_{11}^{-1}(\underline{y}_{j^*,t}-\underline{I}_{j^*}Ay_{t-1}-T^{-\zeta}H_{11}\underline{I}_{j^*}\alpha(L)\varepsilon_t) +  A^h H_{\bullet, j^*} \varepsilon_{j^*,t} + A^h \overline{H}_{j^*}\overline{\varepsilon}_{j^*,t} \\
&\quad + T^{-\zeta}A^h H\alpha(L)\varepsilon_t + \sum_{\ell=1}^h A^{h-\ell} (H\varepsilon_{t+\ell} + T^{-\zeta}H\alpha(L)\varepsilon_{t+\ell}),
\end{align*}
where the last equality follows from substituting \eqref{eqn:lp_coefficients_aux1}. Re-arranging terms we get 
\begin{align*}
y_{i^*,t+h} &=  \left( e'_{i^*,n} A^h H_{\bullet, j^*} \right) \varepsilon_{j^*,t}  + \underbrace{\left(e'_{i^*,n}A^h \underline{H}_{j^*} H_{11}^{-1} \right)}_{\equiv \underline{B}_{h,i^*,j^*}'}\underline{y}_{j^*,t} + \underbrace{\left( e'_{i^*,n}\left [ A^{h+1} -  A^h \underline{H}_{j^*} H_{11}^{-1} \underline{I}_{j^*}A \right]  \right)}_{\equiv B'_{h,i^*,j^*}} y_{t-1} \nonumber \\
&+  \underbrace{e'_{i^*,n} \left( A^h \overline{H}_{j^*}\overline{\varepsilon}_{j^*,t} + \sum_{\ell=1}^h A^{h-\ell} H\varepsilon_{t+\ell} \right )}_{= \xi_{i^*,h,t}} \nonumber \\
&+ T^{-\zeta}  e'_{i^*,n}  \left( - A^h \underline{H}_{j^*} H_{11}^{-1} H_{11}\underline{I}_{j^*}\alpha(L)\varepsilon_t  + \sum_{\ell=0}^h A^{h-\ell} H\alpha(L)\varepsilon_{t+\ell} \right).
\end{align*}
Using the definition of $\theta_{h,T} \equiv e_{i^*,n}'(A^h H + T^{-\zeta}\sum_{\ell=1}^h A^{h-\ell}H \alpha_\ell)e_{j^*,m}$ and adding and subtracting $ e_{i^*,n}' \left( T^{-\zeta}\sum_{\ell=1}^h A^{h-\ell}H \alpha_\ell \right) e_{j^*,m}  \varepsilon_{j^*,t} $ in the display above, we obtain a representation of the form 
\begin{equation} \label{eqn:lp_repr_aux2} 
y_{i^*,t+h} = \theta_{h,T}\varepsilon_{j^*,t} + \underline{B}_{h,i^*,j^*}'\underline{y}_{j^*,t} + B_{h,i^*,j^*}' y_{t-1} + \xi_{i^*,h,t} + T^{-\zeta} \tilde{u}_{t},
\end{equation}
where 
\begin{equation}
\tilde{u}_{t} \equiv   e'_{i^*,n}  \left( - A^h \underline{H}_{j^*} \underline{I}_{j^*}\alpha(L)\varepsilon_t  + \sum_{\ell=0}^h A^{h-\ell} H\alpha(L)\varepsilon_{t+\ell} - \left( \sum_{\ell=1}^h A^{h-\ell}H \alpha_\ell  e_{j^*,m} e'_{j^*,m} \right) \varepsilon_{t}  \right).
\end{equation}
Algebra shows that $\tilde{u}_{t}$ can be written as a two-sided lag polynomial, $\Theta_h(L)=\sum_{\ell=-\infty}^\infty \Theta_{h,\ell} L^\ell$, with coefficients of dimension $1 \times n$ given by the following formulae:
\[\Theta_{h,\ell} = \begin{cases}
-  e'_{i^*,n} A^h \underline{H}_{j^*} \underline{I}_{j^*} \alpha_{\ell} +  \sum_{s=0}^{h} e'_{i^*,n} A^{h-s} H \alpha_{\ell +  s} & \text{for } \ell \geq 1, \\
\sum_{s=1}^{h} e'_{i^*,n} A^{h-s} H \alpha_{s} - \sum_{s=1}^h e'_{i^*,n} A^{h-s}H \alpha_s  e_{j^*,m} e'_{j^*,m} & \text{for } \ell=0, \\
\sum_{s=1}^{h+\ell} e'_{i^*,n} A^{h-s+\ell} H \alpha_{s} & \text{for } \ell \in \{ -(h-1), \ldots, -1\}, \\
0_{1 \times n} & \text{for } \ell \leq -h.
\end{cases}\]
In particular,  $\Theta_{h,0,j^*} \equiv \Theta_{h,0}  e_{j^*,m} = 0$.

We next show that $\Theta_{h}(L)$ is absolutely summable, that is $\sum_{\ell= -\infty  }^{\infty} \| \Theta_{h,l} \| < \infty$. To do this, it suffices to show that $\sum_{\ell= 1  }^{\infty} \| \Theta_{h,l} \| < \infty$, since all the coefficients with index $\ell \leq -h$ are 0. Note that, by definition,
\[\sum_{\ell=1} ^{\infty} \| \Theta_{h,\ell} \| \leq \| A^h \| \|  \underline{H}_{j^*} \underline{I}_{j^*} \| \sum_{\ell=1}^{\infty} \| \alpha_{\ell} \| + \| H \|  \sum_{\ell=1}^{\infty}  \sum_{s=0}^{h} \| A^{h-s} \| \| \alpha_{\ell + s}\|.\]
Let $\lambda \in [0,1)$ and $C >0$ be chosen such that $\| A^{\ell} \| \leq C \lambda^{\ell}$ for all $\ell \geq 0$ (such constants exists by Assumption 2.1(ii)). Then 
\[\sum_{\ell=1}^{\infty}  \sum_{s=0}^{h} \| A^{h-s} \| \| \alpha_{\ell + s}\| \leq C  \sum_{\ell=1}^{\infty}  \sum_{s=1}^{h} \lambda^{h-s} \| \alpha_{\ell + s}\| \leq C  \sum_{\ell=1}^{\infty}  \sum_{s=1}^{h} \| \alpha_{\ell + s}\|  \leq C h \sum_{\ell=1}^{\infty}  \| \alpha_{\ell}\| < \infty,\]
where the last inequality holds because the coefficients of $\alpha(L)$ are summable. We thus conclude that 
\begin{equation*} 
y_{i^*,t+h} = \theta_{h,T}\varepsilon_{j^*,t} + \underline{B}_{h,y}'\underline{y}_{j^*,t} + B_{h,y}' y_{t-1} + \xi_{i^*,h,t} + T^{-\zeta} \Theta_h(L)\varepsilon_t,
\end{equation*}
where $\Theta_{h}(L)$ is a two-sided lag-polynomial with summable coefficients.  

Finally, we show that 
\begin{equation} \label{eqn:twosidedmovavg}
T^{-1}\sum_{t=1}^{T-h} (\Theta_h(L)\varepsilon_t) \varepsilon_{j^*,t} = O_p(T^{-1/2}).
\end{equation}
We can decompose $(\Theta_h(L)\varepsilon_t)\varepsilon_{j^*,t}$ into finitely many terms of the form $(\psi_j(L)\varepsilon_{j,t})\varepsilon_{j^*,t}$ for some two-sided, absolutely summable lag polynomial $\psi_j(L)$ and integer $j$. Since $\Theta_{h,0,j^*}=0$, each of these terms has mean zero by shock orthogonality. By \cref{thm:squared}, each term also has absolutely summable autocovariances. Hence, \citet[Thm.\ 7.1.1]{Brockwell1991} and Chebyshev's inequality imply that the sample average of each term is $O_p(T^{-1/2})$.
\end{proof}

\begin{lem} \label{thm:A}
	\[\hat{A}-A = T^{-\zeta}H\sum_{\ell=1}^{\infty} \alpha_{\ell}DH'(A')^{\ell-1}S^{-1} + T^{-1}\sum_{t=1}^T H\varepsilon_t \tilde{y}_{t-1}'S^{-1} + o_p(T^{-\zeta}).\]
	In particular, $\hat{A}-A = O_p(T^{-\zeta}+T^{-1/2})$.
\end{lem}
\begin{proof}
	Since,
	\[\hat{A} - A = \left(T^{-1}\sum_{t=1}^{T-h} u_t y_{t-1}'\right)\left(T^{-1}\sum_{t=1}^{T-h} y_{t-1}y_{t-1}'\right)^{-1},\]
	the result follows from \cref{thm:uy,thm:y_second_moment}.
\end{proof}

\begin{lem} \label{thm:nu}
	\[\hat{\nu}-H_{\bullet, j^*} = \frac{1}{\sigma_{j^*}^2}T^{-1}\sum_{t=1}^T \xi_{0,t} \varepsilon_{j^*,t} + O_p(T^{-2\zeta}) + o_p(T^{-1/2}).\]
\end{lem}
\begin{proof}
	By \cref{thm:impact}, $\hat{\nu}=(0_{1 \times (j^*-1)},1,\hat{\overline{\nu}}')$, where the $j$-th element of $\hat{\overline{\nu}}$ equals the on-impact local projection of $y_{i^*+j,t}$ on $y_{j^*,t}$, controlling for $\underline{y}_{j^*,t}$ and $y_{t-1}$. The statement of the lemma is therefore a direct consequence of Proposition 3.1 and the fact that (by definition) $\xi_{0,i,t}=0$ for $i \leq j^*$.
\end{proof}

\begin{lem} \label{thm:x}
	Fix $h \geq 0$. Consider the regression of $y_{j^*,t}$ on $q_{j^*,t}\equiv (\underline{y}_{j^*,t}',y_{t-1}')'$, using the observations $t=1,2,\dots,T-h$:
	\[y_{j^*,t} = \hat{\vartheta}_h'q_{j^*, t} + \hat{x}_{h,t}.\]
	Note that the residuals $\hat{x}_{h,t}$ are consistent with the earlier definition in the proof of Proposition 3.1. Let $\underline{\lambda}_{j^*}'$ be the row vector containing the first $j^*-1$ elements of the last row of $-\tilde{H}^{-1}$ (where $\tilde{H}$ is defined in Assumption 2.1(iii)). Let $\lambda_{j^*}' \equiv (-\underline{\lambda}_{j^*}',1,0_{1 \times (n-j^*)})$ and $\vartheta \equiv (\underline{\lambda}_{j^*}', (\lambda_{j^*}'A))'$. Then:
	\begin{enumerate}[i)]
		\item \label{itm:thm_x_i} $\hat{\vartheta}_h - \vartheta = O_p(T^{-\zeta}+T^{-1/2})$.
		\item \label{itm:thm_x_iv_star} For $j \geq j^*$, $T^{-1}\sum_{t=1}^{T-h}(\hat{x}_{h,t}-\varepsilon_{j^*,t})\varepsilon_{j,t} = O_p(T^{-2\zeta}) + o_p(T^{-1/2})$.
		\item \label{itm:thm_x_iv} For $\ell \geq 1$, $T^{-1}\sum_{t=1}^{T-h}(\hat{x}_{h,t}-\varepsilon_{j^*,t})\varepsilon_{t+\ell} = O_p(T^{-2\zeta}) + o_p(T^{-1/2})$.
		\item \label{itm:thm_x_iii} $T^{-1}\sum_{t=1}^{T-h}(\hat{x}_{h,t}-\varepsilon_{j^*,t})\hat{x}_{h,t} = O_p(T^{-2\zeta}) + o_p(T^{-1/2})$.
		\item \label{itm:thm_x_ii} $T^{-1}\sum_{t=1}^{T-h} \hat{x}_{h,t}^2 \overset{p}{\to} \sigma_{j^*}^2$.
		\item \label{itm:thm_x_v} For any absolutely summable two-sided lag polynomial $B(L)$, $T^{-1}\sum_{t=1}^{T-h}(\hat{x}_{h,t}-\varepsilon_{j^*,t})B(L)\varepsilon_t = O_p(T^{-\zeta}+T^{-1/2})$.
	\end{enumerate}
\end{lem}
\begin{proof}
By Equation (2.1), the outcome variables in the model satisfy
\[ y_t = A y_{t-1} + H[I_m + T^{-\zeta}\alpha(L)]\varepsilon_t,\; t=1,2,\dots,T. \]
By Assumption 2.1(iii), the first $j^*$ rows of the matrix $H$ above are of the form $(\tilde{H},0_{j^* \times (j^*-m)})$, where $m$ is the number of shocks and $\tilde{H}$ is a $j^* \times j^*$ lower triangular matrix with 1's on the diagonal.

We can premultiply the first $j^*$ equations of (2.1) by $\tilde{H}^{-1}$ to obtain:
\[ [\tilde{H}^{-1}, 0_{j^* \times (n-j^*)}] y_t = [\tilde{H}^{-1}, 0_{j^* \times (n-j^*)}]  A y_{t-1} + [ I_{j^*}, 0_{j^* \times (m-j^*)}] [I_m + T^{-\zeta}\alpha(L)]\varepsilon_t . \] 
By definition, $-\underline{\lambda}_{j^*}'$ is the row vector containing the first $j^*-1$ elements of the last row of $\tilde{H}^{-1}$ and $\lambda_{j^*}' \equiv (-\underline{\lambda}_{j^*}',1,0_{1 \times (n-j^*)})$. Thus, we can re-write the $j^*$-th equation above as
\[ [ -\underline{\lambda}_{j^*}', 1  , 0_{j^* \times (n-j^*)}] y_t = \lambda'_{j^*}A y_{t-1} +  \varepsilon_{j^*,t} + T^{-\zeta}\alpha_{j^*}(L)\varepsilon_t, \] 
where $\alpha_{j^*}(L)$ is the $j^*$-th row of $\alpha(L)$. Re-arranging terms we get 
\[y_{j^*,t} = \vartheta'q_{j^*,t} + \varepsilon_{j^*,t} + T^{-\zeta}\alpha_{j^*}(L)\varepsilon_t,\]
where $\vartheta \equiv (\underline{\lambda}_{j^*}', (\lambda_{j^*}'A))'$ and $q_{j^*,t}\equiv (\underline{y}_{j^*,t}',y_{t-1}')'$. In a slight abuse of notation, and for notational simplicity, we henceforth replace $q_{j^*,t}$ by $q_t$. 

Statement (\ref{itm:thm_x_i}) follows from standard OLS algebra if we can show that a) $T^{-1}\sum_{t=1}^{T-h}q_t \varepsilon_{j^*,t}=O_p(T^{-\zeta} + T^{-1/2})$, b) $(T^{-1}\sum_{t=1}^{T-h} q_t q_t')^{-1}=O_p(1)$, and c) $T^{-\zeta-1}\sum_{t=1}^{T-h}q_t(\alpha_{j^*}(L)\varepsilon_t)=O_p(T^{-\zeta})$. \cref{thm:OLS_error} establishes these results.

The proofs of statements (\ref{itm:thm_x_iv_star}) and (\ref{itm:thm_x_iv}) are similar, so we focus on the latter. By definition of $\hat{x}_{h,t}$,  we have $\hat{x}_{h,t} - \varepsilon_{j^*,t} = (\vartheta-\hat{\vartheta}_h)'q_t + T^{-\zeta}\alpha_{j^*}(L)\varepsilon_t$. Let $\tilde{q}_{t} \equiv (\underline{\tilde{y}}_{j^*,t}',\tilde{y}_{t-1}')'$ and $\Delta_{t} \equiv q_t - \tilde{q}_{t}$. Then
\begin{align}
T^{-1}\sum_{t=1}^{T-h}(\hat{x}_{h,t}-\varepsilon_{j^*,t})\varepsilon_{t+\ell} &= (\vartheta-\hat{\vartheta}_h)' \left( \frac{1}{T} \sum_{t=1}^{T-h} \Delta_t \varepsilon_{t+\ell} \right) \label{eqn:aux1_cross_terms}\\
&\quad +  (\vartheta-\hat{\vartheta}_h)' \left( \frac{1}{T} \sum_{t=1}^{T-h} \tilde{q}_{t} \varepsilon_{t+\ell} \right) \label{eqn:aux2_cross_terms} \\
&\quad + \frac{1}{T^{\zeta}}  \left( \frac{1}{T} \sum_{t=1}^{T-h} \left( \alpha_{j^*}(L) \varepsilon_{t}  \right) \varepsilon_{t+\ell} \right) \label{eqn:aux3_cross_terms}.
\end{align}
\noindent By (\ref{itm:thm_x_i}), $(\vartheta-\hat{\vartheta}_h) = O_p(T^{-\zeta} + T^{-1/2})$. \cref{thm:ytilde}, \cref{asn:mds}, and Cauchy-Schwarz imply that the sample average in parenthesis in \eqref{eqn:aux1_cross_terms} is $O_{p}\left(  T^{-\zeta} \right)$. The two sample averages in parentheses in lines \eqref{eqn:aux2_cross_terms}--\eqref{eqn:aux3_cross_terms} have mean zero, since the shocks are white noise and mutually orthogonal, so \cref{thm:squared} and \citet[Thm.\ 7.1.1]{Brockwell1991} imply that they are each $O_p(T^{-1/2})$. It follows that
\[ T^{-1}\sum_{t=1}^{T-h}(\hat{x}_{h,t}-\varepsilon_{j^*,t})\varepsilon_{t+\ell} = O_p(T^{-2\zeta}) + o_p(T^{-1/2}).\]

For statement (\ref{itm:thm_x_iii}), note that
\[ T^{-1}\sum_{t=1}^{T-h}(\hat{x}_{h,t}-\varepsilon_{j^*,t})\hat{x}_{h,t}  = T^{-1}\sum_{t=1}^{T-h}(\hat{x}_{h,t}-\varepsilon_{j^*,t})^2 + T^{-1}\sum_{t=1}^{T-h}(\hat{x}_{h,t}-\varepsilon_{j^*,t})\varepsilon_{j^*,t}. \]
\cref{thm:x_minus_eps_2} shows that $T^{-1}\sum_{t=1}^{T-h}(\hat{x}_{h,t}-\varepsilon_{j^*,t})^2 = O_p(T^{-2\zeta}) + o_p(T^{-1/2})$. This result, combined with (\ref{itm:thm_x_iv_star}), implies that statement (\ref{itm:thm_x_iii}) holds. 

For statement (\ref{itm:thm_x_ii}), note that
\begin{align*}
T^{-1}\sum_{t=1}^{T-h}(\hat{x}_{h,t})^2 &= T^{-1}\sum_{t=1}^{T-h}(\hat{x}_{h,t}-\varepsilon_{j^*,t}+\varepsilon_{j^*,t} )^2  \\
&=  T^{-1}\sum_{t=1}^{T-h}(\hat{x}_{h,t}-\varepsilon_{j^*,t})^2 - 2  T^{-1}\sum_{t=1}^{T-h}(\hat{x}_{h,t}-\varepsilon_{j^*,t}) \varepsilon_{j^*,t}  +  T^{-1}\sum_{t=1}^{T-h} \varepsilon_{j^*,t}^2. 
\end{align*}
\cref{thm:x_minus_eps_2} and statement (\ref{itm:thm_x_iv_star}) imply that the first two terms converge in probability to zero. Since $T^{-1}\sum_{t=1}^{T-h} \varepsilon_{j^*,t}^2 \overset{p}{\to} \sigma_{j^*}^2$ by \cref{thm:squared} and \citet[Thm.\ 7.1.1]{Brockwell1991}, statement (\ref{itm:thm_x_ii}) holds.

Finally, statement (\ref{itm:thm_x_v}) obtains by decomposing
\begin{align*}
T^{-1}\sum_{t=1}^{T-h}B(L)\varepsilon_t(\hat{x}_{h,t}-\varepsilon_{j^*,t}) &= T^{-1}\sum_{t=1}^{T-h} B(L)\varepsilon_t q_t'(\vartheta-\hat{\vartheta}_h) +T^{-\zeta}T^{-1}\sum_{t=1}^{T-h} B(L)\varepsilon_t[\alpha_{j^*}(L)\varepsilon_t]' \\
&= O_p(1) \times O_p(T^{-\zeta}+T^{-1/2}) + T^{-\zeta} \times O_p(1),
\end{align*}
where the last line follows from statement (\ref{itm:thm_x_i}), \cref{thm:ytilde}, and \cref{thm:squared}.\end{proof}

\subsection{Auxiliary numerical lemma}

\begin{lem} \label{thm:impact}
Define $\overline{y}_{i,t} \equiv (y_{i+1,t},y_{i+2,t},\dots,y_{nt})'$ to be the (possibly empty) vector of variables that are ordered after $y_{i,t}$ in $y_t$. Partition
\[\hat{\Sigma}=\begin{pmatrix}
	\hat{\Sigma}_{11} & \hat{\Sigma}_{12} & \hat{\Sigma}_{13} \\
	\hat{\Sigma}_{21} & \hat{\Sigma}_{22} & \hat{\Sigma}_{23} \\
	\hat{\Sigma}_{31} & \hat{\Sigma}_{32} & \hat{\Sigma}_{33}
\end{pmatrix},\quad \hat{C}=\begin{pmatrix}
	\hat{C}_{11} & 0 & 0 \\
	\hat{C}_{21} & \hat{C}_{22} & 0 \\
	\hat{C}_{31} & \hat{C}_{32} & \hat{C}_{33}
\end{pmatrix},\]
conformably with $y_t=(\underline{y}_{j^*,t}',y_{j^*,t},\overline{y}_{j^*,t}')'$, where $\hat{\Sigma}=\hat{C}\hat{C}'$ (in particular, $\hat{C}_{22}=\hat{C}_{j^*,j^*}$). Then
\begin{equation} \label{eqn:cholesk}
(\hat{\Sigma}_{31}, \hat{\Sigma}_{32})\begin{pmatrix}
	\hat{\Sigma}_{11} & \hat{\Sigma}_{12} \\
	\hat{\Sigma}_{21} & \hat{\Sigma}_{22}
\end{pmatrix}^{-1}e_{j^*,j^*} = \hat{C}_{22}^{-1}\hat{C}_{32}.
\end{equation}
\end{lem}
Note that the lemma implies $\hat{\beta}_0=\hat{\delta}_0$: If $i^* < j^*$ or $i^*=j^*$, then both estimators equal 0 or 1 (by definition), respectively; if $i^*>j^*$, then $\hat{\beta}_0$ is defined as the $i^*-j^*$ element of the left-hand side of \eqref{eqn:cholesk} (by Frisch-Waugh), while $\hat{\delta}_0$ is defined as the $i^*-j^*$ element of the right-hand side of \eqref{eqn:cholesk}.
\begin{proof}
From the relationship $\hat{\Sigma}=\hat{C}\hat{C}'$, we get
\[\begin{pmatrix}
	\hat{\Sigma}_{11} & \hat{\Sigma}_{12} \\
	\hat{\Sigma}_{21} & \hat{\Sigma}_{22} \\
	\hat{\Sigma}_{31} & \hat{\Sigma}_{32}
\end{pmatrix} = \begin{pmatrix}
\hat{C}_{11}\hat{C}_{11}' & \hat{C}_{11}\hat{C}_{21}' \\
\hat{C}_{21}\hat{C}_{11}' & \hat{C}_{21}\hat{C}_{21}' + \hat{C}_{22}^2 \\
\hat{C}_{31}\hat{C}_{11}' & \hat{C}_{31}\hat{C}_{21}' + \hat{C}_{32}\hat{C}_{22}
\end{pmatrix}.\]
The partitioned inverse formula implies
\begin{align*}
\begin{pmatrix}
	\hat{\Sigma}_{11} & \hat{\Sigma}_{12} \\
	\hat{\Sigma}_{21} & \hat{\Sigma}_{22}
\end{pmatrix}^{-1}e_{j^*,j^*} &= \frac{1}{\hat{C}_{21}\hat{C}_{21}' + \hat{C}_{22}^2-\hat{C}_{21}\hat{C}_{11}'(\hat{C}_{11}\hat{C}_{11}')^{-1}\hat{C}_{11}\hat{C}_{21}'}\begin{pmatrix}
	-(\hat{C}_{11}\hat{C}_{11}')^{-1}\hat{C}_{11}\hat{C}_{21}' \\
	1
\end{pmatrix} \\
&= \frac{1}{\hat{C}_{22}^2}\begin{pmatrix}
	-\hat{C}_{11}^{-1\prime}\hat{C}_{21}' \\
	1
\end{pmatrix},
\end{align*}
so
\[(\hat{\Sigma}_{31}, \hat{\Sigma}_{32})\begin{pmatrix}
	\hat{\Sigma}_{11} & \hat{\Sigma}_{12} \\
	\hat{\Sigma}_{21} & \hat{\Sigma}_{22}
\end{pmatrix}^{-1}e_{j^*,j^*} = \frac{1}{\hat{C}_{22}^2}\left(-\hat{C}_{31}\hat{C}_{11}'\hat{C}_{11}^{-1\prime}\hat{C}_{21}'+\hat{C}_{31}\hat{C}_{21}' + \hat{C}_{32}\hat{C}_{22}\right) = \frac{1}{\hat{C}_{22}}\hat{C}_{32}. \qedhere\]
\end{proof}

\subsection{Auxiliary asymptotic lemmas}

\begin{lem} \label{thm:ytilde}
	$T^{-1}\sum_{t=1}^T \|y_t-\tilde{y}_t\|^2 = O_p(T^{-2\zeta})$ and $T^{-1}\sum_{t=1}^T u_t(y_{t-1}-\tilde{y}_{t-1})' = O_p(T^{-2\zeta}+T^{-\zeta-1/2})$, where $u_t \equiv y_t - A y_{t-1}$.
\end{lem}

\begin{proof}
Using Equation (2.1) and the definition $\tilde{y}_t \equiv (I_n-AL)^{-1}H\varepsilon_t$, we have
\[y_t-\tilde{y}_t = T^{-\zeta} (I_n-AL)^{-1}H\alpha(L) \varepsilon_t.\]
The lag polynomial $(I_n-AL)^{-1}H\alpha(L)$ is absolutely summable by virtue of being a product of absolutely summable polynomials, see Assumption 2.1(ii) and (v).

\citet[Prop.\ 3.1.1]{Brockwell1991} implies $E[T^{2\zeta}\|y_t-\tilde{y}_t\|^2] < \infty$. The first statement of the lemma then follows from Markov's inequality.

In order to establish the second part of the lemma, note that
\[ \frac{1}{T}\sum_{t=1}^T u_{t} \left( y_{t-1} - \tilde{y}_{t-1}  \right)' = H\left(\frac{1}{T}\sum_{t=1}^T \varepsilon_{t}  \left( y_{t-1} - \tilde{y}_{t-1}  \right)'\right) + T^{-\zeta}H\left(\frac{1}{T}\sum_{t=1}^T \alpha(L) \varepsilon_{t}  \left( y_{t-1} - \tilde{y}_{t-1}  \right)'\right).   \]	
The product process $T^{\zeta}\varepsilon_t \otimes (y_{t-1}-\tilde{y}_{t-1})$ is a martingale difference sequence with finite variance by \cref{thm:squared} (note that this is a standard stochastic process and not a triangular array). Hence, the first sample average in parenthesis on the right-hand side above is $O_p(T^{-\zeta-1/2})$ by Chebyshev's inequality. The second sample average in parenthesis on the right-hand side above is  $O_p(T^{-\zeta})$ by \cref{thm:squared}. Hence, the entire right-hand side in the above display is $O_p(T^{-\zeta-1/2} + T^{-2\zeta})$, as claimed.
\end{proof}

\begin{lem} \label{thm:uy}
	\[T^{-1}\sum_{t=1}^T u_t y_{t-1}' = T^{-\zeta}H\sum_{\ell=1}^{\infty} \alpha_{\ell}DH'(A')^{\ell-1} + T^{-1}\sum_{t=1}^T H\varepsilon_t \tilde{y}_{t-1}' + o_p(T^{-\zeta}).\]
\end{lem}
\begin{proof}
	\begin{align*}
		T^{-1}\sum_{t=1}^T u_t y_{t-1}' &= T^{-1}\sum_{t=1}^T u_t \tilde{y}_{t-1}' + \underbrace{T^{-1}\sum_{t=1}^T u_t (y_{t-1}-\tilde{y}_{t-1})'}_{=o_p(T^{-\zeta}) \text{ by \cref{thm:ytilde}}} \\
		&= T^{-1}\sum_{t=1}^T H\varepsilon_t \tilde{y}_{t-1}' + T^{-\zeta-1}\sum_{t=1}^T H\alpha(L)\varepsilon_t \tilde{y}_{t-1}' + o_p(T^{-\zeta}) \\
		&= T^{-1}\sum_{t=1}^T H\varepsilon_t \tilde{y}_{t-1}' + T^{-\zeta}H\left( E[\alpha(L)\varepsilon_t \tilde{y}_{t-1}'] + o_p(1)\right) + o_p(T^{-\zeta}).
	\end{align*}
	The last line invokes a law of large numbers, which applies because $\tilde{y}_t$ is an absolutely summable linear filter of the shocks, so we can apply \cref{thm:squared} in conjunction with \citet[Thm.\ 7.1.1]{Brockwell1991} and Chebyshev's inequality. Finally, note that
	\[E[\alpha(L)\varepsilon_t \tilde{y}_{t-1}'] = \sum_{\ell=1}^\infty \sum_{s=0}^{\infty} \alpha_\ell  E[\varepsilon_{t-\ell}\varepsilon_{t-s-1}']H'(A')^{s} = \sum_{\ell=1}^{\infty} \alpha_{\ell}DH'(A')^{\ell-1}. \qedhere\]
\end{proof}

\begin{lem} \label{thm:y_second_moment}
	$T^{-1}\sum_{t=1}^T y_{t-1}y_{t-1}' \overset{p}{\to} S$.
\end{lem}
\begin{proof}
	By \cref{thm:ytilde} and Cauchy-Schwarz, $T^{-1}\sum_{t=1}^T y_{t-1}y_{t-1}'=T^{-1}\sum_{t=1}^T \tilde{y}_{t-1}\tilde{y}_{t-1}' + o_p(1)$. \cref{thm:squared}, \citet[Thm.\ 7.1.1]{Brockwell1991}, and Chebyshev's inequality imply that the law of large numbers holds for $\lbrace \tilde{y}_{t-1}\tilde{y}_{t-1}' \rbrace$.
\end{proof}

\begin{lem} \label{thm:OLS_error}
Omitting the subscript $j^*$ in a slight abuse of notation, let $q_{t} \equiv (\underline{y}_{j^*,t}',y_{t-1}')'$. Then

\begin{enumerate}[i)]
\item \label{itm:OLS_error_i} $T^{-1}\sum_{t=1}^{T-h}q_t \varepsilon_{j^*,t}=O_p(T^{-\zeta} + T^{-1/2})$,
\item \label{itm:OLS_error_ii} $(T^{-1}\sum_{t=1}^{T-h} q_t q_t')^{-1}=O_p(1)$,
\item \label{itm:OLS_error_iii} $T^{-1}\sum_{t=1}^{T-h}q_t(\alpha_{j^*}(L)\varepsilon_t)=O_p(1)$,
\end{enumerate}
where $\alpha_{j^*}(L)$ is the $j^*$-th row of $\alpha(L)$.
\end{lem}

\begin{proof}
Let $\tilde{q}_{t} \equiv (\underline{\tilde{y}}_{j^*,t}',\tilde{y}_{t-1}')'$ and $\Delta_{t} \equiv q_t - \tilde{q}_{t}$. Note that 
\begin{equation} \label{eqn:aux1_OLS_error}
T^{-1}\sum_{t=1}^{T-h}q_t \varepsilon_{j^*,t}= T^{-1}\sum_{t=1}^{T-h} \Delta_{t} \varepsilon_{j^*,t} + T^{-1}\sum_{t=1}^{T-h} \tilde{q}_{t} \varepsilon_{j^*,t}.    
\end{equation}
Cauchy-Schwarz implies
\[ \left \| T^{-1}\sum_{t=1}^{T-h} \Delta_{t} \varepsilon_{j^*,t} \right \| \leq \left(  \frac{1}{T} \sum_{t=1}^{T-h} \|  \Delta_t \|^2    \right)^{1/2} \left( \frac{1}{T} \sum_{t=1}^{T-h} \varepsilon_{j^*,t}^2  \right)^{1/2} = O_p(T^{-\zeta}) \times O_p(1),   \]
using \cref{thm:ytilde} and \cref{asn:mds}. The summands in the last sample average in \eqref{eqn:aux1_OLS_error} have mean zero due to shock orthogonality, so this sample average is $O_{p}\left(  T^{-1/2} \right)$ by \cref{thm:squared} and \citet[Thm.\ 7.1.1]{Brockwell1991}. This establishes part (\ref{itm:OLS_error_i}) of the lemma.

For part (\ref{itm:OLS_error_ii}) of the lemma, note that
\begin{equation} \label{eqn:aux2_OLS_error}
\frac{1}{T} \sum_{t=1}^{T-h} q_{t} q'_{t} = \frac{1}{T} \sum_{t=1}^{T-h} \Delta_{t} \Delta'_{t} + \frac{1}{T} \sum_{t=1}^{T-h} \tilde{q}_{t} \Delta'_{t} + \frac{1}{T} \sum_{t=1}^{T-h} \Delta_{t}\tilde{q}_{t}' + \frac{1}{T} \sum_{t=1}^{T-h} \tilde{q}_{t} \tilde{q}'_{t}. 
\end{equation} 
\cref{thm:ytilde} implies that the first term is $O_{p} \left( T^{-2\zeta} \right)$. Cauchy-Schwarz, along with \cref{thm:ytilde,thm:y_second_moment}, imply that the second and third terms are $O_{p}(T^{-\zeta})$. The last term converges in probability to $\var(\tilde{q}_t)$, as in the proof of \cref{thm:y_second_moment}. This matrix is non-singular, since $\tilde{q}_{t} = (\underline{\tilde{y}}_{j^*,t}',\tilde{y}_{t-1}')'$, where $\var(\tilde{y}_{t-1})=S$ is non-singular by Assumption 2.1(iv), and Assumption 2.1(iii) implies that $\underline{\tilde{y}}_{j^*,t}$ equals a linear transformation of $\tilde{y}_{t-1}$ plus a non-singular orthogonal noise term.

Part (\ref{itm:OLS_error_iii}) follows from Cauchy-Schwarz, \cref{thm:y_second_moment}, and Assumption 2.1(v).
\end{proof}

\begin{lem} \label{thm:x_minus_eps_2}
Use the same notation as \cref{thm:OLS_error}, and let
\[  \hat{x}_{h,t}  \equiv  (\vartheta-\hat{\vartheta}_h)'q_t + \varepsilon_{j^*,t} + T^{-\zeta}\alpha_{j^*}(L)\varepsilon_t.\]
Then
\begin{equation} \label{eqn:aux1_x_minus_eps_2}
T^{-1}\sum_{t=1}^{T-h}(\hat{x}_{h,t}-\varepsilon_{j^*,t})^2 = O_p(T^{-2\zeta}) + o_p(T^{-1/2}).
\end{equation}
\end{lem}

\begin{proof}
It suffices by the $c_r$-inequality to show that
\begin{enumerate}[a)]
\item \label{itm:x_minus_eps_2_a} $T^{-1}\sum_{t=1}^{T-h} \left( (\vartheta-\hat{\vartheta}_h)'q_t \right)^2 = O_p(T^{-2\zeta}) + o_{p} \left(  T^{-1/2} \right)$,
\item \label{itm:x_minus_eps_2_b} $T^{-1}\sum_{t=1}^{T-h} \left( \alpha_{j^*}(L)\varepsilon_t \right)^2 = O_{p} \left( 1 \right)$.
\end{enumerate}
To establish (\ref{itm:x_minus_eps_2_a}), note that Cauchy-Schwarz implies
\[\frac{1}{T} \sum_{t=1}^{T-h} \left( (\vartheta-\hat{\vartheta}_h)'q_t \right)^2 \leq \left \| \vartheta-\hat{\vartheta}_h \right \|^2  \left( \frac{1}{T} \sum_{t=1}^{T-h}   \| q_t \| ^2 \right) = O_{p} \left( \left(T^{-\zeta} + T^{-1/2}\right)^2 \right) \times O_p(1),\]
by \cref{thm:y_second_moment,thm:OLS_error}. Hence, the right-hand side is $O_p(T^{-2\zeta}) + o_{p}(  T^{-1/2})$.

Statement (\ref{itm:x_minus_eps_2_b}) follows from Assumption 2.1(v) and Markov's inequality. 
\end{proof}

\begin{lem} \label{thm:squared}
Let $\psi(L)$ and $\varphi(L)$ be two absolutely summable, univariate, two-sided lag polynomials. Then for any $j,k \in \lbrace 1,\dots,m\rbrace$, the product process $z_t \equiv [\psi(L)\varepsilon_{j,t}] \times [\varphi(L)\varepsilon_{k,t}]$ has absolutely summable autocovariance function.
\end{lem}
\begin{proof}
Bound $\sum_{\ell=-\infty}^\infty |\cov(z_t,z_{t+\ell})|$ by
\begin{align}
&\sum_{\ell=-\infty}^\infty \sum_{\tau_1=-\infty}^\infty \sum_{\tau_2=-\infty}^\infty \sum_{\tau_3=-\infty}^\infty \sum_{\tau_4=-\infty}^\infty |\psi_{\tau_1}||\varphi_{\tau_2}||\psi_{\tau_3}||\varphi_{\tau_4}| |\cov(\varepsilon_{j,t+\tau_1}\varepsilon_{k,t+\tau_2},\varepsilon_{j,t+\tau_3+\ell}\varepsilon_{k,t+\tau_4+\ell})| \nonumber \\
&\leq  \left(\sum_{\tau=-\infty}^\infty |\psi_{\tau}|\right)^2 \left(\sum_{\tau=-\infty}^\infty |\varphi_{\tau}|\right)^2 \left(\sup_{\tau_1,\dots,\tau_4} \sum_{\ell=-\infty}^\infty |\cov(\varepsilon_{j,t+\tau_1}\varepsilon_{k,t+\tau_2},\varepsilon_{j,t+\tau_3+\ell}\varepsilon_{k,t+\tau_4+\ell})|\right). \label{eqn:squared_sup}
\end{align}
To finish the proof, we show that the supremum above is finite. By stationarity, the sum inside the supremum can be written as
\begin{equation} \label{eqn:squared_sum}
\sum_{r=-\infty}^\infty |\cov(\varepsilon_{j,0}\varepsilon_{k,s},\varepsilon_{j,r}\varepsilon_{k,r+\tau})|,
\end{equation}
where we have substituted $s=\tau_2-\tau_1$, $\tau=\tau_4-\tau_3$, and $r = \ell+\tau_3-\tau_1$. Now fix $s$ and $\tau$. Let $N \colon \mathbb{Z} \to \lbrace 1,2,3,4 \rbrace$ denote the function that assigns each integer $r$ to the number of times the maximum value of the tuple $(0,s,r,r+\tau)$ appears in the tuple. We group the terms indexed by $r$ in the sum $\eqref{eqn:squared_sum}$ according to their value of $N(r)$. First, since $\lbrace \varepsilon_t \rbrace$ is a martingale difference sequence, all terms $r$ with $N(r)=1$ yield a covariance of 0 and so do not contribute to the sum. Second, a simple enumeration of cases shows that there is at most 1 value of $r$ with $N(r)=3$ and at most one with $N(r)=4$. Finally, consider terms $r$ with $N(r)=2$. If $s \neq 0$ and $\tau \neq 0$, then there are at most 4 such terms (since this requires $r \in \lbrace 0,s,-\tau,s-\tau\rbrace$). If $s=0$ and/or $\tau=0$, terms with $N(r)=2$ must be of the form $|\cov(\varepsilon_{j,0}\varepsilon_{k,0},\varepsilon_{j,r}\varepsilon_{k,r+\tau})|$ (with $r,r+\tau<0$) and/or $|\cov(\varepsilon_{j,0}\varepsilon_{k,0},\varepsilon_{j,-r}\varepsilon_{k,s-r})|$ (with $r,r-s>0$). The preceding arguments show that the sum \eqref{eqn:squared_sum} is bounded by
\[6E[\|\varepsilon_t\|^4] + 2\sum_{r=1}^\infty \sum_{b=1}^\infty \|\cov(\varepsilon_0 \otimes \varepsilon_0, \varepsilon_{-r} \otimes \varepsilon_{-b})\|,\]
which is finite by \cref{asn:mds} and does not depend on $s$ or $\tau$. Thus, the supremum in \eqref{eqn:squared_sup} is finite.
\end{proof}

\end{appendix}

%


\bibliographystyle{ecta-fullname} 
\bibliography{ref}  

